\documentclass[a4paper, 11pt, dvipsnames]{article}
\usepackage[utf8]{inputenc}
\usepackage[T1]{fontenc}
\usepackage[english]{babel}
\usepackage{fullpage}

\usepackage{graphicx}
\usepackage[labelfont=sc]{caption}
\usepackage{todonotes} \usepackage{tabularx} \usepackage{multirow} \usepackage{enumitem} \setitemize[0]{label=$\bullet$}
\usepackage{comment} \usepackage[numbers]{natbib} \usepackage{xspace} \usepackage{xargs} \usepackage{setspace} 

\usepackage{algorithm}
\usepackage[]{algpseudocode}
  \algrenewcommand\algorithmicrequire{\textbf{Input:}}
  \algrenewcommand\algorithmicensure{\textbf{Output:}}
\renewcommand{\Comment}[2][.58\linewidth]{\leavevmode\hfill\makebox[#1][l]{\scalebox{0.9}{$/\!\!/$\,#2}}}
  \algnewcommand\algorithmiccommentass{$\dot$}
  \algnewcommand\CommentAss{\hfill $\clubsuit$ \,}
  
\usepackage{mathtools}
\usepackage{amssymb}
\usepackage{stmaryrd} \usepackage{amsthm} \usepackage{bm} \usepackage[normalem]{ulem} \usepackage{mathrsfs}  \usepackage{tikz} \usetikzlibrary{cd} \newcommand{\mypos}[2]{\tikz[remember picture]
  {\node[inner sep=0pt,anchor=base](#2){#1};}} 

\usepackage{array}   \newcolumntype{C}{>{$}c<{$}}

\usepackage{color}
\def\rev#1{{#1}}
\def\newrev#1{{#1}}

\newcommand{\polun}{\theta}
\newcommand{\polde}{\xi}
\newcommand{\bpolun}{\bm{\polun}}
\newcommand{\bpolde}{\bm{\polde}}

\renewcommand{\tilde}[1]{\widetilde{#1}}

\newcommand{\pscal}[1]{\left\langle #1 \right\rangle}

\newcommandx{\XXi}[1][1=n]{x_1,\dotsc,x_{#1}}
\newcommandx{\EEi}[1][1=i]{e_1,\dotsc,e_{#1}}
\newcommandx{\ZZi}[1][1=r]{z_1,\dotsc,z_{#1}}
\newcommandx{\Lindet}[1][1=c+i]{L_1,\dotsc,L_{#1}}
\newcommandx{\Tindet}[1][1=c+i]{T_1,\dotsc,T_{#1}}
\newcommand{\nrec}{m}
\newcommand{\Nrec}{M}
\newcommand{\Xrec}{\mathfrak{x}}
\newcommand{\XXreci}[1][\nrec]{\Xrec_1,\dotsc,\Xrec_{#1}}
\newcommand{\XXrec}{\mathfrak{X}}

\newcommand{\NNo}{\mathbb{N}}
\newcommand{\KK}{\mathbf{K}}
\newcommand{\CC}{\mathbf{C}}
\newcommand{\RR}{\mathbf{R}}
\newcommand{\QQ}{\mathbf{Q}}
\newcommand{\CCo}{\mathbb{C}}
\newcommand{\RRo}{\mathbb{R}}
\newcommand{\QQo}{\mathbb{Q}}

\newcommand{\CCgen}{\Kfrak}
\newcommand{\Kclos}{\Kfrak}

\newcommand{\VKclos}{V_{\Kclos}}

\newcommand{\Iminor}{\mfrak}
\newcommand{\IPminor}{\pfrak}

\newcommand{\vphiWvphii}{
\bvphi_i(\Wphii[i][\VKclos][\bvphi])
}

\newcommandx{\ybar}[1][1=]{[y_{#1}]}
\newcommandx{\xbar}[1][1=]{[x_{#1}]}

\newcommand{\taubar}{[\tau]}

\newcommand{\Inci}[1][\bphi]{\Psi_{#1}}
\newcommand{\Inciphi}[1][e]{\Psi_{{\map[#1]}}}
\newcommand{\incid}[1]{\widetilde{#1}}
\newcommand{\Gammat}{\incid{\Gamma}}
\newcommand{\scrPt}{\incid{\mathscr{P}}}
\newcommand{\scrSt}{\incid{\mathscr{S}}}
\newcommand{\scrQt}{\incid{\mathscr{Q}}}
\newcommand{\scrKt}{\incid{\mathscr{K}}}
\newcommand{\scrWt}{\incid{\mathscr{W}}}
\newcommand{\scrFt}{\incid{\mathscr{F}}}

\newcommand{\scrRFQt}{\incid{\mathscr{R}_{F_Q}}}
\newcommand{\scrWunt}{\incid{\mathscr{W}_1}}
\newcommand{\Pcalt}{\incid{\mathcal{P}}}
\newcommand{\Lt}{\incid{L}}
\newcommand{\LtF}{\incid{L_F}}
\newcommand{\LtW}{\incid{L_W}}

\newcommand{\Vt}{\incid{V}}
\newcommand{\St}{\incid{S}}
\newcommand{\Qt}{\incid{Q}}
\newcommand{\Wt}{\incid{W}}

\newcommand{\FQt}{\incid{F_Q}}
\newcommand{\SQt}{\incid{S_Q}}

\newcommand{\Gammaphi}{\Gamma^{\bphi}}
\newcommand{\Gammalpha}{\Gamma^{\alpha}}
\newcommand{\Vphi}{V^{\bphi}}
\newcommand{\Sphi}{S^{\bphi}}
\newcommand{\ffphi}{\ff^{\bphi}}
\newcommand{\hhphi}{\hh^{\bphi}}
\newcommand{\ggphi}[1][]{\gg^{\map[#1]}}
\newcommand{\yyphi}{\bm{y}^{\bphi}}
\newcommand{\bchiphi}{\bchi^{\bphi}}
\newcommand{\chiphi}{\chi^{\bphi}}

\newcommand{\scrWphi}{\scrW^{\bphi}}
\newcommand{\scrFphi}{\scrF^{\bphi}}

\newcommand{\Vo}{V^\circ}
\newcommand{\Voreg}{\V^\circ_{\text{reg}}}
\newcommand{\Vreg}{\V_{\text{reg}}}
\newcommand{\Zo}{Z^{\circ}}

\newcommandx{\Fo}[1][1=\yy]{F^{\circ}_{#1}}

\newcommandx{\proj}[3][1=,2=,3=\pi]{\bm{#3}_{#1}^{#2}}
\newcommandx{\map}[3][1=,2=,3=\phi]{\bm{#3}_{#1}^{#2}}
\newcommand{\mapun}{\map[1]}
\newcommand{\mapunrec}{\map[1][-1]}
\newcommandx{\phia}[2][1=\phi,2=\aa]{{#1}^{+#2}}

\newcommandx{\mapfbr}[3][2=Q,3=e]{#1_{\mid \map[#3] \in \,#2}}
\newcommandx{\mapfbreq}[3][2=\yy,3=e]{#1_{\mid \map[#3] = \,#2}}

\newcommandx{\projfbr}[3][2=Q,3=e]{#1_{\mid \proj[#3] \in \,#2}}
\newcommandx{\projfbreq}[3][2=\yy,3=e]{#1_{\mid \proj[#3] = \,#2}}

\newcommandx{\phibr}[2][1=Q, 2=e]{\map[#2][-1](#1)}

\newcommand{\Wo}{W^\circ}

\newcommandx{\Wophii}[3][1=i,2=V,3=\bphi]{W^\circ_{#3}(#1,#2)}
\newcommandx{\Wphii}[3][1=i,2=V,3=\bphi]{W_{#3}(#1,#2)}
\newcommandx{\Kphii}[3][1=i,2=V,3=\bphi]{K_{#3}(#1,#2)}
\newcommandx{\WWophii}[4][1=1,2=i,3=V,4=\bphi]{W^\circ_{#4}
(#1,W_{#4}(#2,#3))}
\newcommandx{\WWphii}[4][1=1,2=i,3=V,4=\bphi]{W_{#4}
(#1,W_{#4}(#2,#3))}
\newcommandx{\KWphii}[4][1=1,2=i,3=V,4=\bphi]{K_{#4}
(#1,W_{#4}(#2,#3))}
\newcommandx{\Woproji}[2][1=i,2=\Vphi]{W^\circ(\proj[#1],#2)}
\newcommandx{\Wproji}[2][1=i,2=\Vphi]{W(\proj[#1],#2)}
\newcommandx{\Kproji}[2][1=i,2=\Vphi]{K(\proj[#1],#2)}
\newcommandx{\Hphi}[1][1=\bphi]{\mathcal{H}_{#1}}
\newcommandx{\Wovphii}[3][1=i,2=V,3=\bvphi]{W^\circ_{#3}(#1,#2)}
\newcommandx{\Wvphii}[3][1=i,2=V,3=\bvphi]{W_{#3}(#1,#2)}

\newcommand{\Wchart}{W_{\mathrm{chart}}}
\newcommand{\Watlas}{W_{\mathrm{atlas}}}
\newcommand{\Wlag}{W_{\mathrm{lag}}}
\newcommand{\Fchart}{F_{\mathrm{chart}}}
\newcommand{\Fatlas}{F_\mathrm{atlas}}
\newcommand{\Flag}{F_{\mathrm{lag}}}

\newcommand{\minor}{\mfrak}
\newcommand{\Eminor}{\Mfrak}
\newcommandx{\Ji}[1][1=i]{J_{#1}}
\newcommandx{\Jm}[1][1=\minor]{J_{#1}}
\newcommandx{\gm}[1][1=\minor]{\delta_{#1}}
\newcommandx{\bgm}[1][1=\minor]{\bm{\delta}_{#1}}
\newcommandx{\subi}[1][1=i]{\mathsf{Sub}_{#1}}
\newcommandx{\Rm}[1][1=\minor]{\widetilde{\mathsf{R}}_{#1}}
\newcommandx{\Em}[1][1=\minor]{\mathfrak{E}_{#1}}
\newcommandx{\psim}[1][1=\minor]{\bm{\psi}_{#1}}
\newcommand{\MRk}{M}
\newcommandx{\JMRi}[2][1=i,2=\MRk]{J_{#1}^{#2}}

\newcommand{\SA}{semi-algebraic\xspace}
\newcommand{\SAC}{semi-algebraically connected\xspace}
\newcommand{\SACC}{semi-al\-geb\-rai\-cal\-ly con\-nected component\xspace}
\newcommand{\SACCs}{semi-al\-geb\-rai\-cal\-ly con\-nec\-ted components\xspace}
\newcommand{\SLP}{straight-line program\xspace}
\newcommand{\SLPs}{straight-line programs\xspace}
\newcommand{\opsf}{\textsf{op}}
\newcommand{\ZDP}{zero-dimensional parametrization\xspace}
\newcommand{\ZDPs}{zero-dimensional parametrizations\xspace}

\newcommand{\ZDPsQ}{\ZDPs with coefficients in $\QQ$\xspace}
\newcommand{\ODP}{one-dimensional parametrization\xspace}

\newcommand{\wcoeffQ}{with coefficients in $\QQ$\xspace}
\newcommand{\NEZO}{non-empty Zariski open\xspace}
\newcommand{\LCS}{locally closed set\xspace}
\newcommand{\RRS}{reduced regular sequence\xspace}

\newcommandx{\ULbar}[1][1=L]{\overline{\scrU(#1)}}
\newcommandx{\ZUL}[1][1=L]{\Zclosure{\scrU(#1)}}
\newcommandx{\piX}[1][1=]{\pi_{\XX}^{#1}}
\newcommand{\mLNF}{\mfrak}
\newcommand{\dLNF}{\dfrak}
\newcommand{\LNF}{local normal form\xspace}
\newcommand{\GNF}{global normal form\xspace}

\newcommand{\GNFP}{global normal form property\xspace}

\newcommand{\lagrange}{\normalfont{\textsf{Lagrange}}\xspace}
\newcommand{\SingPts}{\textsf{SingularPoints}\xspace}
\newcommand{\SolveLag}{\textsf{SolveLagrange}\xspace}
\newcommand{\SolvePolar}{\textsf{SolvePolar}\xspace}

\newcommand{\Rand}{\textsf{Random}\xspace}
\newcommand{\IncVar}{\textsf{IncSLP}\xspace}
\newcommand{\IncParam}{\textsf{IncParam}\xspace}
\newcommand{\Wun}{\textsf{W\textsubscript{1}}\xspace}
\newcommand{\CritPolar}{\textsf{CritPolar}\xspace}
\newcommand{\Union}{\textsf{Union}\xspace}
\newcommand{\Proj}{\textsf{Projection}\xspace}

\newcommand{\Image}{\textsf{Image}\xspace}
\newcommand{\Fiber}{\textsf{Fiber}\xspace}
\newcommand{\FiberPolar}{\textsf{FiberPolar}\xspace}

\newcommand{\RMBound}{\textsf{RoadmapBounded}\xspace}
\newcommand{\RMRecLag}{\textsf{RoadmapRecLagrange}\xspace}
\newcommand{\WLag}{W_{\textup{Lagrange}}\xspace}
\newcommand{\fail}{\textsf{fail}\xspace}
\newcommand{\SLPPhi}{\textsf{PhiGen}\xspace}
\newcommand{\Crit}{\textsf{Crit}\xspace}

\newcommand{\Otilde}{O\;\tilde{}\:}
\newcommand{\softOh}[1]{\Otilde\left(#1\right)}
\newcommand{\Oh}[1]{O\left(#1\right)}

\newcommand{\inOh}{\;=\;}
\newcommand{\deltaK}{\delta_{\scrK}}
\newcommand{\deltaP}{\delta_{\scrP}}
\newcommand{\deltaS}{\sigma_{\scrS}}
\newcommand{\deltaPF}{\mu_{\scrP_F}}
\newcommand{\deltaSF}{\sigma_{\scrS_F}}

\newcommand{\et}{\quad \text{and} \quad}
\newcommand{\V}{\bm{V}}
\newcommand{\I}{\bm{I}}
\newcommand{\OO}{\bm{0}}
\newcommand{\Zparam}{\mathsf{Z}}

\newcommand{\Zclosure}[1]{\overline{#1}^{\scalebox{0.5}{$Z$}}}
\newcommand{\setenumi}[1]{\setcounter{enumi}{#1}}

\DeclareMathOperator{\sing}{sing}
\DeclareMathOperator{\rank}{rank}
\DeclareMathOperator{\jac}{Jac}
\DeclareMathOperator{\reg}{reg}

\newcommand{\Tg}{T}

\theoremstyle{plain}
\newtheorem{theorem}{Theorem}[section]
\newtheorem{lemma}[theorem]{Lemma}
\newtheorem{proposition}[theorem]{Proposition}
\newtheorem*{proposition*}{Proposition}

\newtheorem{corollary}[theorem]{Corollary}
\newtheorem{definition}[theorem]{Definition}
\theoremstyle{definition}
\newtheorem{example}{Example}

\theoremstyle{remark}
\newtheorem{remark}[theorem]{Remark}

\newenvironment{myproof}[1]{\paragraph{Proof of {#1}}}{\hfill$\square$}

\newcommand{\ifi}{\mathfrak{r}}

\let\vphi\phi
\renewcommand{\phi}{\varphi}
\newcommand{\eps}{\varepsilon}
\newcommand{\vtheta}{\vartheta}

\newcommand{\balpha}{\bm{\alpha}}
\newcommand{\alf}{\alpha}
\newcommand{\balf}{\bm{\alpha}}

\newcommand{\blambda}{\bm{\lambda}}

\newcommand{\bvtheta}{\bm{\vartheta}}
\newcommand{\bstar}{\bm{*}}
\newcommand{\bphi}{\bm{\phi}}
\newcommand{\bpsi}{\bm{\psi}}
\newcommand{\bchi}{\bm{\chi}}
\newcommand{\brho}{\bm{\rho}}
\newcommand{\btau}{\bm{\tau}}
\newcommand{\bdelta}{\bm{\delta}}
\newcommand{\bvphi}{\bm{\vphi}}
\newcommand{\bpi}{\bm{\pi}}

\renewcommand{\AA}{\mathbf{A}}

\newcommand{\EE}{\mathbf{E}}

\newcommand{\LL}{\bm{L}}
\newcommand{\FF}{\bm{F}}

\newcommand{\HH}{\bm{H}}

\newcommand{\XX}{\mathbf{X}}
\newcommand{\YY}{\mathbf{Y}}
\newcommand{\ZZ}{\mathbf{Z}}
\renewcommand{\aa}{\bm{a}}
\newcommand{\bb}{\bm{b}}

\newcommand{\ff}{{\bm{f}}}
\renewcommand{\gg}{\bm{g}}
\newcommand{\hh}{{\bm{h}}}

\newcommand{\bell}{\bm{\ell}}
\newcommand{\mm}{\bm{m}}
\newcommand{\nn}{\bm{n}}
\newcommand{\pp}{\bm{p}}
\renewcommand{\tt}{\bm{t}}
\newcommand{\uu}{\bm{u}}
\newcommand{\vv}{\bm{v}}
\newcommand{\ww}{\bm{w}}

\newcommand{\yy}{{\bm{y}}}
\newcommand{\zz}{{\bm{z}}}
\newcommand{\sfA}{\mathsf{A}}
\newcommand{\sfB}{\mathsf{B}}
\newcommand{\sfC}{\mathsf{C}}
\newcommand{\sfF}{\mathsf{F}}
\newcommand{\sfG}{\mathsf{G}}
\newcommandx{\sfH}[1][1=]{\mathsf{H_{#1}}}
\newcommand{\sfI}{\mathsf{I}}
\newcommand{\sfL}{\mathsf{L}}
\newcommand{\sfK}{\mathsf{K}}
\newcommand{\sfP}{\mathsf{P}}
\newcommand{\sfR}{\mathsf{R}}
\newcommand{\sfW}{\mathsf{W}}

\newcommand{\scrA}{\mathscr{A}}
\newcommand{\scrB}{\mathscr{B}}

\newcommand{\scrD}{\mathscr{D}}
\newcommand{\scrE}{\mathscr{E}}
\newcommand{\scrF}{\mathscr{F}}
\newcommand{\scrG}{\mathscr{G}}
\newcommand{\scrI}{\mathscr{I}}
\newcommand{\scrK}{\mathscr{K}}
\newcommand{\scrR}{\mathscr{R}}
\newcommand{\scrP}{\mathscr{P}}
\newcommand{\scrQ}{\mathscr{Q}}
\newcommand{\scrU}{\mathscr{U}}
\newcommand{\scrS}{\mathscr{S}}
\newcommand{\scrW}{\mathscr{W}}
\newcommand{\scrY}{\mathscr{Y}}

\newcommand{\Ical}{\mathcal{I}}

\newcommand{\Ocal}{\mathcal{O}}
\newcommand{\Pcal}{\mathcal{P}}
\newcommand{\Rcal}{\mathcal{R}}

\newcommand{\Scal}{\mathcal{S}}

\newcommand{\Kfrak}{\mathfrak{K}}
\newcommand{\Mfrak}{\mathfrak{M}}

\newcommand{\dfrak}{\mathfrak{d}}

\newcommand{\mfrak}{\mathfrak{m}}
\newcommand{\pfrak}{\mathfrak{p}}

\newcommand{\ZOpen}{\Omega}
\newcommand{\ZOnoether}{\ZOpen_{\sfI}}
\newcommand{\ZOpolar}{\ZOpen_{\sfW}}
  \newcommand{\ZOpolarloc}{\ZOpen^{\hh}_i}
  \newcommand{\ZOpolarrank}{\scrE}
  \newcommand{\ZOpolardim}{\scrD}
\newcommand{\ZOcritpolar}{\ZOpen_{\sfK}}
\newcommand{\ZOfiber}{\ZOpen_{\sfF}}
\newcommand{\linF}{\mathfrak{l}}

\newcommand{\logde}[1]{\log_2\!{#1}\,}

\newcommand{\deltavalue}{(n+c+4)D^{c+2}(D-1)^{d}(c+2)^d}
\newcommand{\overallcomplexity}{\softOh{\mu^3 16^{9d} E
(n\logde{n})^{6(2d-2 + 12\logde{d}) (\logde{d} + 7)}
D^{3(2n+3)(\logde{d} + 5)}}}
\newcommand{\overalldegree}{\softOh{\mu 16^{3d}
(n\logde{n})^{2(2d-2 + 12\logde{d}) (\logde{d} + 6)}
D^{(2n+3)(\logde{d} + 4)}}}

\def\degP{{\mu}}
\def\degQ{{\kappa}}
\def\degS{{\sigma}}
\def\Zeroes#1{{\textsf{Z}(#1)}}

\def\Var#1#2#3{\projfbr{V(#1)}[\Zeroes{#2}][#3]}

\def\bdelta{\bm{\delta}}
\def\bzeta{\bm{\zeta}}

\def\Vlag#1{\mathcal{V}(#1)}
\def\Polar{W}

\def\degB{{\beta}}
\def\degY{{\gamma}}

 \usepackage[colorlinks=true,linktocpage=true]{hyperref}
 \hypersetup{pdfauthor = {Rémi Prébet, Mohab Safey El Din, Éric Schost},
 pdftitle = {Computing roadmaps in unbounded smooth real algebraic sets II:
 algorithm and complexity},
 pdfsubject = {},
 pdfkeywords = {symbolic computation, real algebraic geometry, robotics},
 pdfstartview={FitH}}

\begin{document}

\title{Computing roadmaps in unbounded smooth real algebraic sets II: algorithm
and complexity}

\author{R\'{e}mi \textsc{Pr\'ebet}\\
 Inria, CNRS, ENS de Lyon, Universit\'e Claude
Bernard Lyon 1,\\ LIP, UMR 5668, 69342, Lyon cedex 07, France\\
  \texttt{remi.prebet@ens-lyon.fr}
  \and
  Mohab \textsc{Safey El Din}\\
  Sorbonne Universit\'e, LIP6 CNRS UMR 7606, Paris, France\\
  \texttt{Mohab.Safey@lip6.fr}
  \and
  \'{E}ric \textsc{Schost}\\
  University of Waterloo, David Cheriton School of Computer Science,\\ Waterloo ON, Canada\\
  \texttt{eschost@uwaterloo.ca}
}

\date{\today}

\maketitle

\begin{abstract}
    A roadmap for an algebraic set $V$ defined by polynomials with
\rev{coefficients in the field $\mathbb{Q}$ of rational numbers} is an algebraic
curve contained in $V$ whose intersection with all connected components of
$V\cap\mathbb{R}^{n}$ is connected. These objects, introduced by Canny, can be
used to answer connectivity queries over $V\cap \mathbb{R}^{n}$ provided that
they are required to contain the finite set of query points $\mathcal{P}\subset
V$; in this case,we say that the roadmap is associated to $(V, \mathcal{P})$.

In this paper, we make effective a connectivity result we previously proved,
to design a Monte Carlo algorithm which, on input {\it (i)} a finite sequence of
polynomials defining $V$ (and satisfying some regularity assumptions) and  {\it (ii)}  an algebraic
representation of finitely many query points $\mathcal{P}$ in $V$, computes a 
roadmap for $(V, \mathcal{P})$. This algorithm generalizes the nearly optimal 
one introduced by the last two authors by dropping a boundedness assumption on 
the real trace of $V$.

The output size and running times of our algorithm are both polynomial in
$(nD)^{n\log d}$, where $D$ is the maximal degree of the input equations and
$d$ is the dimension of $V$. As far as we know, the best previously known
algorithm dealing with such sets has an output size and running time \rev{respectively}
polynomial in \rev{$(n^{\log{n}}D)^{n\log n}$} and $(n^{\log{n}}D)^{n\log^2 n}$.
\end{abstract}

\section{Introduction}
Let $\QQ$ be a real field and let $\RR$ (resp. $\CC$) be a real (resp.
algebraic) closure of $\QQ$. One can think about $\QQo$, $\RRo$ and $\CCo$
instead, for the sake of understanding. Further, $n \geq 0$ is an integer which
stands for the dimension of the ambient space in which we compute roadmaps.
In this document we deal with sets in $\RR^n$ and $\CC^n$ defined by polynomial 
equations with coefficients in $\QQ$, that are referred to as respectively 
algebraic sets and real algebraic sets defined over $\QQ$. 
We refer to \cite{Sh1994,Ei1995} and \cite{BCR2013} for precise definition and 
properties of these sets. Considering sets in $\RR^n$ defined by 
polynomial equations and inequalities defines the class of \SA sets; we 
refer to \cite{BCR2013,BPR2006} for a comprehensive study of these sets and 
their properties.

In particular, semi-algebraic and real algebraic sets can be decomposed into 
finitely many \SACCs by \cite[Theorem 2.4.4.]{BCR2013}. Counting these 
components
\cite{GV1992,VG1990} or answering connectivity queries over these sets
\cite{SS1983} finds many applications in e.g. robotics \cite{Ca1988, CSS2020,
  capco:hal-03389500}.

Following \cite{Ca1988,Ca1993}, such computational problems are tackled by
computing objects called \emph{roadmaps} and introduced by Canny in
\cite{Ca1988}. \rev{It is worth noting that some algorithms for computing
roadmaps, which enjoy a so-called property of \emph{divergence}, can also be used to
compute semi-algebraic descriptions of the connected components of the set under
study (see \cite[Chap. 15]{BPR2006}).} In this paper, we focus on the case of
real algebraic sets and provide such a roadmap algorithm.

Given an algebraic set $V\subset \CC^n$ and a finite set of query points $\Pcal
\subset V$, both defined over $\QQ$, a \emph{roadmap} $\Rcal$ associated to $(V,
\Pcal)$ is an algebraic curve which is contained in $V$, which contains
$\Pcal$, and whose intersection with each \SACC of $V\cap \RR^n$ is non-empty
and \SAC. 
\rev{Once a roadmap is computed, one can obtain a graph that is
semi-algebraically homeomorphic to its real trace (see e.g. \cite{IPP23}), which
can then be used to answer connectivity queries.}

Given a polynomial system defining $V$, the effective construction of
roadmaps relies on connectivity statements which allow one to
construct real algebraic subsets of $V\cap \RR^n$, of smaller
dimension, having a connected intersection with the connected
components of $V\cap \RR^n$. Such statements in
\cite{SS2011,BR2014,BRSS2014} make the assumption that $V$ has
finitely many singular points and that $V\cap \RR^n$ is bounded. In
\cite{PSS2024}, a generalization was obtained by dropping the
boundedness assumption.  In this paper, we design a Monte Carlo algorithm for
computing roadmaps based on this latter result, assuming regularity
assumptions\footnote{\newrev{As a sequel to \cite{PSS2024}, this paper similarly
abuses the terminology by calling an algebraic set with isolated singularities
"smooth". This is also consistent with previous literature dealing with such
algebraic sets \cite{SS2011,SS2017}.}} on the system defining $V$. Under those
assumptions, this improves the state of the art complexity.

\subsection{Prior works}

Canny provided the first algorithms for computing roadmaps; we call
such algorithms \emph{roadmap algorithms}. Suppose that $V\subset
\CC^n$ is defined by $s$ polynomials of degrees at most $D$. Canny
obtained in \cite{Ca1988, Ca1993} a Monte Carlo roadmap algorithm
using $(sD)^{O(n^2)}$ arithmetic operations in $\QQ$. A deterministic
version is also given, with a runtime $(sD)^{O(n^{4})}$. This striking
and important result was then reconsidered and improved in
\cite{VG1990,GL93,HRS1994} (among others) to obtain in \cite{BPR2000}
a deterministic algorithm using \rev{$s^{n+1}D^{O(n^2)}$} field operations; this
was the state-of-the-art for decades. 

All these algorithms are based
on the same following geometric solving pattern. First a curve,
defined as the critical locus of a projection on a plane, is computed;
it meets all \SACCs of the set under study.  Next, connectivity
failures are repaired by slicing our set with appropriate hyperplanes,
performing recursive calls over these slices. \rev{This geometric pattern also
provides us 
with a connecting procedure when a new query point is considered: it basically
consists in slicing the variety with a hyperplane connecting this point to the
first critical curve we  considered, computing a roadmap in the slice, which is
then connected to the whole roadmap with the critical curve.} \rev{This gives a
special recursive structure to such connecting procedures; combined with the
fact that they share the initial roadmap as ``skeleton'', this 
yields the aforementioned divergence property.}

 The algorithm designed in \cite{SS2011}, \rev{which is Monte Carlo}, 
 is the first one to be based
on a different geometric solving pattern, thanks to an innovative
geometric connectivity theorem (under assumptions on the input
variety, in particular boundedness). This theorem gives much more
freedom in the way we can construct roadmaps; in particular, it allows one
to choose critical loci \rev{of projections} of higher dimension, which makes it possible
to slice the input with sections of smaller dimension. This yields a
better balance between the dimensions of these geometric objects,
reducing the depth of the recursion. The algorithm in \cite{SS2011} is
a first prototype of this new family of roadmap algorithms; it is
Monte Carlo and runs in time $(nD)^{O(n\sqrt{n})}$; the algorithm
in~\cite{BRSS2014} has similar runtime, but drops all assumptions from
\cite{SS2011} \rev{and is deterministic}.

In \cite{BR2014}, the authors provide a deterministic
algorithm which runs in time $(n^{\log{n}}D)^{O(n\log^2{n})}$, for an output
roadmap with degree $(nD)^{O(n^{\log{n}}\log{n})}$. By re-introducing some
regularity and boundedness assumption, the first roadmap
algorithm running in time $(nD)^{O(n\log{d})}$ is given in
\cite{SS2017}, where $d$ is the dimension of the input algebraic
set. \rev{This last algorithm is  Monte Carlo and explicit constants are
provided in the big-O exponent}, showing that the algorithm runs in time
sub-quadratic in the degree bound of the output.
\rev{Precisely, letting \[\scrB
= 16^{3d} \left(n\logde{n}
    \right)^{4d\logde{d} + O(d)}
D^{2n\logde{d} + O(n)},\] this algorithm uses $E \scrB^{3}$ arithmetic
operations in  $\QQ$, where $E$ is the complexity of evaluating the input system
of polynomial equations, whereas the output size is at most $\scrB^{2}$.}

\rev{It is worth noting that all these algorithms enjoy the aforementioned
divergence property even if the connectivity result they use is not the same. The recursive structure of these algorithms, which first consider a critical
locus of prescribed dimension, and repair connectivity failures with slices of
lower dimension, leads to connecting procedures that have a similar pattern,
sharing the initial roadmap as a skeleton to which all additional query points
can be connected.}

\rev{The result of \cite{SS2017} is appealing and raises the hope to obtain
practically faster implementations for computing roadmaps, as it makes
explicit the constants which were hidden by the Landau notation used in the
exponent in prior works, and because it relates in a close manner the number of arithmetic
operations with the worst-case output size.}

\rev{
However,  in several applications where
roadmaps are used for mechanism design (see e.g. \cite{CSS2020,
  capco:hal-03389500, chablat:hal-03596704}), there is a need for
computing roadmaps in unbounded real algebraic sets, still satisfying smoothness
properties. Also, there is a need to obtain roadmap algorithms that can be used for
higher level algorithms in real algebraic geometry for e.g. computing
semi-algebraic descriptions of connected components of real algebraic sets.} 

\rev{This leaves open the problem of obtaining roadmap algorithms, dropping 
the boundedness assumption, while still using at most $E\scrB^3$ arithmetic
operations in $\QQ$.}

\subsection{\rev{Bottlenecks}}

\newrev{We discuss the bottlenecks one encounters
when tackling this open problem, starting with an
overview. To drop the boundedness assumption, one can use
reductions to the bounded case by e.g.  lifting \(V\) 
to \(\CC^{n+1}\) and taking the intersection with a
hypersphere of radius infinitesimally large, or by considering the
intersection of \(V\subset \CC^{n}\) with a hyperball of
\(\RR^n\) of large enough
radius (which would then be pre-computed). Both cases
result in an increase of the comple\-xity, beyond the bounds
we target. Moreover, the second approach yields the loss of
the aforementioned divergence property. These strategies are
discussed in detail below.}

\rev{Several natural techniques can be considered to remove the boundedness assumption. A
    first one, which is used already in e.g. \cite{BRSS2014} or \cite{BR2014}, 
    is to embed the algebraic
set $V$ defined by the input system in  $\CC^{n+1}$ and consider the
intersection of  $V$ with a  $n$-dimensional sphere $\Scal\subset \RR^{n+1}$ of
 infinitesimally large radius. The next step is then to compute a roadmap 
for $V\cap \Scal$. This requires one more variable and the use of an infinitesimal,
which makes all arithmetic operations more expensive. As far as we know, this
does not allow one to obtain a roadmap algorithm within our objective of roughly $\scrB^3$
 operations in $\QQ$. }

\rev{A variant of such an approach would be to pre-compute a large enough
    radius for $\Scal$, say $\rho$, as for instance in \cite{capco:hal-03389500}.
    However, there are $n+1$ variables instead of  $n$ to handle and the
    intersection with  $\Scal$ doubles the degree of the algebraic set we work with.
    More importantly, the roadmap we get would be valid
    only in $\Scal$: it must be recomputed when query points in $V\cap \RR^n$ are
    chosen outside of $\Scal$. Hence, such
    an approach breaks the divergence property and could not be
    used to compute semi-algebraic descriptions of the real algebraic set $V\cap
    \RR^n$.}

\rev{An approach that would consist in intersecting $V$ with a
    ball $B$ in  $\RR^n$ of large enough radius, and then computing a roadmap for
    the semi-algebraic set $V\cap B$, would share these shortcomings. More importantly, the current
    state-of-the-art algorithms for computing roadmaps in semi-algebraic sets is
exponential in $n^2$, while we target complexities which are exponential in
$n\logde{d}$.}

\rev{Hence, in order to obtain roadmap algorithms for smooth real algebraic sets
with arithmetic complexity cubic in $\scrB$, without using a reduction to a bounded
input, a number of new ingredients are required.}
\rev{A first one is a connectivity result which leaves as much freedom as the
one in \cite{SS2011}, but does not make use of any boundedness assumption.}
\rev{This is done in~\cite{PSS2024} where we showed how to drop
the boundedness assumption in such connectivity results.} 
\rev{A key new ingredient is that instead of considering critical loci of
projections as is done in \cite{SS2011}, we consider critical loci of
\emph{proper polynomial maps} (e.g. with a quadratic instead of linear component).
Slices of $V$ to repair connectivity failures are
chosen accordingly.}

\rev{This connectivity result by far not sufficient to reach our goal.}
\rev{The purpose of this article is to show which data structures are needed and 
    how the computations can be
organized and analyzed to reach a runtime similar to
the one of \cite{SS2017}, without the boundedness assumption. }

\subsection{Data representations}
Before entering into a detailed description of this complexity result, we start
by recalling the data representations we use to encode our input and output.

\paragraph*{Straight-line programs}
Input polynomials will be represented as \emph{\SLPs}, which
is a flexible way of representing multivariate polynomials as a
division\rev{-} and loop-free sequences of operations. Formally, a \SLP $\Gamma$
of length $E$, computing polynomials in $\QQ[\XX]$, with $\XX = x_1,
\ldots, x_n$, is a sequence
$\Gamma=(\gamma_1,\dotsc,\gamma_E)$ such that for all $1\leq i \leq
E$, one of the two following holds:
\begin{itemize}
 \item $\gamma_i = \lambda_i$ with $\lambda_i \in \QQ$;
\item $\gamma_i = (\opsf_i,a_i,b_i)$ with $\textsf{op}_i \in \{+,-,\times\}$ 
and $-n+1 \leq a_i,b_i < i$.
\end{itemize}
To $\Gamma$ we associate polynomials $G_{-n+1},\dotsc,G_E$ such that $G_i = 
x_{i+n}$ for $-n+1 \leq i \leq 0$, and for $i\geq 1$:
\begin{itemize}
 \item if $\gamma_i = \lambda_i$ then $G_i = \lambda_i$;
\item if $\gamma_i = (\opsf_i,a_i,b_i)$ then $G_i = G_{a_i} \: \opsf_i \: 
G_{b_i}$.
\end{itemize}
Then we say that $\Gamma$ computes some polynomials $f_1,\dotsc,f_c \in 
\QQ[\XX]$ if it holds that $\{f_1,\dotsc,f_c\} \subset \{G_{-n+1},\dotsc,G_E\}$.

\begin{example} We give an illustrating example presented in \cite[Section 
1.1]{Kr2002}.
 For $m\in \NNo^{*}$, a \SLP computing $x_1^{2^m}$ in $\QQ[x_1,x_2]$ is given by
 taking
 \[
   \left\{
     \begin{array}{ccl}
       \gamma_1 & = &(\times,-1,-1) \\
       \gamma_2 &= &(\times, 1, 1) \\
       &\vdots& \\
       \gamma_{m} &= &(\times, m-1,m-1) \\
     \end{array}
   \right.
 \]
 where we associate $G_1 = x_1^2$ to $\gamma_1$, $G_2 = G_1^{2} = x_1^{4}$ to 
$\gamma_2$ and
 so on with $G_m = G_{m-1}^{2} = x_1^{2^{m}}$ which is associated to 
$\gamma_{m}$.
 Such a program has length $m$, while the dense and sparse representations of
 $x_1^{2^m}$ have respective length $2^m+1$ and $2$. But remark that a \SLP
 computing $(x_1+x_2)^{2^m}$ can be obtained by setting
$\gamma_1=(+,-1,0)$, which computes $x_1+x_2$, and adding $\gamma_{m+1} = (\times, m,m)$ 
at the end. The latter modification 
increments the length by one, while the sparse representation has now length 
$\binom{\,2^m}{\,2^{\hphantom{m}}}$.
\end{example}

Because of the good behaviour of such a representation with respect to
linear changes of variables, it is used as input in many algorithms
for solving polynomial
systems~\cite{Kr2002,GHMMP1998,GHMP1997,GHMP1995,GLS2001,Le2000}. It
is not restrictive since any polynomial of degree $D$ in $n$
variables can be computed with a \SLP of length $O(D^n)$ by simply
evaluating and summing all its monomials. By convention, we note
$\Gamma^0 = (0)$ the \SLP of length $1$ that computes the zero
polynomial.

\paragraph*{Zero-dimensional parametrizations}
A finite set of points defined over $\QQ$ is represented using zero-dimensional
parametrizations. A \emph{zero-dimensional parametrization} $\scrP$ with
coefficients in $\QQ$ consists of:
\begin{itemize}
\item  polynomials  $(\omega, \rho_1, \ldots, \rho_n)$ in $\QQ[u]$ where $u$ is
  a new variable, $\omega$ is a monic square-free polynomial and it holds that
  $\deg(\rho_i) < \deg(w)$,
\item  a linear form $\linF$  in variables $x_1,\dotsc, x_n$,
\end{itemize}
such that 
\[
\linF(\rho_1,\dotsc,\rho_n) = u \mod \omega.
\]
Such a data structure encodes the finite set of points, denoted by 
$\Zparam(\scrP)$ defined as follows
\[
\Zparam(\scrP) = \left\{\left (\rho_1(\vartheta), 
\ldots, \rho_n(\vartheta) \right) 
\in \CC^n \mid \omega(\vartheta) = 0\right\}.
\]
According to this definition, the roots of $\omega$ are exactly the
values taken by $\linF$ on $\Zparam(\scrP)$. We define the {\em
  degree} of such a parametrization $\scrP$ as the degree of the
polynomial $\omega$, which is exactly the cardinality of
$\Zparam(\scrP)$. By convention we note $\scrP_{\emptyset} = (1)$ the
zero-dimensional parametrization that encodes the empty set.

Such parametrizations will be used to encode input query points and
internally in our roadmap algorithm to manipulate finite sets of
points.

\paragraph*{One-dimensional parametrizations}
Algebraic curves defined over $\QQ$ will be represented using one-dimensional
rational parametrizations. A \emph{one-dimensional rational parametrization}
$\scrR$ with coefficients in $\QQ$ is a couple as follows:
\begin{itemize}
\item polynomials $(\omega, \rho_1, \ldots, \rho_n)$ in $\QQ[u, v]$
  where $u$ and $v$ are new variables, $\omega$ is a monic square-free
  polynomial and with $\deg(\rho_i) < \deg(w)$,
\item  linear forms $(\linF, \linF')$ in the variables $x_1,\dotsc, x_n$,
\end{itemize}
such that
\[
  \linF(\rho_1,\dotsc,\rho_n) = u \frac{\partial \omega}{\partial u} \mod
  \omega
\]
and 
\[
\linF'(\rho_1,\dotsc,\rho_n) = v \frac{\partial \omega}{\partial u} 
\mod \omega.
\]
Such a data structure encodes the algebraic curve, denoted by $\Zparam(\scrR)$, 
defined as the Zariski closure of the following constructible set
\[
 \left \{\left (\frac{\rho_1(\vartheta, \eta)}{\partial \omega / \partial u 
(\vartheta, \eta)}, 
\ldots, 
\frac{\rho_n(\vartheta, \eta)}{\partial \omega / \partial u (\vartheta, \eta)}
\right) \in\CC^n \mid \omega(\vartheta, \eta) = 0, \frac{\partial 
\omega}{\partial u}(\vartheta, \eta)\neq 0\right \}.
\]
We define the {\em degree} $\deg(\scrR)$ of such a parametrization
$\scrR$ as the degree of $\Zparam(\scrR)$ (that is, the maximum of the
cardinalities of the finite sets obtained by intersecting
$\Zparam(\scrR)$ with a hyperplane). Any algebraic curve $C$ can be
described as $C=\Zparam(\scrR)$, for some one-dimensional rational
parametrization that satisfies $\deg(\omega) = \deg(\scrR)$ (that is,
the degree of $C$); guaranteeing this is the reason why we use
rational functions with the specific denominator $\partial \omega /
\partial u$.  This will always be our choice in the sequel; in this
case, storing a one-dimensional parametrization $\scrR$ of degree $\delta$ involves
$O(n\delta^2)$ coefficients.

Our algorithm computes a roadmap $R$ of an algebraic set $V$,
\rev{with $R$ having}, by definition, dimension one. The output is given by means of a
one-dimensional parametrization of $R$.

\subsection{Contributions}

Recall that an algebraic set $V$ can be uniquely decomposed into
finitely many \emph{irreducible components}. When all these components
have the same dimension $d$, we say that $V$ is
\emph{$d$-equidimensional}. For a set of polynomials $\ff \subset
\CC[\XX]$, with $\XX = x_1, \ldots, x_n$, we denote by $\V(\ff)
\subset \CC^n$ the vanishing locus of $\ff$, and by $\Ocal(\ff)$ its
complement.

Given an algebraic set $V\subset \CC^n$, we denote by $\I(V)$ the
\emph{ideal of $V$}, that is the ideal of $\CC[\XX]$ of polynomials
vanishing on $V$. For $\yy \in \CC^n$, we denote by $\jac_\yy(\ff)$
the Jacobian matrix of the polynomials $\ff$ evaluated at $\yy$. For
$V$ $d$-equidimensional in $\CC^n$, those points $\yy\in V$ at 
which the Jacobian
matrix of a finite set of generators of $\I(V)$ has rank $n-d$ are
called \emph{regular} points and the set of those points is denoted by
$\reg(V)$. The others are called \emph{singular} points; the set of
singular points of $V$ (its singular locus) is denoted by $\sing(V)$
and is an algebraic subset of $V$.

We say that $(f_1,\dotsc,f_p)\subset \QQ[\XX]$ is a \emph{reduced regular
  sequence} if the following holds for every $i\in\{1,\dotsc,p\}$:
\begin{itemize}
\item the algebraic set $\V(f_1,\dotsc,f_i)\subset \CC^{n}$ is either
  equidimensional of dimension $n-i$ or empty,
\item the ideal $\langle f_1,\dotsc,f_i \rangle$ is radical.
\end{itemize}

In the following main result, and in all this work, we design, and
also use known algorithms, whose success relies on the successive
choice of several parameters $\blambda_1,\blambda_2,\dots$ in affine
spaces $\QQ^{a_1},\QQ^{a_2},\dots$ These algorithms are probabilistic
in the sense that, in any such case, there exist non-zero polynomials
$\Delta_1,\Delta_2,\dots$, such that the algorithm is successful if
$\Delta_i(\blambda_i)\neq0$ for all $i$. These algorithms are
\emph{Monte Carlo}, as we cannot always guarantee correctness of the
output with reasonable complexity.  Nevertheless, in cases when we can
detect failure, our procedures will output \fail (though not returning
\fail does not guarantee correctness). 

Our main result is the
following one. Below, the soft-O notation $\Otilde(g)$ denotes the class
$O(g(\logde{g})^a)$ for some constant $a>0$, where $\log_2$ is the binary
logarithm function.
\begin{theorem}\label{thm:main}
Let $\ff=(f_1,\dotsc,f_p)\subset \QQ[\XXi]$ be a reduced regular sequence, let
  $D\geq 2$ be bounding the degrees of the $f_i$'s and suppose that $\Gamma$ is a \SLP
  of length $E$ evaluating $\ff$. Assume that $\V(\ff)\subset \CC^n$ has
  finitely many singular points.

Let $\scrP$ be a zero-dimensional parametrization of degree $\mu$ with
$\Zparam(\scrP)\subset V(\ff)$. 
There exists a Monte Carlo algorithm which, on input $\Gamma$ and $\scrP$ 
computes a \ODP $\scrR$ of a roadmap of $(\V(\ff),\Zparam(\scrP))$ of 
degree 
\[
    \scrB = \mu\, 16^{3d} n^{4d\logde{d} + O(d)}\; D^{2n\!\logde{d} + O(n)}
= \mu (nD)^{O(n\log{d})},
\]
where $d=n-p$, using \underline{$E\scrB^3$ arithmetic operations in $\QQ$}.
\end{theorem}

Hence, we dropped the boundedness assumption on $\V(\ff)\cap\RR^n$ made in
\cite[Theorem 1.1]{SS2017}, still keeping a complexity similar to the algorithm
presented in \cite{SS2017}.  Note that the arithmetic complexity statement above
is cubic in the {\em degree} bound $\mathscr{B}$ on the output; the output {\em
size} itself is $O\left( n\mathscr{B}^2 \right)$ elements in $\QQ$. Hence, as in
\cite{SS2017}, our runtime is sub-quadratic in the bound on the output size.
Note that a \emph{complexity bound with a comprehensive exponent}, that is
without using big Oh notation, is given below in Theorem~\ref{thm:corrcomp}, of
which the above main result is a direct consequence.
\newrev{Note that the (more general) algorithm in \cite{BR2014} runs in
    (deterministic) time \(\mu^{\log^2{n}} \left( n^{\log{n}}D
        \right)^{O\left( n\log^2{n} \right)
} \) which is not polynomial in its output size. Note also that the
complexity constant hidden by the Landau notation in \cite{BR2014} is
not known. Our output size is smaller, and for families of smooth real algebraic sets
of {\em fixed dimension}, our algorithm runs in time \(\left( nD
\right)^{O(n)} \). }

Our algorithm works as follows:
\begin{itemize}
\item it starts by considering the critical locus $W$ associated to
  some special polynomial map $\bphi$ with image in $\RR^{2}$;
\item next, it uses a
variant of the algorithm in \cite{SS2017}, to
  deal with some slices $V(\ff)\cap \bphi^{-1}(v_i)$ for some $v_1,
  \ldots, v_\ell$ in $\RR$.
\end{itemize}
We will see that the map $\bphi$ depends on some parameters in
$\CC^{N}$, for some $N\geq 0$, and that there exists a non-zero
polynomial $G$ such that when choosing these parameters in the open
set $\Ocal(G)$, one can apply the connectivity result in
\cite{PSS2024}, which does not need any boundedness assumption. This
is where first elements of randomization are needed. A second element
comes from \rev{our} use of a variant of \cite{SS2017}, which is also a
Monte Carlo algorithm.

\begin{remark}
  {
    As already emphasized, our algorithm is Monte Carlo and its
    success depends on the choice of several parameters (the
    aforementioned map \(\bphi\) as well as linear change of
    coordinates and linear forms). The set of "bad" choices is
    enclosed in some proper Zariski closed subsets (see e.g.
    Proposition~\ref{prop:fibergenalg}, as well as \cite[Prop.
    3.5]{SS2017} and \cite[Prop. 3.7]{SS2017}). In order to estimate
    the probabi\-lity of success of such choices, one needs to bound the
    degrees of these proper Zariski closed subsets as is done in
    \cite{ElGiSc23} for algorithms computing sample points per
    connected component of a smooth real algebraic set. While
    \cite{ElGiSc23} provides already several tools towards obtaining
    such degree bounds, some dedicated future work is needed.
    
    However, it is worth pointing out that in several situations,
    these choices are made to ensure regularity properties
    (dimension, smoothness, etc.) of some algebraic sets, and the
    algorithm will detect and return "fail" when such properties are not satisfied. 
    It may also return "fail" when the algebraic elimination routines
    it relies on do so (this is due to probabilistic aspects of these
    elimination
    routines, such as choice of random projections, or of some prime numbers).
    }
\end{remark}

 \section{Preliminaries}
\subsection{Minors, rank and submatrices} 
We present here some technical results on the minors and the rank of a certain 
class of matrices that will occur in this paper, when dealing with particular 
cases and incidence varieties in
Section~\ref{sec:subroutines}.

\begin{lemma}\label{lem:minorssubmat}
  Let $q\geq 1$ and $1\leq c \leq p$ be integers. Let $A,B,C$ be respectively
  $c\times p$, $c\times q$ and $q\times p$ matrices with coefficients in a
  commutative ring $R$ such that $M_1$ and $M_2$ are the following
  $(c+q)\times(q+p)$ matrices:
 \[
  M_1=\begin{bmatrix}
   B & A\\
   I_q & \OO
  \end{bmatrix}\qquad\text{and}\qquad
  M_2=\begin{bmatrix}
   \OO & A\\
   I_q & C
  \end{bmatrix},
 \]
 where $I_q$ is the identity matrix of size $q$.
 Let $m\in R$ and $e$ be in $\{0,\dots, c\}$; then, for $k=1,2$ the following 
conditions are equivalent:
 \begin{enumerate}
  \item $m$ is the determinant of a $(q+e)$-submatrix of $M_k$ that contains 
$I_q$;
  \item $(-1)^{qe} m$ is an $e$-minor of $A$.
 \end{enumerate}
 In this case, if $1\leq i_1\leq\cdots\leq i_{e}\leq c$ and $1\leq
 j_1,\dotsc,j_{e}\leq p$ are the indices of respectively the rows and
 the columns of $A$ selected in item 2, then the corresponding
 rows and columns in $M_k$ are of respective indices
 \[
  1\leq i_1\leq \cdots\leq i_{e} \leq c+1 \leq \cdots \leq c+q \et
  1\leq \cdots \leq q\leq  q+j_1 \leq\cdots\leq q+j_{e}.
 \]
\end{lemma}
\begin{proof}
  The determinant of any submatrix of $M_i$ containing $I_q$ can be
  reduced, up to the sign $(-1)^{qe}$, to a minor of $A$ by using the
  cofactor expansion with respect to the last $q$ rows of $M_1$
  (resp.\ the first $q$ columns of $M_2$). Conversely, any $e$-minor
  of $A$ is a $(q+e)$-minor of $M_k$, by extending the associated
  submatrix of $A$ to a submatrix of $M_k$ containing $I_q$. The
  correspondence between indices stated above is then straightforward.
\end{proof}

\begin{lemma}\label{lem:ranksubmatrix}
  With the notation of Lemma~\ref{lem:minorssubmat}, if $R$ is a
  field, then $\rank(M_k) = \rank(A) + q \geq q$ for $k=1,2$.
\end{lemma}
\begin{proof}
  For $k=1$, performing row operations allows us to replace $B$ by
  the zero matrix, after which the claim becomes evident. For $k=2$,
  use column operations.
\end{proof}

\subsection{Polynomial maps, generalized polar varieties and fibers}

Let $Z\subset \CC^n$ be an equidimensional algebraic set and $\bphi =
(\phi_1,\dotsc,\phi_m)$ be a finite set of polynomials of $\CC[\XX]$;
we still denote by $\bphi\colon Z \rightarrow \CC^m$ the restriction
of the polynomial map induced by $\bphi$ to $Z$. For $\yy \in \reg(Z)$, 
we say that $\yy\in Z$ is a \emph{critical point} of $\bphi$ if
\[ 
 d_\yy \bphi(T_\yy Z) \neq \CC^m,
\]
where $d_\yy \bphi$ is the differential of $\bphi$ at $\yy$. We will denote by
$\Wo(\bphi,Z)$ the set of the critical points of $\bphi$ on $Z$, and by
$W(\bphi,Z)$ its Zariski closure. A \emph{critical value} is the image of a 
critical point by $\bphi$.

Besides, we let $K(\bphi,Z) = \Wo(\bphi,Z) \bigcup \sing(Z)$ be the
set of \emph{singular points} of $\bphi$ on $Z$. When $Z$ is defined
by a reduced regular sequence $\ff = (f_1, \ldots, f_c)$, $K(\bphi,
Z)$ is then defined as the intersection of $Z$ with the set of points
of $\CC^{n}$ where the Jacobian matrix of $(\ff, \bphi)$ has rank at
most $c + m - 1$ (see~\cite[Lemma A.2]{SS2017}).

For $1\leq i \leq m$, we set
\[
 \begin{array}{cccc}
   \map[i] \colon & \CC^n & \to & \CC^i\\
  &\yy & \mapsto & (\phi_{1}(\yy),\dotsc,\phi_{i}(\yy)).
 \end{array}
\]
Given the maps $(\map[i])_{1 \leq i \leq m}$, we denote
$\Wo(\map[i],Z)$, $W(\map[i],Z)$ and $K(\map[i],Z)$ by respectively
$\Wophii[i][Z]$, $\Wphii[i][Z]$ and $\Kphii[i][Z]$. For $i=0$, we let
$\CC^0 $ be a singleton of the form $ \CC^{0} = \{ \bullet \}$, and
$\map[0] \colon \yy\in\CC^n \to \bullet\in\CC^0$ be the unique
possible map. Then for all $ \yy \in \CC^0, \; \map[0][-1](\yy) =
\CC^n$; we set $\Wophii[0][Z] = \Wphii[0][Z] = \emptyset$. The sets 
$\Wphii[i][Z]$,
for $0\leq i\leq m$ are called the \emph{generalized polar varieties
  associated to $\phi$ on $Z$}.

The main result we state in this subsection is the following (the
somewhat lengthy proof is in Section~\ref{sec:fibergenalg}). It
establishes some genericity properties of generalized polar varieties
associated to a class of polynomial maps.
It is a generalization of \cite[Theorem 1]{SS2003}, which only deals with 
projections.

\begin{proposition}\label{prop:fibergenalg}
  Let $V\subset \CC^n$ be a $d$-equidimensional algebraic set with finitely many
  singular points and $\polun$ be in $\CC[\XX]$. Let $2\leq \ifi \leq d+1$. For 
$\balf=(\balf_1,\dots,\balf_{\ifi})$  in $\CC^{\ifi n}$, we define
   $\bphi = \left(\phi_1(\XX,\balf_1),  
\dotsc,\phi_{\ifi}(\XX,\balf_{\ifi})\right)$, where for $2\leq j \leq \ifi$
 \begin{align*}
  \phi_1(\XX,\balf_1) = \polun(\XX) + 
  \sum_{k=1}^n \alf_{1,k}x_k
  \et
  \phi_j(\XX,\balf_j) = \sum_{k=1}^n \alf_{j,k}x_k.
 \end{align*}
  Then, there exists a non-empty Zariski open subset $\ZOnoether(V, \polun,\ifi)
 \subset \CC^{\ifi n}$ such that for every $\balf \in 
\ZOnoether(V,\polun,\ifi)$ and $i \in \{1,\dotsc,\ifi\}$, the following holds:
 \begin{enumerate}
  \item either $\Wphii$ is empty or $(i-1)$-equidimensional;
  \item the restriction of $\map[i-1]$ to $\Wphii$ is a Zariski-closed map;
  \item for any $\zz \in \CC^{i-1}$, the fiber $\Kphii \cap \map[i-1][-1](\zz)$
    is finite.
\end{enumerate}
\end{proposition}

The connectivity result \cite[Theorem 1.1]{PSS2024} makes use of
generalized polar varieties satisfying these properties, but also of
fibers of polynomial maps.

\begin{remark}
Let $\bphi=(\phi_1,\dotsc,\phi_{\ifi})$ be polynomials in
$\CC[\XX]$ and an integer $1\leq e \leq \ifi$. Given an algebraic set $V
\subset \CC^n$ and a set $Q\subset \CC^e$, the \emph{fiber of $V$ over
  $Q$ with respect to $\bphi$} is the set $\mapfbr{V} = V \cap
\map[e][-1](Q)$. We say that \emph{$V$ lies over $Q$ with respect to
  $\bphi$} if $\map[e](V) \subset Q$.  Finally, for $\zz \in \CC^e$,
the set $\mapfbr{V}[\{\zz\}]$ will be denoted by $\mapfbreq{V}[\zz]$.
\end{remark}

\subsection{Charts and atlases of algebraic sets}

We say that an algebraic set is \emph{complete intersection} if it can
be defined by a number of equations equal to its codimension. Not all
algebraic sets are complete intersections; for instance determinantal
varieties and, consequently, a whole class of generalized polar
varieties, are a prototype of non complete intersections. This creates
complications to control the complexity of algorithms manipulating
generalized polar varieties recursively.

However, we may use local representations which describe Zariski open
subsets of an algebraic set with a number of equations equal to its
codimension.

Such local representations are obtained by considering {\itshape locally
  closed sets}. We say that a subset $\Vo$ of $\CC^n$ is locally
closed if there exist an open $\Ocal$ and a closed Zariski subset $Z$
of $\CC^n$ such that $\Vo=Z\cap\Ocal$. In that case, the dimension of
$\Vo$ is the dimension of its Zariski closure $V$, and $\Vo$ is said
to be equidimensional if $V$ is. In this situation, we define
$\reg(\Vo) = \reg(V) \cap\Vo$ and $\sing(\Vo)= \sing(V) \cap \Vo$, and
$\Vo$ is said to be non-singular if $\reg(\Vo) = \Vo$.
For $\ff=(f_1,\dotsc,f_p) \subset \CC[\XX]$ with $p\leq n$, we define
the locally closed set $\Voreg(\ff)$ as the set of all $\yy$ where the
Jacobian matrix $\jac(\ff)$ of $\ff$ has full rank $p$. We
will denote by $\Vreg(\ff)$ the Zariski closure of $\Voreg(\ff)$.

A {\em chart} associated to an algebraic set $V\subset \CC^{n}$ can be
seen as a local representation of $V$ by another locally closed subset
of $V$ that is smooth and in complete intersection. We recall
hereafter the definitions introduced in \cite[Section 2.5]{SS2017},
which we slightly generalize. Below, for a polynomial $m$ in $\CC[\XX]$,
recall that we write $\Ocal(m)=\CC^n-\V(m)$.

\begin{definition}[Charts of algebraic sets]\label{def:chart}
 Let $1 \leq e \leq \ifi \leq n+1$ be integers and 
$\bphi=(\phi_1,\dotsc,\phi_{\ifi})
 \subset \CC[\XX]$.  Let $Q \subset \CC^e$ be a \emph{finite} set and
 $V,S \subset \CC^n$ be algebraic sets lying over $Q$ with respect to
 $\bphi$.  We say that a pair of the form $\chi = (m,\hh)$ with $m$
 and $\hh=(h_1\dotsc,h_c)$ in $\CC[\XX]$ is a {\em chart} of $(V,Q,S,\bphi)$
 if the following holds:
 \begin{enumerate}[label=$(\sfC_\arabic*)$]
  \item $\Ocal(m)\cap V - S$ is non-empty; 
  \item $\Ocal(m) \cap V - S = \Ocal(m) \cap \mapfbr{\V(\hh)} - S$;
  \item $e + c \leq n$;
  \item for all $\yy \in \Ocal(m) \cap V - S$, $\jac_\yy([\hh,\map[e]])$ has 
full rank $c+e$.
 \end{enumerate}
 When $\bphi=(x_1,\dotsc,x_{\ifi})$ defines the canonical projection, one will 
simply refer to $\chi$ as a chart of $(V,Q,S)$, and if $e=0$ as a chart of 
$(V,S)$ (no matter what $\bphi$ is).
\end{definition}

The first condition $\sfC_1$ ensures that $\chi$ is not trivial, and
the following ones ensure that $\chi$ is a smooth representation of
$V-S$ in complete intersection (for $V$ equidimensional,
$S$ contains the singular points of $V$). This is a generalization of
\cite[Definition 2.2]{SS2017} in the sense that, if
$\bphi=(x_1,\dotsc,x_n)$, one recovers the same definition.

\begin{lemma}\label{lem:jacrankchart}
  Let $1 \leq \ifi \leq n+1$ be integers and $\bphi=(\phi_1,\dotsc,\phi_{\ifi})
 \subset \CC[\XX]$.
  Let $V,S\subset \CC^n$ be two algebraic sets with $V$
  $d$-equidimensional and let $\chi = (m,\hh)$, with $\hh =  (h_1,\dotsc,h_c)$, 
be a chart of $(V,S)$. Then, for $1\leq i\leq \ifi$ and $\yy \in \Ocal(m) \cap 
V - S$, $\yy$ lies in $\Wphii$ if and only if $\jac_{\yy}([\hh,\map[i]])$ does 
not have full rank $c+i$.
\end{lemma}
\begin{proof}
  Let $\yy \in \Ocal(m) \cap V - S$. By \cite[Lemma A.8]{SS2017}, $\yy
  \in \reg(V)$, so that $\yy$ lies in $\Wphii$ if and only if it lies
  in $\Wophii$.  Besides, by \cite[Lemma A.7]{SS2017}, $\Tg_{\yy}V$
  coincide with $\ker\jac_\yy(\hh)$. Hence, by definition $\yy$ lies
  in $\Wphii$ if and only if $d_{\yy}\map[i](\ker\jac_{\yy}(\hh)) \neq
  \CC^i$. But the latter, is equivalent to saying that the matrix
  $\jac_{\yy}([\hh,\map[i]])$ does not have full rank $c+i$.
\end{proof}
A straightforward rewriting of Lemma~\ref{lem:jacrankchart} is the following
which provides a local description of a polar variety by means of a critical
locus on a variety defined by a complete intersection.
\begin{lemma}\label{lem:polaronchart}
  Reusing the notation of Lemma~\ref{lem:jacrankchart}, it holds that the sets
  $\Wphii$ and $\Wophii[i][\Vreg(\hh)]$ coincide in $\Ocal(m) - S$.
\end{lemma}

Together with the notion of charts, we define  atlases as a collection of
 charts that cover the whole algebraic set we consider.
\begin{definition}[Atlases of algebraic sets]\label{def:atlas}
  Let $1 \leq e \leq \ifi \leq n+1$ be integers and 
$\bphi=(\phi_1,\dotsc,\phi_{\ifi}) \subset \CC[\XX]$. Let $Q \subset \CC^e$ be 
a \emph{finite} set and $V,S \subset  \CC^n$ be algebraic sets lying over $Q$ 
with respect to $\bphi$. Let $\bchi =
  (\chi_j)_{1\leq j \leq s}$ with $\chi_j = (m_j,\hh_j)$ for all $j$. We say 
that $\bchi$ is an {\em atlas} of $(V,Q,S,\bphi)$ if the following holds:
\begin{enumerate}[label=$(\sfA_\arabic*)$]
 \item $s\geq 1$;
 \item for each $1 \leq j \leq s$, $\chi_i$ is a chart of $(V,Q,S,\bphi)$;
 \item $V-S \subset \bigcup_{1 \leq j \leq s} \Ocal(m_j)$.
\end{enumerate}
When $\bphi=(x_1,\dotsc,x_\ifi)$ is defines the canonical projection, one simply
refers to $\bchi$ as an atlas of $(V,Q,S)$, and if $e=0$ as an atlas of $(V,S)$.
\end{definition}
Here the definition is the same as \cite[Definition 2.3]{SS2017}. Note that,
according to \cite[Lemma A.13]{SS2017}, there exists an atlas of $(V,\sing(V))$
for any equidimensional algebraic set $V$.

\subsection{Charts and atlases for generalized polar
  varieties}\label{ssec:polarvar}

We deal now with the geometry of generalized polar varieties (under
genericity assumptions) and show how to define charts and atlases for
them. In this whole subsection, we let $\bphi = (\phi_1,\dotsc,\phi_{n+1})
\subset \CC[\XX]$; for $1\leq i \leq n$, we denote by $\map[i]$ the
sequence $(\phi_1, \ldots, \phi_i)$ and, by a slight abuse
of notation, the polynomial map it defines.

\begin{definition}\label{def:hphi}
  Let $\hh = (h_1,\dotsc,h_c) \subset \CC[\XX]$ with $1\leq c \leq n$ and let $i
  \in \{1,\dotsc,n-c+1\}$. Let $m''$ be a $(c+i-1)$-minor of 
$\jac([\hh,\map[i]])$ containing the rows of $\jac(\map[i])$. We denote by 
$\Hphi(\hh,i,m'')$ the
  sequence of $(c+i)$-minors of $\jac([\hh,\map[i]])$ obtained by successively
  adding the missing row and a missing column of $\jac([\hh,\map[i]])$ to $m''$.
  This sequence has length $n-c-i+1$.
\end{definition}

Then, given a chart $\chi = (m, \hh)$ of some algebraic set $V$, we
can define a candidate for being a chart of generalized polar
varieties associated to $\map[i]$ and $V(\hh)\cap \Ocal(m)$.

\begin{definition}\label{def:Wchart}
  Let $V,S \subset \CC^n$ be two algebraic sets, $\chi = (m,\hh)$ be a chart of
  $(V,S)$, with $\hh$ of length $c$ and $i \in \{1,\dotsc,n-c+1\}$. For every
  $c$-minor $m'$ of $\jac(\hh)$ and every $(c+i-1)$-minor $m''$ of
  $\jac(\hh,\map[i])$ containing the rows of $\jac(\map[i])$, we define
  $\Wchart(\chi,m',m'',\map[i])$ as the couple:
 \[
  \Wchart(\chi,m',m'',\map[i]) = \Big( \: mm'm'' , \: \big(\hh,\Hphi(\hh,i,m'') 
\big) \: \Big)
 \]
\end{definition}

Then, the definition of the associated atlas comes naturally. Let $V,S \subset
\CC^n$ be two algebraic sets with $V$ $d$-equidimensional, $\bchi = (\chi_j)_{1
  \leq j \leq s}$ be an atlas of $(V,S)$ (with $\chi_j = (m_j,\hh_j)$) and $i 
\in
\{1,\dotsc,d+1\}$. Since $V$ is $d$-equidimensional, by \cite[Lemma
A.12]{SS2017}, all the sequences of polynomials $\hh_j$ have same cardinality 
$c 
= n-d$.
\begin{definition}
  We define $\Watlas(\bchi,V,S,\bphi,i)$ as the sequence of all charts
  $\Wchart(\chi_j,m',m'',\map[i])$ for every $j \in \{1,\dotsc,s\}$, every
  $c$-minor $m'$ of $\jac(\hh_j)$ and every $(c+i-1)$-minor $m''$ of
  $\jac(\hh_j,\map[i])$ containing the rows of $\jac(\map[i])$, for which
  $\Ocal(m_jm'm'') \cap \Wphii - S$ is not empty.
\end{definition}

These constructions generalize the ones introduced in \cite[Section
  3.1]{SS2017} in the following sense: for $\bphi=(x_1, \ldots, x_{n+1})$,
except for some trivial cases, the objects we just defined match the
ones in \cite[Definition 3.1 to 3.3]{SS2017}, possibly up to signs
(which have no consequence). The next lemma makes this more precise; in
this lemma, we write $\proj = (x_1, \ldots, x_{n+1})$ and ${\proj[i]} =
(x_1, \ldots, x_i)$.
\begin{lemma}\label{lem:chartaltlasWproj}
  Let $V,S$ be algebraic sets and $\hh=(h_1,\dotsc,h_c)\subset \CC[\XX]$. Let
  $1\leq i \leq n-c+1$ and $m''$ be a $(c+i-1)$-minor of $\jac(\hh,\proj[i])$,
  containing the rows of $\jac(\proj[i])$. Then either $m''=0$ or
 \begin{enumerate}
  \item $\mu'' = (-1)^{i(c-1)} m''$ is a $(c-1)$-minor of $\jac(\hh,i)$;
  \item $\Hphi[\proj](\hh,i, m'')=(-1)^{ic} H$, where $H$ is the
    $(n-c-i+1)$-sequence of $c$-minors of $\jac(\hh,i)$ obtained by
    successively adding the missing row and the missing columns of
    $\jac(\hh,i)$ to $\mu''$;
  \item if $\chi=(\hh,m)$ is a chart of $(V,S)$, then for every $c$-minor $m'$ 
    of $\jac(\hh)$,
    \[
    \Wchart(\chi,m',m'') =\Wchart(\chi,m',(-1)^{i(c-1)} \mu'') 
    \]
    which is $\left(mm' m'',(\hh, (-1)^{ic} H)\right)$ with $H$ as above.
 \end{enumerate}
 Assume, in addition, that $V$ is $d$-equidimensional, with $d=n-c$. Let $\bchi 
=
 (\chi_j)_{1\leq j\leq s}$ be an atlas of $(V,S)$, with $\chi_j=(\hh_j,m_j)$, 
and
 let $c$ be the common cardinality of the $\hh_j$'s. Then
 \begin{enumerate}
 \setcounter{enumi}{3}
\item $\Watlas(\bchi,V,S,\proj,i)$ is the sequence of all 
  $\Wchart(\chi_j,m',(-1)^{i(c-1)} \mu'')$, for $j\in\{1,\dotsc,s\}$ and for 
$m',\mu''$
  respectively a $c$-minor of $\jac(\hh_j)$ and a $(c-1)$-minor of
  $\jac(\hh_j,i)$ for which $\Ocal(m_jm' \mu'')\cap\Wproji[i][V]-S$ is not 
empty.
 \end{enumerate}
\end{lemma}
\begin{proof}
  According to Lemma~\ref{lem:minorssubmat}, up to the sign
  $(-1)^{i(c-1)}$, the $(c-1)$-minors of $\jac(\hh,i)$ are exactly the
  $(i+c-1)$-minors of $\jac(\hh,\proj[i])$ containing the identity
  matrix $I_i=\jac(\proj[i])$, since
 \[
  \jac_{x_1,\dotsc,x_n}(\hh,\proj[i]) =
  \begin{bmatrix}
   \jac_{x_1,\dotsc,x_i}(\hh) & \jac_{x_{i+1},\dotsc,x_n}(\hh)\\
   I_i & \OO
  \end{bmatrix}.
 \]
 Since $m''$ contains the rows of $\jac(\proj[i])=[I_i~~\OO]$, either it 
actually contains
 $I_i$ or it is zero, as a zero row appears. We assume the first case; then,
 by the discussion above, $\mu'' = (-1)^{i(c-1)} m''$ is the determinant of a 
$(c-1)$-submatrix
 $M$ of $\jac(\hh,i)=\jac_{x_{i+1},\dotsc,x_n}(\hh)$. 
 
 The row and columns of $\jac(\hh,i)$ that are not in $M$ have
 respective indices $1\leq k \leq c$ and $1\leq \ell_1\leq\dotsc\leq
 \ell_{n-i-c+1}\leq n$. Since $m''$  contains $I_i$,
 the rows and columns of $\jac(\hh,\proj[i])$ that are not in $m''$ have
 respective indices $1\leq k'\leq c$ and $i+1\leq \ell_1'\leq \dotsc\leq
 \ell_{n-c-i+1}'\leq n$. Then, according to Lemma~\ref{lem:minorssubmat}, for
 all $1\leq j \leq n-c-i+1$,
 \[
  k=k' \et \ell_j = \ell_j' - i.
 \]
 Hence, by Lemma~\ref{lem:minorssubmat}, the $(c+i)$-minors obtained by 
 adding the missing row and the missing columns of $\jac(\hh,\proj[i])$ to
 the submatrix used to define $m''$ are exactly the $c$-minors of $\jac(\hh,i)$
 obtained by adding the missing row and the missing columns of $\jac(\hh,i)$ to
 $\mu''$, up to a factor $(-1)^{ic}$. This gives the second statement.
 The third statement is then nothing but the definition of 
$\Wchart(\chi,m,m'')$.

 Finally, consider an atlas $\bchi$ of $(V,S)$. By
 Lemma~\ref{lem:minorssubmat}, for $j\in\{1,\dotsc,s\}$, all
 $(c-1)$-minors $\mu''$ of $\jac(\hh_j,i)$ are, up to sign,
 $(c+i-1)$-minors of $\jac(\hh_j,\proj[i])$ built with the rows of
 $\jac(\proj[i])$. Conversely, let $j \in \{1,\dotsc,s\}$, $m'$ be a
 $c$-minor of $\jac(\hh_j)$ and let $m''$ be a $(c+i-1)$-minor of
 $\jac(\hh_j,\proj[i])$ containing the rows of $\jac(\proj[i])$. Then
 either $m''=0$, so that $\Ocal(m'')$ and then
 $\Ocal(m_jm'm'')\cap\Wproji[i][V]-S$ is empty, or $\mu''=(-1)^{i(c-1)}m''$ is 
a $(c-1)$-minor
 of $\jac(\hh_j,i)$.
  Hence, according to the third item, for $j\in\{1,\dotsc,s\}$ and any
 $c$-minor $m'$ of $\jac(\hh_j)$, the sequences of
 \begin{itemize}
 \item all $\Wchart(\chi_j,m',m'')$ for every $(c+i-1)$-minor $m''$ of
   $\jac(\hh,\proj[i])$ containing the rows of $\jac(\proj[i])$, for which
   $\Ocal(m_jm'm'') \cap \Wproji[i][V] - S$ is not empty, and
 \item all $\Wchart(\chi_j,m',(-1)^{i(c-1)}m'')$ for every $(c-1)$-minor $\mu'$
   of $\jac(\hh_j,i)$ for which $\Ocal(m_jm'm'')\cap\Wproji[i][V]-S$ is not
   empty,
 \end{itemize}
 are equal to $\Watlas(\bchi,V,S,\proj,i)$.
\end{proof}

We can now state the main result of this subsection, which we prove in
Section~\ref{sec:polargen}. This is, a generalization of \cite[Proposition 
3.4]{SS2017} which only deals with the case of projections.
\begin{proposition}\label{prop:polargen}
  Let $V,S \subset \CC^n$ be two algebraic sets with $V$ $d$-equidimen\-sio\-nal and
  $S$ finite, and let $\bchi$ be an atlas of $(V,S)$. For $2\leq \ifi\leq d+1$, let 
$\bpolun = (\polun_1,\dotsc,\polun_{\ifi})$ and $\bpolde = 
(\polde_1,\dotsc,\polde_{\ifi})$, and for $1\leq j \leq \ifi$, let
  $\balf_j =(\alf_{j,1}, \ldots, \alf_{j,n})\in \CC^{n}$ and
  \[
    \phi_j(\XX,\balf_j) = \polun_j(\XX) + \sum_{k=1}^n \alf_{j,k}x_k +
    \polde_j(\balf_j) \in \CC[\XX].
  \]
where $\polun_j \in \CC[\XX]$ and $\polde_j\colon\CC^n\to\CC$ is a polynomial 
map, with coefficients in $\CC$. 

  There exists a non-empty Zariski open subset 
$\ZOpolar(\bchi,V,S,\bpolun,\bpolde,\ifi) \subset \CC^{\ifi n}$ such that for 
every $\balf \in \ZOpolar(\bchi,V,S,\bpolun,\bpolde,\ifi)$, writing $\bphi = 
\left(\phi_1(\XX,\balf),\dotsc,\phi_{\ifi}(\XX,\balf)\right)$, the following 
holds. For $i$ in $\{1,\dots,\ifi\}$, either $\Wphii$ is empty or
 \begin{enumerate}
 \item $\Wphii$ is an equidimensional algebraic set of dimension $i-1$;
 \item if $2 \leq i \leq (d+3)/2$, then $\Watlas(\bchi,V,S,\bphi,i)$ is an atlas
   of $(\Wphii,S)$\\ and $\sing(\Wphii) \subset S$.
 \end{enumerate}
\end{proposition}

We end this subsection with a statement we use further for the proof of our 
main algorithm; it addresses the special case $i=2$.
\begin{proposition}\label{prop:critpolargen}
  Let $V\subset \CC^n$ be a $d$-equidimensional algebraic set with $d\geq 1$ 
and $\sing(V)$ finite. Let $\polun \in \CC[\XX]$, and for $i\in\{1,2\}$, let 
$\balf_i=(\alf_{i,1}, \ldots, \alf_{i,n})$ in $\CC^{n}$ and
 \begin{align*}
   \phi_1(\XX,\balf_1) = \polun(\XX) + 
   \sum_{k=1}^n \alf_{1,k}x_k  
   \et
   \phi_2(\XX,\balf_2) = \sum_{k=1}^n \alf_{2,k}x_k.
 \end{align*}
 Then, there exists a non-empty Zariski open subset
 $\ZOcritpolar(V,\polun) \subset \CC^{2n}$ such that for every $\balf = 
(\balf_1, \balf_2) \in \ZOcritpolar(V,\polun)$, and $\bphi = \left(\phi_1(\XX, 
\balf_1),\phi_{2}(\XX, \balf_2)\right)$, the following holds. Either 
$\Wphii[2]$ is empty or
 \begin{enumerate}
 \item $\Wphii[2]$ is $1$-equidimensional; 
 \item the sets $\WWophii[1][2]$, $\WWphii[1][2]$ and $\KWphii[1][2]$ are 
finite.
\end{enumerate}
\end{proposition}
\begin{proof} 
  Let $\bchi$ be an atlas of $(V,\sing(V))$, as obtained by applying \cite[Lemma
  A.13]{SS2017}. Let $\ZOcritpolar(V,\bpolun)$ be the intersection of the
  non-empty Zariski open subsets $\ZOnoether(V,\polun,2)$ and
  $\ZOpolar(\bchi,V,\sing(V),(\polun,0),(0),2)$ of $\CC^{2n}$, obtained by 
  applying Propositions~\ref{prop:fibergenalg} and \ref{prop:polargen} 
  respectively (recall that we assume $d\geq 1$).
  From now on, choose $\balf = (\balf_1, 
  \balf_2) \in \ZOcritpolar(V,\polun)$ and let $\bphi = \left(\phi_1(\XX, 
  \balf_1),\phi_{2}(\XX, \balf_2)\right)$. In the following, we denote 
  $\Wphii[2]$ by $W_2$. Suppose $W_2$ is non-empty, otherwise the result 
  trivially holds.
  
  Since $\balf \in \ZOpolar(\bchi,V,\sing(V),(\bpolun,0),(0),2)$ and $2\leq 
(d+3)/2$ for $d\geq 1$, then, by Proposition~\ref{prop:polargen}, $W_2$ is 
equidimensional of dimension $1$ and $\sing(W_2)\subset \sing(V)$ is finite. 
Hence, $K_W =\Wphii[1][W_2]$ is well defined and the following inclusion holds
\[
 K_W \subset \bigcup_{z \in \mapun(K_W)}  W_2 \cap \mapunrec(z)
\]
By an algebraic version of Sard's lemma from \cite[Proposition B.2]{SS2017}, we
deduce that $\mapun(\Wphii[1][W_2])$ is finite. Besides, since $\balf \in
\ZOnoether(V,\polun,2)$ then, by Proposition~\ref{prop:fibergenalg},
$\mapunrec(z) \cap W_2$ is finite for any $z \in \CC$.

Hence, as a set contained in a finite union of finite sets, $K_W$ is finite, and
so are $\Wophii[1][W_2]$ and $\Kphii[1][W_2] = K_W \cup \sing(W_2)$.
\end{proof}

\subsection{Charts and atlases for fibers of polynomial maps}
We now study the regularity and dimensions of fibers of some generic polynomial
maps over algebraic sets. The construction we introduce below is quite similar
to the one in \cite{SS2017}, but a bit more general.

\begin{definition}\label{def:Fatlas}
  Let $V,S \subset \CC^n$ be two algebraic sets with $V$ $d$-equidimensional, and let
  $\bchi = (\chi_j)_{1\leq j\leq s}$ be an atlas of $(V,S)$. Let $1\leq e \leq 
\ifi \leq n+1$ be integers and $\bphi=(\phi_1,\dotsc,\phi_{\ifi})\subset 
\CC[\XX]$. For $Q \subset \CC^{e}$ we define $\Fatlas(\bchi,V,Q,S,\bphi)$ as 
the sequence of all $\chi_j=(m_j,\hh_j)$ such that $\Ocal(m_j) \cap F_Q - S_Q$ 
is not empty, where
 \[
    F_Q = \mapfbr{V} \et S_Q = \mapfbr{\big(S\cup\Wphii[e]\big)}.
 \]
\end{definition}
The above definition is a direct generalization of \cite[Definition 
3.6]{SS2017}, where $\bphi=(x_1,\dotsc,x_n)$.
The main result of this subsection is the following proposition, which we prove
in Section~\ref{sec:fibergen}.
\begin{proposition}\label{prop:dimfiber}
 Let $V,S \subset \CC^n$ be two algebraic sets with $V$ $d$-equidimen\-sio\-nal and 
$S$ finite. Let $\bchi$ be an atlas of $(V,S)$.
 Let $2\leq \ifi \leq d+1$ and $\bphi=(\phi_1,\dotsc,\phi_{\ifi}) \subset 
\CC[\XX]$. For $2\leq j \leq d$, let $\balf_j = (\alf_{j,1}, \ldots, 
\alf_{j,n})\in \CC^{n}$ and 
 \begin{align*}
  \phi_1(\XX,\balf_1) = \polun(\XX) + 
  \sum_{k=1}^n \alf_{1,k}x_k
  \et
  \phi_j(\XX,\balf_j) = \sum_{k=1}^n \alf_{j,k}x_k 
 \end{align*}
 where $\polun \in \CC[\XX]$.

 There exists a non-empty Zariski open subset $\ZOfiber(\bchi,V,S,\polun,\ifi)
 \subset \CC^{\ifi n}$ such that for every $\balf = (\balf_1, \ldots, 
   \balf_{\ifi}) \in \ZOfiber(\bchi,V,S,\polun,\ifi)$ and writing
       \[\bphi = 
         \left(\phi_1(\XX, \balf_1), \dotsc,\phi_{\ifi}(\XX,
             \balf_{\ifi})\right),
           \] the following holds. Let $0 \leq e \leq d$,
           $Q \in \CC^e$ a finite subset and $F_Q$ and $S_Q$ be as in
           Definition~\ref{def:Fatlas}. Then either $F_Q$ is empty or
 \begin{enumerate}
 \item $S_Q$ is finite;
  \item $V_Q$ is an equidimensional algebraic set of dimension $d-e$;
  \item $\Fatlas(\bchi,V,Q,S,\bphi)$ is an atlas of $(F_Q,S_Q)$ and $\sing(F_Q) 
\subset S_Q$.
 \end{enumerate}
\end{proposition}
 \section{The algorithm}

\subsection{Overall description}

Recall that $\XX$ denotes a sequence of $n$ indeterminates $x_1,\dotsc,x_n$. 
In this document, we also consider a family $\AA = (a_{i,j})_{1\leq i,j \leq 
n}$ of $n^2$ new indeterminates, which stand for generic parameters. For $1\leq 
i,j \leq n$, we note $a_i = (a_{i,1},\dotsc,a_{i,n})$, so that 
$\AA_{\leq i}$ represents the subfamily $(a_1,\dotsc,a_i)$. An element 
$\balf 
\in \CC^{in}$ will often be represented as a vector of length $i$ of the form 
$(\balf_1,\dotsc,\balf_i)$, with all $\balf_j = 
(\alf_{j,1},\dotsc,\alf_{j,n}) 
\in \CC^n$. 

Then, as suggested by Propositions~\ref{prop:fibergenalg}, 
\ref{prop:polargen}, \ref{prop:critpolargen} and \ref{prop:dimfiber}, we will 
consider polynomials of the form:
 \begin{equation} \label{eqn:defvphi}
  \vphi_i(\XX,a_i) = \polun_i(\XX) + 
  \sum_{j=1}^n a_{i,j}x_j  +
  \polde_i(a_i) \in \RR[\XX,\AA].
 \end{equation}
where $1\leq i \leq n$, $\polun_i \in \RR[\XX]$ and $\polde_i \in \RR[\AA]$.
We will choose $\polun_i$ so that the polynomial map $\vphi_i$ inherits some 
useful properties. For instance, taking $\polun_i = x_1^2 + \cdots +x_n^2$, for
any $\balpha_i$ in $\RR^n$, the polynomial map associated to 
$\vphi_i(\XX,\balpha_i)$ is proper and bounded from below on $\RR^n$.

Hereafter, we describe, on an example, the core idea of the the strategy that 
guided the design of our algorithm and the choice of data structures.
\begin{example}\label{exa:algo}
Consider the algebraic set $V=\V(f) \subset \CCo^4$ defined as the vanishing 
locus of the polynomial 
\[
  f = \sum_{i=1}^4(x_i^3 - x_i) - 1 \in 
\QQo[x_1,x_2,x_3,x_4].
\]

We want to compute a roadmap of $(V,\emptyset)$ (or simply $V$). 
Following the strategy we designed in our previous work (see \cite[Section 
5]{PSS2024}), $V$ must satisfy some regularity properties, that is
\begin{enumerate}[label=$(\sfH_{\arabic*})$]
 \item $V$ is $d$-equidimensional, $d\geq 2$, and $\sing(V)$ is finite.
\end{enumerate}
The first part of the assumption can be satisfied by computing an 
equidimen\-sio\-nal decomposition of $V$, which can be done within the complexity bounds 
considered in this work (see e.g. \cite{Le2003} for the best-known 
complexity bound for a probabilistic algorithm). However, it is worth noting 
that this increases the degrees of the generators.

The condition $d\geq2$ is not restrictive, as the case $d=1$ is trivial 
for roadmap computations.
The smoothness assumption is more restrictive. Indeed, it can be satisfied 
using deformation techniques, such as done in \cite{BR2014,BRSS2014}, but these 
steps would not fit, as such, in our complexity bounds. 

Let us check that, in our example, $V$ satisfies $\sfH_1$. We will
describe further a subroutine \SingPts, to compute $\sing(V)$ as long
as this holds.
\begin{enumerate}[label=Checking $(\sfH_{\arabic*}).$]
 \item As an hypersurface, $V$ is irreducible, and then equidimensional, of 
dimension $3$.
The partial derivatives of $f$, $\frac{\partial f}{\partial x_i} = 3x_i^2-1$, 
for $1\leq i\leq 4$, do not simultaneously vanish on $V$.
Hence, $\sing(V)=\emptyset$, and $V$ satisfies assumption $(\sfH_1)$.
\end{enumerate}

Following a particular case of \cite[Section 5]{PSS2024}, we want to 
choose a sequence of polynomial $\map=(\map[1],\dotsc,\map[n])$ in 
$\QQ[\XX]$ such that the following holds:
\begin{enumerate}[label=$(\sfH_{\arabic*})$]
 \setcounter{enumi}{1}
 \item the restriction of $\map[1]$ to $V\cap\RRo^n$ is proper and bounded 
below;
 \item $W_2=\Wphii[2]$ is $1$ equidimensional, and smooth outside $\sing(V)$;
 \item for any $z \in \CC$, $\mapfbreq{V}[z][1]$ is 
$(d-1)$-equidimensional;
 \item $\Kphii[1][W_2]$ is finite. 
\end{enumerate}
In addition, let $K = \Kphii[1][W_2] \cup \sing(V)$ and $F = V \cap 
\mapun[-1](\mapun(K))$. We require that
\begin{enumerate}[label=$(\sfH_{\arabic*})$]
 \setcounter{enumi}{5}
 \item $\Pcal_W=F \cap W_2$ is finite.
\end{enumerate}
Then, under the above assumptions  if $\Rcal_F$ 
is a roadmap of $(F,\Pcal_W)$, then $W_2 \cup \Rcal_F$ is a roadmap of $V$.
This statement is a consequence of both \cite[Proposition 2]{SS2011} and 
\cite[Theorem 1.1]{PSS2024}, and will be properly stated and proved in 
Proposition~\ref{prop:connectresult}. This splits the problem of computing a 
roadmap of $V$ into the computation of representations of $W_2$, 
$F$ and $\Pcal_W$, and a roadmap of $(F,\Pcal_W)$.
Since $F \cap\RRo^n$ is bounded, by assumption $(\sfH_2)$, the latter 
computation can be 
done using the algorithm of \cite{SS2017}.

We describe this process more precisely with our example. Each 
step consisting in checking the assumptions, and computing the associated 
objects.
\begin{enumerate}[label=Checking $(\sfH_{\arabic*})$.]
\item[Checking $(\sfH_{2/3})$.] Set first $\map =
  \left(\sum_{i=1}^4 x_i^2,\; x_2,\; x_3,\; x_4\right)$. The
  restriction of $\mapun$ to $V\cap\RR^4$ is proper and non-negative.
  We could then compute a representation of $W_2=\Wphii[2]$, before computing
  one for its singular locus $\sing(W_2)$.  However, the latter
  singular set is not empty, while $\sing(V)$ is.  This contradicts
  the assumptions needed in \cite[Theorem 1.1]{PSS2024} and the
  strategy for computing a roadmap of $V$ designed in \cite[Section
    5]{PSS2024} might fail.

  Following Propositions~\ref{prop:fibergenalg}, \ref{prop:polargen},
  \ref{prop:critpolargen} and \ref{prop:dimfiber} from the
preliminaries, we propose the following.  To prevent these regularity
failures, and to satisfy all assumptions of \cite[Theorem
  1.1]{PSS2024}, while keeping the properties of $\bphi$, we add to
$\mapun$ a linear form; here we take $x_1-x_4$, but in general it should be taken with
random coefficients.

Hence, consider now the sequence $\map$ of polynomials maps
\[
 \map = \left(\sum_{i=1}^4 x_i^2 + x_1 - x_4,\; x_2,\; x_3,\; x_4\right),
\]
whose restriction to $\RR^4$ is still proper and bounded below, by 
construction.
If the linear form we added has been sufficiently randomly chosen, 
Proposition~\ref{prop:polargen} claims that $W_2$ satisfies assumption 
$(\sfH_3)$.

Using Gr\"obner basis computations on a determinantal ideal defining $W_2$, we 
compute a representation of $W_2$, and next $\sing(W_2)$, that turns out to be 
empty, as requested.
More generally, computing the two previous sets efficiently is the purpose of 
the algorithm \SolvePolar, presented further in Lemma~\ref{lem:solvepolar}.
\setenumi{3}
\item By Proposition~\ref{prop:dimfiber}, this
  assumption holds if we have added to $\map$ a linear form that is
  generic enough.  Using the Jacobian criterion, we can check that in
  our case, for any $z \in \CCo$, the fiber $F_{z}=V\cap\mapunrec(z)$
  is an equidimensional algebraic set of dimension $2$ (if it is not
  empty). Moreover the singular locus of $F_{z}$ is contained in the
  finite set $\Wphii[1][V]$. Computing the latter set is tackled by
  the subroutine \Crit, presented in Lemma~\ref{lem:crit}.

\item We also need to check the finiteness and compute the set 
$\Kphii[1][W_2]$.
If $\map$ is generic enough, the finiteness is ensured by
Proposition~\ref{prop:critpolargen}; computing this set is the purpose
of the algorithm \CritPolar, presented in Lemma~\ref{lem:fiberpolar}.
In our case, there are finitely many (more precisely $129$) such
points, and $23$ of them are real.

\item We need to compute the set
  $K=\Kphii[1][W_2]\cup\sing(V)$. As the two members of the unions
  have been computed by the algorithms \CritPolar and \SingPts,
  respectively, one can compute this union using the procedure \Union
  from \cite[Lemma J.3]{SS2017} (also presented in the next
  subsection).

  Then, for $\map$ generic enough, Proposition~\ref{prop:fibergenalg} 
ensures that the last assumption holds.
The computation of $\Pcal_W$ boils down to computing finitely many fibers on 
the restriction of $\mapun$ to $W_2$. This is the purpose of the algorithm 
\FiberPolar, presented in Lemma~\ref{lem:fiberpolar}.
\end{enumerate}

At this point, we have computed representations of $W_2$ and $\Pcal_W$, and 
ensured that all assumptions of \cite[Theorem 1.1]{PSS2024} are satisfied.
Hence,  one only need to compute a roadmap of 
$(F,\Pcal_W)$.
This is the purpose of algorithm \RMBound, presented in 
Proposition~\ref{prop:comp-broadmap}.

\end{example}

\subsection{Subroutines}\label{ssec:subroutinesintro}
Our main algorithm (Algorithm~\ref{alg:algormunbound}) makes use of
several subroutines which allow us to manipulate zero-dimensional and
one-dimensional parametrizations, polar varieties and fibers of polynomial
maps in order to make \cite[Theorem 1.1]{PSS2024}  effective.

As a reminder, in this document, we manipulate subroutines that
involve selecting suitable parameters in $\QQ^i$, for various $i\geq
1$.  These algorithms are probabilistic, which means that for any
choice of (say) $i$ parameters we have to do, there exists a non-zero
polynomial $\Delta$, such that for $\blambda \in \QQ^i$, success is
achieved if $\Delta(\blambda)\neq 0$.  It is also important to
note that these algorithms are considered Monte Carlo, as their
output's correctness cannot be guaranteed within a reasonable
complexity. In certain cases, where we can identify errors, we require
our procedures to output \fail. However, not returning \fail
does not guarantee correctness.

Let $1 \leq c\leq n$, and $\ff=(f_1,\dotsc,f_c)$ be a sequence of polynomials 
in $\RR[\XX]$. We say that $\ff$ satisfies assumption \ref{ass:A} if
\begin{enumerate}[label=$(\sfA)$]
 \item\label{ass:A} $\ff$ is a reduced regular sequence, with $d=n-c\geq 2$, and
$\sing(\V(\ff))$ is finite.
\end{enumerate}
In particular, the zero-set of $\ff$ in $\CC^n$ is then either empty or 
$d$-equidimensional.

\subsubsection{Basic subroutines}

The first two subroutines we use are described in \cite{SS2017} and are used to
compute $\sing(V(\ff))$ (on input a straight-line program evaluating $\ff$) and
to compute a rational parametrization encoding the union of zero-dimensional
sets or the union of algebraic curves. They are both Monte Carlo algorithms, in 
the sense described above, and can output \fail in case errors have been 
detected during the execution. However, in case of success, the following holds.
\begin{itemize}
\item \SingPts, described in \cite[Section 
J.5.4]{SS2017}, takes as input a straight-line program $\Gamma$ that evaluates
polynomials $\ff \in \CC[\XX]$ satisfying assumption \ref{ass:A}  and outputs a
\ZDP of $\sing(\V(\ff))$.

\item \Union, described in \cite[Lemma J.3]{SS2017} (resp.\
\cite[Lemma J.8]{SS2017}), takes as input two zero-dimensional (resp.\
one-dimensional) parametrizations $\scrP_1$ and $\scrP_2$ and outputs a 
zero-dimensional (resp.\ one-dimensional) parametrization encoding 
$\Zparam(\scrP_1)\cup\Zparam(\scrP_2)$.
\end{itemize}

We now describe basic subroutines performing elementary operations on 
\SLP and \ZDPs. The first one allows us to generate a generic polynomial with a 
prescribed structure.
\begin{lemma}\label{lem:phigen}
Let $1\leq i \leq n$ and $\balf=(\balf_1,\dotsc,\balf_i) \in \CC^{in}$. Then 
there 
exists an algorithm \SLPPhi which takes as input $\balf$ and returns in time 
$O(n)$ a 
\SLP $\Gammaphi$ of length $O(n)$
computing in $\QQ[\XX]$:
 \begin{align*}
   \phi_1 = \sum_{k=1}^n x_k^2 + \alf_{1,k}x_k  
   \et
   \phi_j = \sum_{k=1}^n \alf_{j,k}x_k \text{\quad for $2\leq j \leq i$}.
 \end{align*}
\end{lemma}

\begin{proof}
Given the constants $\balf$, it suffices
  to generate a straight-line program that computes the linear forms
  $\sum_{k=1}^n \balf_{j,k}x_k$, for $j=1,\dots,i$, and adds the
  quadratic form $\sum_{k=1}^n x_k^2$ to the first one. This takes
  linear time, and the result straight-line program has linear length.
 \end{proof}

Next, we present a procedure computing the image of a \ZDP by a polynomial 
map, given as a \SLP, generalizing the subroutine \textsf{Projection} from 
\cite[Lemma J.5]{SS2017}. The proof of the next lemma is given in
Subsection~\ref{ssec:basicsubroutines}.
\begin{lemma}\label{lem:imageZDP}
  Let $\scrP$ be a \ZDP of degree $\kappa$ such that 
  $\Zparam(\scrP)\subset\CC^n$ and let $\Gammaphi$ be a \SLP of length $E'$ 
  computing polynomials $\bphi=(\phi_1,\dotsc,\phi_i)$.
  There exists a Monte Carlo algorithm \Image which, on input $\Gammaphi$, 
  $\scrP$ and $j \in \{1,\dotsc, i\}$, outputs either \fail or a \ZDP $\scrQ$, 
  of degree at most $\kappa$, using
  \[
    \Otilde\left((n^2\kappa + E')\kappa\right)
  \]
  operations in $\QQ$. In case of success, $\Zparam(\scrQ) 
  =\map[j](\Zparam(\scrP))$.
\end{lemma}

\subsubsection{Subroutines for polar varieties}\label{ssec:subpolarintro}
The next subroutines are used to compute generalized polar varieties and 
quantities related to them. The proof of all statements below can be found in 
Subsection~\ref{ssec:computepolar}.
In this subsection, we fix $1\leq c \leq n-2$ and we refer to the following
objects:
\begin{itemize}
 \item sequences of polynomials $\gg =(g_1,\dotsc,g_c)$ and 
$\bphi=(\phi_1,\phi_2)$ all of them in $\QQ[\XXi]$, of degrees bounded by $D$, such that 
$\gg$ satisfies assumption \ref{ass:A}; we note $d=n-c$;

 \item \SLPs $\Gamma$ and $\Gammaphi$, of respective lengths $E$ and $E'$, 
computing respectively $\gg$ and $\bphi$;

 \item \ZDPs $\scrS$ and $\scrQ''$, of respective degrees $\sigma$ and 
$\kappa''$, describing finite sets $S\subset \CC^n$ and $Q'' \subset 
\CC$, such that $\sing(\V(\gg)) \subset S$ (the ${}''$ superscripts we use 
here match those used in the algorithm);

 \item an atlas $\bchi$ of $(\V(\gg),S)$, given by 
\cite[Lemma A.13]{SS2017}, as $S$ is finite and contains $\sing(\V(\gg))$.
\end{itemize}

We start with the subroutine \Crit, which is used for computing critical and 
singular points of some polynomial map, again under some regularity assumption. 
These critical points are nothing but zero-dimensional polar varieties.
\begin{lemma}\label{lem:crit}

  Assume that $\Kphii[1][\V(\gg)]$ is finite.
  There exists a Monte Carlo algorithm \Crit which takes as input $\Gamma$,
  $\Gammaphi$ and $\scrS$, and outputs either \fail or a
  \ZDP $\scrS_F$, \wcoeffQ, of degree at most 
  \[
   \binom{n+1}{d} D^{c+2}(D-1)^{d} + \sigma
  \]
 such that, in case of success, $\Zparam(\scrS_F) = \Kphii[1][\V(\gg)]\cup S$.
 It uses
  \[
    \softOh{E''(n+2)^{4d+8}D^{2n+3}(D-1)^{2d}+n\sigma^2\,}
  \]
  operations in $\QQ$, where $E''=E+E'$.
\end{lemma}

We now tackle higher dimensional cases, with the subroutine \SolvePolar which, 
under some assumptions, computes \ODP encoding one-dimensional generalized 
polar varieties.

\begin{lemma}\label{lem:solvepolar}
Let $W=\Wphii[2][\V(\gg)]$ and assume that one of the following holds
  \begin{itemize}
  \item $W$ is empty, or
  \item $W$ is 1-equidimensional, with $\sing(W)\subset S$, and 
    $\Watlas(\bchi,\V(\gg),S,\bphi,2)$ is an atlas of $(W,S)$.
  \end{itemize}
  Then, there exists a Monte Carlo probabilistic algorithm \SolvePolar which 
  takes as input $\Gamma$, $\Gammaphi$ and $\scrS$ and which outputs
  either \fail or a \ODP $\scrW_2$, \wcoeffQ, of degree at most
  \[
    \delta=\deltavalue,
  \]
  such that, in case of success, $\Zparam(\scrW_2) = W$. It uses at most
  \[
    \softOh{(n+c)^3(E''+(n+c)^3)D\delta^3+(n+c)\delta\sigma^2\,}
  \]
  operations in $\QQ$, where $E''=E+E'$.
\end{lemma}

The subroutine \CritPolar is devoted to compute critical points of the
restriction of some polynomial map to a generalized polar variety of dimension 
at most one. It generalizes the subroutine $\textsf{W}_1$ from 
\cite[Proposition 6.4]{SS2017}.

\begin{lemma}\label{lem:critpolar}

  Let $W=\Wphii[2][\V(\gg)]$ and assume that either $W$ is empty, or
  \begin{itemize}
  \item $W$ is 1-equidimensional, with $\sing(W)\subset S$, and 
    $\Watlas(\bchi,\V(\gg),S,\bphi,2)$ is an atlas of $(W,S)$,
  \item and $\Wphii[1][W]$ is finite.
  \end{itemize}

  There exists a Monte Carlo algorithm \CritPolar which takes as input $\Gamma$,
  $\Gammaphi$ and $\scrS$ and which outputs either \fail or a
  \ZDP $\scrK$, \wcoeffQ, such that $\Zparam(\scrK) = \Wphii[1][W]\cup S$ using 
  at most
  \[
    \softOh{(n+c)^{12}E''D^3\delta^2+(n+c)\sigma^2\,}
  \]
  operations in $\QQ$, where $E''=E+E'$, and $\delta=\deltavalue$.
  Moreover $\scrK$
  has degree at most $\delta(n+c+4)D+\sigma$.
\end{lemma}

Finally, we consider the subroutine \FiberPolar which, given polynomials 
defining a generalized polar variety of dimension at most one, the polynomial 
map $\bphi$
and a description of $Q''$, computes the fibers of the polynomial map $\bphi$ 
over $Q''$
on the polar variety. 
\begin{lemma}\label{lem:fiberpolar}

  Let $W=\Wphii[2][\V(\gg)]$ and assume that either $W$ is empty, or
  \begin{itemize}
  \item $W$ is 1-equidimensional, with $\sing(W)\subset S$, and 
    $\Watlas(\bchi,\V(\gg),S,\bphi,2)$ is an atlas of $(W,S)$;
  \item $W \cap \mapunrec(Q'')$ is finite.
  \end{itemize}
  
  There exists a Monte Carlo algorithm \FiberPolar which takes as input
  $\Gamma$, $\Gammaphi$, $\scrS$ and $\scrQ''$ and which outputs
  either \fail or a \ZDP $\scrQ$, \wcoeffQ, such that
  $\Zparam(\scrQ) = \left (W \cap \mapunrec(Q'')\right )\cup S$, using at most
  \[
    \softOh{(n+c)^4\big(E''+(n+c)^2\big)D\kappa''^2\delta^2 + 
      (n+c)\sigma^2\,}
  \]
  operations in $\QQ$, where $E''=E+E'$, and $\delta=\deltavalue$.
  Moreover, $\scrQ$
has degree at most $\kappa''\delta + \sigma$.
\end{lemma}

\subsubsection{Subroutines for computing roadmaps in the bounded case}
As seen above, in Example~\ref{exa:algo}, we are ultimately led to
compute a roadmap for a bounded real algebraic set (this set is given
as fibers over finitely many algebraic points of the restriction of a
polynomial map to our input).  To do so, we call the algorithm
\RMRecLag from \cite{SS2017}, which internally uses similar techniques
but with projections (where $\map=\proj$). The description and the
complexity analysis of this procedure are given in
Subsection~\ref{ssec:boundedrmp}. The subtlety comes from the fact
that, in \cite{SS2017}, the correction and complexity estimate of
\RMRecLag are given for an input consisting of polynomials $\ff$
defining an algebraic set $V=\V(\ff)$; here, we need an algorithm that
works for an input given as fibers of a polynomial map. More
precisely, we prove the following result in
Subsection~\ref{ssec:boundedrmp}.

\begin{proposition}\label{prop:comp-broadmap}
  Let $\Gamma$ and $\Gammaphi$ be \SLPs, of respective length $E$ and $E'$,
  computing respec\-ti\-ve\-ly sequences of polynomials $\gg = (g_1, \ldots,
  g_p)$ and $\map=(\phi_1,\dotsc,\phi_n)$ in $\QQ[x_1, \ldots, x_n]$, of degrees
  bounded by $D$. Assume that $\gg$ satisfies \ref{ass:A}. Let $\scrQ$ and
  $\scrS_Q$ be \ZDPs of respective degrees $\degQ$ and $\degS$ that encode
  finite sets $Q \subset \CC^{e}$ (for some $0<e\leq n$) and $S_Q \subset
  \CC^n$, respectively. Let $V = \V(\gg)$ and $F_Q = \mapfbr{V}[Q][e]$, and
  assume that
  \begin{itemize}
   \item $F_Q$ is equidimensional of dimension $d-e$, where $d=n-p$;
  \item $\Fatlas(\bchi,V,Q,\bphi)$ is an atlas of 
$(F_Q,S_Q)$, and $\sing(F_Q) \subset S_Q$;
  \item the real algebraic set $F_Q\cap 
\RR^n$ is bounded.
  \end{itemize}
  Consider additionally a zero-dimensional parametrization $\scrP$ of
  degree $\degP$ encoding a finite subset  $\Pcal$ of $F_Q$, which contains 
$S_Q$. Assume that $\degS \leq ((n+e)D)^{n+e}$.
  
  There exists a probabilistic algorithm $\RMBound$ which takes as input
  $((\Gamma, \Gammaphi, \scrQ,\scrS), \scrP)$ and which, in 
case of success, outputs a roadmap of 
  $(F_Q, \Pcal)$,
  of degree
  \[\softOh{
    (\degP + \degQ) 16^{3d_F}
    (n_F\logde{n_F})^{2(2d_F + 12\logde{d_F}) (\logde{d_F} + 5)}
    D^{(2n_F+1)(\logde{d_F} + 3)}}
\]
 where $n_F = n+e$ and $d_F = d-e$, and using
\[
\softOh{
    \mu'^3 16^{9d_F} E''
(n_F\logde{n_F})^{6(2d_F + 12\logde{d_F}) (\logde{d_F} + 6)}
D^{(6n_F+3)(\logde{d_F} + 4)}
}
\]
arithmetic operations in $\QQ$, where $\mu' = \degP + \degQ$ and $E'' = E+E'+e$.
\end{proposition}

\subsection{Description of the main algorithm}\label{ssec:algodescription}

Now, we describe the main algorithm to compute roadmaps of smooth
unbounded real algebraic sets.  In addition to the subroutines
mentioned above, we define \Rand as a procedure that takes as input a
set $X$ and returns a random element in $X$.
Together with \SLPPhi, it allows us to generate ``generic enough'' polynomial 
maps so that the results of the previous section do apply 
(Propositions~\ref{prop:fibergenalg}, \ref{prop:polargen}, 
\ref{prop:critpolargen} and~\ref{prop:dimfiber}).

\begin{algorithm}[h]
 \setstretch{1.2}
 \caption{Roadmap algorithm for smooth unbounded real algebraic 
sets.}\label{alg:algormunbound}
 \begin{algorithmic}[1] \newcommand{\alignalg}[1]{\mathrlap{#1}\hphantom{\scrP_F}}
  \Statex \hspace*{-0.7cm}
  \begin{minipage}[t]{0.055\linewidth}
   \textbf{Input:} 
  \end{minipage}
  \begin{minipage}[t]{.98\linewidth}
  \begin{itemize}[label=$\rhd\!$]
  \item a straight-line program $\Gamma$ that evaluates polynomi\-als
   $\ff=(f_1,\dotsc,f_c) \subset\QQ[\XX]$, satisfying assumption \ref{ass:A}; we
  note $V = \V(\ff)$;\\[-2em]
\item a \ZDP $\scrP_0$ encoding a finite set $\Pcal_0\subset V$.\\[-0.5em]
  \end{itemize}
  \end{minipage}
  
  \Ensure a one-dimensional parametrization $\scrR$ encoding a roadmap of 
$(V,\Pcal_0)$.
 \Statex \vspace*{-0.5em}
 \State\label{step:sing} 
 $\alignalg{\scrS}\gets\SingPts(\Gamma)$;
 \Comment{$\Zparam(\scrS)=\sing(V)$;}
 
 \State\label{step:union} 
 $\alignalg{\scrP}\gets \Union(\scrP_0,\scrS);$
\Comment{$\Pcal := \Zparam(\scrP) = \Pcal_0\cup\sing(V)$}
  \State\label{step:randa} $\alignalg{\balf} \gets\Rand(\QQ^{2n})$;
  
  \State\label{step:phigen} 
  $\alignalg{\Gammaphi} \gets\SLPPhi(\balf)$;
  \Comment{$\Gammaphi$ computes $\phi=\left(||\XX||^2 + \pscal{\balf_1,\XX}, 
\pscal{\balf_2,\XX}\right)$}
  
  \State\label{step:solvepolar} 
  $\alignalg{\scrW_2} \gets \SolvePolar(\Gamma,\Gammaphi,\scrS);$
  \Comment{$W_2 := \Zparam(\scrW_2) = \Wphii[2]$;}

  \State\label{step:critpolar}
  $\alignalg{\scrK}\gets \CritPolar(\Gamma,\Gammaphi,\scrP)$;
  \Comment{$K:=\Zparam(\scrK)=\Wphii[1][W_2]\cup \Pcal_0\cup\sing(V)$;}
  
  \State\label{step:image} 
$\alignalg{\scrQ}\gets\Image(\Gammaphi,1,\scrK)$;
  \Comment{$Q:=\Zparam(\scrQ)=\mapun(K)$;}
  
  \State\label{step:fiberpolar} $\scrP_F\gets\FiberPolar(\Gamma, \Gammaphi, 
\scrQ, \scrP)$;
  \Comment{$\Zparam(\scrP_F)= \left[W_2 \cup \Pcal_0\cup \sing(V)\right]  
\cap\mapunrec(Q)$;}
  
  \State\label{step:crit} $\scrS_F \gets \Crit(\Gamma,\Gammaphi,\scrS)$
  \Comment{$\Zparam(\scrS_F) = \Kphii[1]$;}
  
  \State\label{step:rmbound} $\scrR_F 
  \gets\RMBound\big((\Gamma,\Gammaphi,\scrQ,\scrS_F),\scrP_F\big)$
  \Statex
  \Comment{$\Zparam(\scrR_F)$ is a roadmap of 
$\left(V\cap\mapunrec(Q),\Zparam(\scrP_F)\right)$;}

  \State\label{step:return} return \textsf{Union}$(\scrW_2, \scrR_F)$
  \Comment{$W_2 \cup \Zparam(\scrR_F)$ is a roadmap of $(V,\Pcal_0)$.}
 \end{algorithmic}
\end{algorithm}

\subsection{Correctness and complexity estimate}
This subsection is devoted to the proof of the following theorem, which 
directly implies Theorem~\ref{thm:main}.

\begin{theorem}\label{thm:corrcomp}
Let $\Gamma$ be a \SLP of length $E$ evaluating polynomials
$\ff=(f_1,\dotsc,f_c)$ of degrees bounded by $D\geq2$, satisfying \ref{ass:A}. Let
$\scrP_0$ be a \ZDP of degree $\mu$ encoding a finite subset of $\V(\ff)\subset
\CC^n$.
Then there exists a \NEZO $\ZOpen \subset \CC^{2n}$ such that the following 
holds.

Let $\balf \in \QQ^{2n}$ the vector randomly chosen in the execution of 
Algorithm~\ref{alg:algormunbound}, then if $\balf \in \ZOpen$, and if the 
calls to the subroutines
\begin{center}\normalfont
 \hfill\SingPts,\hfill \Union,\hfill \SolvePolar,\hfill \CritPolar,\hfill 
\Image,\hfill \FiberPolar,\hfill \Crit \hfill \text{and} \hfill \RMBound\hfill
\end{center}
are successful then, on inputs $\Gamma$ and 
$\scrP_0$, Algorithm~\ref{alg:algormunbound} either returns 
a \ODP of degree
\[
  \overalldegree
\]
using
\[
  \overallcomplexity
\]
arithmetic operations in $\QQ$, with $d=n-c$.

In case of success, its output describes a
roadmap of $(\V(\ff),\Zparam(\scrP_0))$.
\end{theorem}

The correctness of Algorithm~\ref{alg:algormunbound} relies mainly on the
conjunction of \cite[Theorem 1.1]{PSS2024} and \cite[Proposition 2]{SS2011},
that form the following statement, with slightly stronger assumptions, which
hold in our context.
\begin{proposition}\label{prop:connectresult}
Let $V\subset\CC^n$ be a $\QQ$-algebraic set of dimension $d\geq 2$, and let
$\Pcal_0$ be a finite subset of $V$.
Let $\phi = (\phi_1,\phi_2) \subset \RR[\XX]$ and $W = \Wphii[2]$. Suppose that 
the following holds:
\begin{enumerate}[label=$\mathsf{(H_\arabic*)}$]
 \item $V$ is equidimensional and $\sing(V)$ is finite;
 \item the restriction of $\bphi_1$ to $\V\cap\RR^n$ is a proper map bounded 
from below;
 \item $W$ is either empty or $1$-equidimensional and smooth outside 
$\sing(V)$;
 \item for any $\yy \in \CC^2$, the set $V \cap \mapun[-1](\yy)$ is either 
empty or $(d-1)$-equidi\-men\-sio\-nal;
 \item $\Kphii[1][W]$ is finite.
 \end{enumerate}
Let further $K = \Kphii[1][W] \cup \Pcal_0 \cup \sing(V)$ and $F = V \cap 
\mapun[-1](\mapun(K))$. \\
Assume in addition that
 \begin{enumerate}
 \item[$\mathsf{(H_6)}$] $\Pcal_W=F \cap W$ is finite.
\end{enumerate}
If $\Rcal_F$ is a roadmap of $(F,\Pcal_0 \cup \Pcal_W)$, then $W \cup \Rcal_F$
is a roadmap of $(V,\Pcal_0)$.
\end{proposition}
\begin{proof}
 Remark first that the so-called assumptions $\sfA$, $\sfP$ and $\sfB$ from the 
connectivity result from \cite[Theorem 1.1]{PSS2024} are direct consequences of 
assumptions $\sfH_1$ to $\sfH_4$. 
Besides, $\Wphii[1] \subset \Kphii[1][W]$ and $\sing(W) \subset \sing(V)$, by 
\cite[Lemma A.5]{SS2017} together with assumption $\sfH_3$.
Hence, one can write
 \[
  K = \Wphii[1] \cup S \cup \sing(V).
 \]
 where $S = \Wphii[1][W] \cup \Pcal_0$. By $\sfH_5$, $S$ is a finite subset of
 $V$, that intersects every \SACC of $\Wphii[1][W] \cap \RR^n$ by definition.
 Hence, $S$ satisfies assumption $\sfC$ of \cite[Theorem 1.1]{PSS2024}. By
 application of this latter result, $W\cup F$ has then a non-empty and \SAC
 intersection with every \SACC of $V\cap\RR^n$ and it contains $\Pcal_0$ by
 construction.
 
Moreover, by $\sfH_6$, $F\cap W$ is finite, so that by \cite[Proposition 
2]{SS2011}, the following holds. 
If $\Rcal_W$ and $\Rcal_F$ are roadmaps of respectively $(W,\Pcal_0 \cup 
\Pcal_W)$ and $(F,\Pcal_0 \cup \Pcal_W)$, then $\Rcal_W \cup \Rcal_F$ is a 
roadmap of $(V,\Pcal_0)$.
But remark that $W$ is a roadmap of $(W,\Pcal_W)$ since $W$ has dimension one.
Besides, \cite[Proposition 2]{SS2011} can be slightly generalized as only one 
of $\Rcal_W$ or $\Rcal_F$ must contain $\Pcal_0$.
Hence, taking $\Rcal_W=W$ allows us to conclude.
\end{proof}

\begin{proof}[Proof of Theorem~\ref{thm:corrcomp}]~ Let $\Gamma$ and $\scrP_0$
  be the inputs of Algorithm~\ref{alg:algormunbound} and assume that
  $\Gamma$ evaluates polynomials $\ff = (f_1, \ldots, f_c)$ satisfying
  assumption \ref{ass:A}. Let $V = \V(\ff)$ and $\Pcal_0 = \Zparam(\scrP_0)$.

  Recall that we assume all calls to the subroutines \SingPts, \Union,
  \SolvePolar, \CritPolar, \Image, \FiberPolar, \Crit and \RMBound do succeed.
  
\paragraph*{Steps~\ref{step:sing}\,-\ref{step:union}}
By \cite[Proposition J.35]{SS2017}, the procedure \SingPts outputs a
\ZDP $\scrS$ describing $\sing(V)$ using $\Otilde(ED^{4n+1})$ operations in
$\QQ$. By \cite[Proposition I.1]{SS2017} (or \cite[Proposition 3]{SS2018})
$\scrS$ has degree at most
\[
  \deltaS = \binom{n-1}{c-1} D^{c}(D-1)^{d} = \binom{n-1}{d} D^{c}(D-1)^{d}
  \in O(n^dD^n)
\]
Then, according to \cite[Lemma J.3]{SS2017} and our assumptions, the procedure
\Union outputs a \ZDP $\scrP$ of degree at most
\[
  \deltaP=\mu+\deltaS \inOh O(\mu + n^dD^n)\text{,\quad using\quad 
}\Otilde(n(\mu^2+n^{2d}D^{2n})) \text{ operations in $\QQ$}
\] 
which describes $\Pcal := \Pcal_0 \cup \sing(V)$.

\medskip 
Besides, since $V$ is equidimensional, there exists, by \cite[Lemma 
A.13]{SS2017}, an atlas $\bchi$ of $(V,\sing(V))$. According to 
Definition~\ref{def:atlas}, $\bchi$ is an atlas of $(V,\Pcal)$ as well.

\paragraph*{Steps~\ref{step:randa}\,-\ref{step:phigen}}
By definition of the procedure \Rand, $\balf$ is an arbitrary element of 
$\QQ^{2n}$, and according to Lemma~\ref{lem:phigen}, $\Gammaphi$ is a \SLP of 
length $E'=6n-2 \inOh O(n)$, which evaluates $\phi = (\polun(\XX) + 
\pscal{\balf_1,\XX},\pscal{\balf_2,\XX})$, where $\polun = x_1^2 + \cdots + 
x_n^2$. In particular, $E'':= E + E' \inOh O(E+n)$. Note also that since $D\geq
2$, it bounds the degrees of the polynomials in $\phi$.

\medskip
Let $\ZOpen$ be the intersection of the following four \NEZO subsets of 
$\CC^{2n}$:
\[
 \ZOnoether(V,\polun,2),\; \ZOpolar(\bchi,V,\sing(V),(\polun,0),(0),2), 
\; \ZOcritpolar(V,\polun)
 \text{ and } \ZOfiber(\bchi,V,\sing(V),\polun,2),
\]
defined respectively by Propositions~\ref{prop:fibergenalg}, 
\ref{prop:polargen}, 
\ref{prop:critpolargen} and \ref{prop:dimfiber} applied to $V$, $\bphi$ and 
possibly $\bchi$.
The set $\ZOpen$ is a \NEZO subset of $\CC^{2n}$ as well, and for now on, we 
suppose that $\balf \in \ZOpen$.

\paragraph*{Step~\ref{step:solvepolar}}
Let $W = \Wphii[2]$. Since
$\balf\in\ZOpolar(\bchi,V,\sing(V),(\polun,0),(0),2)$, by
Proposition~\ref{prop:polargen}, either $W$ is empty or it is
equidi\-men\-sio\-nal of dimension 1, with $\sing(W) \subset \sing(V)$.
Moreover, in the latter case, since $(d+3)/2\geq 2$ by assumption,
$\Watlas(\bchi,V,\sing(V),\phi,2)$ is an atlas of $(W,\sing(V))$.

Hence, by Lemma~\ref{lem:solvepolar} and our assumptions, \SolvePolar returns a
\ODP $\scrW_2$, of degree at most
\[
 \delta = (n+c+4)D^{c+2}(D-1)^{d}(c+2)^d \inOh O(n^{d+1}D^{n+2}),
\]
such that $\Zparam(\scrW_2) = W$, using at most 
\[
 \softOh{(n+c)^3(E+(n+c)^3)D\delta^3 + (n+c)\delta \deltaS^2}
 \inOh \softOh{n^{3d+4}(E+n^3)D^{3n+7}}
\]
operations in $\QQ$.

\paragraph*{Steps~\ref{step:critpolar}\,-\ref{step:image}}
Since we assume $\balf\in \ZOcritpolar(V,\polun)$,
Proposition~\ref{prop:critpolargen} states that either $W$ is empty or it is
equidimensional of dimension 1, and $\Wphii[1][W]$ is finite. Moreover, since
$\balf\in \ZOpolar(\bchi,V,\sing(V),(\polun,0),(0),2)$, we deduce by
Proposition~\ref{prop:polargen} that $\Watlas(\bchi,V,\Pcal,\phi,2)$ is an atlas
of $(W,\Pcal)$, as $W$ is 1-equidimensional or empty and $\Pcal_0$ is finite.

Let $K=\Wphii[1][W] \cup \Pcal$. By Lemma~\ref{lem:critpolar}, \CritPolar
returns either \fail or a \ZDP $\scrK$, of degree at most
\[
 \deltaK = \delta(n+c+4)D + \deltaP \inOh \Oh{n^{d+2}D^{n+3}+\mu},
\]
using at most
\[
 \softOh{(n+c)^{12}(E+n)D^3\delta^2+(n+c)\deltaP^2}
 \inOh 
 \softOh{n^{2d+14}(E+n)D^{2n+7}+n\mu^2}
\]
operations in $\QQ$. Moreover, by assumption, $\scrK$ describes $K$. 
Finally, let $Q=\mapun(K)$ then, by Lemma~\ref{lem:imageZDP} and our 
assumptions,  on input $\Gammaphi$, $\scrK$ 
and $j=1$,  procedure \Image outputs a \ZDP $\scrQ$, of degree less than $\deltaK$, such that, in 
case of success, $\Zparam(\scrQ) = Q$.
Moreover, since by Lemma~\ref{lem:phigen}, $\Gammaphi$ has length in $O(n)$, 
then the execution of \Image uses at most 
\[
 \softOh{(n^2\deltaK+n)\deltaK} \inOh \softOh{n^{2d+6}D^{2n+6}}
\]
operations in $\QQ$.

\paragraph*{Step~\ref{step:fiberpolar}}
Since $\balf\in \ZOnoether(V,\polun,2)$, by Proposition~\ref{prop:fibergenalg}, 
$W \cap \mapunrec(z)$ is finite for any $z \in \CC$.
In particular, $W \cap \mapunrec(Q)$ is finite, since $Q = \Zparam(\scrQ)$ is.  
Besides, as seen above, $\Watlas(\bchi,V,\Pcal,\phi,2)$ is an atlas of 
$(W,\Pcal)$ since $\balf$ lies in $\ZOpolar(\bchi,V,\sing(V),(\polun,0),(0),2)$.

Let $\Pcal_F=[W \cap \mapunrec(Q)] \cup \Pcal$. By Lemma~\ref{lem:fiberpolar}
and our assumptions, \FiberPolar outputs a \ZDP $\scrP_F$, of degree bounded by
\[
  \deltaPF = \deltaK\delta + \deltaP \inOh O(n^{2d+3}D^{2n+5} + \mu),
\]
using at most $\softOh{(n+c)^4(E+(n+c)^2)D\deltaK^2\delta^2+(n+c)\deltaP^2}$
operations in $\QQ$ which is in 
\[
\softOh{n^{4d+10}(E+n^2)D^{4n+10} + n\mu^2}
\]
and such that $\scrP_F$ describes $\Pcal_F$. Besides, remark
that by definition $\phi(\Pcal) \subset \phi(Q)$ so that $\Pcal_F = [W \cup
\Pcal] \cap \mapunrec(Q)$.

\paragraph*{Step~\ref{step:crit}}
Since $\balf \in \ZOpolar(\bchi,V,\sing(V),(\polun,0),(0),2)$, by
Proposition~\ref{prop:polargen}, $\Wphii[1]$ is finite. Besides, under assumption
\ref{ass:A}, $V$ is equidimensional with finitely many singular points. Let
$\Scal_F = \Kphii[1]$. By Lemma~\ref{lem:crit} and our assumptions, \Crit
outputs a \ZDP $\scrS_F$, which describes $\Scal_F$, of degree bounded by
\[
 \deltaSF = \binom{n+1}{d} D^{c+2}(D-1)^{d} \inOh O(n^dD^{n+2})
\]
using at most
\[
 \softOh{(n+2)^{4d+8}(E+n)D^{2n+3}(D-1)^{2d} + n\deltaS^2}
 \inOh \softOh{n^{4d+8}(E+n)D^{4n+3}}
\]
operations in $\QQ$.

\paragraph*{Step~\ref{step:rmbound}}
Since $\ff$ satisfies assumption \ref{ass:A}, the ideal $\pscal{\ff}$ generated by
the polynomials in $\ff$ is radical.
Besides, the restriction of $\mapun$ to $\V(\ff)\cap \RR^n$ is naturally proper 
and bounded from below by $\sum_{i=1}^n \balf_i^2/4$.
Hence, as $Q=\Zparam(\scrQ)$ is finite, $Q \cap \RR$ is bounded and so 
is
\[
  V \cap \RR^n \cap \mapunrec(Q \cap \RR^2) = V \cap \mapunrec(Q) \cap \RR^n,
\]
as $\bphi \subset \QQ[\XX]$, since $\balf \in \QQ^{2n}$ by above.

Let $F_Q = V\cap\mapunrec(Q)$.
Since $\balf\in\ZOfiber(\bchi,V,\sing(V),\polun,2)$, by 
Proposition~\ref{prop:dimfiber} $F_Q$ is either empty or equidimensional of 
dimension $d-1$, with $\sing(F_Q) \subset S_Q$, where 
\[
 S_Q := \sing(V) \cup [\Wphii[1] \cap \mapunrec(Q)] = \Kphii[1],
\]
since $\mapun(\Kphii[1]) \subset \mapun(Q)$.\\
Moreover, in the latter case, the sequence $\Fatlas(\bchi,V,Q,\sing(V),\bphi)$ is an atlas 
of $(F_Q,S_Q)$.
The \ZDPs $\scrP_F$ and $\scrS_F$ describe respectively 
finite sets $\Pcal_F$ and $\Scal_F$ such that
\[
 S_Q = \Scal_F \subset \Pcal_F \subset F_Q,
\]
and $\scrS_F$ has degree $\deltaSF \leq (nD)^{n+2}$.
Finally, recall that $\scrQ$ and $\scrP_F$ both have degree bounded by 
$\softOh{\mu+n^{2d+3}D^{2n+5}}$.
Hence, according to Proposition~\ref{prop:comp-broadmap}, and after a few
straightforward simplifications, we deduce that \RMBound either outputs \fail 
or a \ODP $\scrR_F$ of degree at most 
  \[
 \scrB_{\scrR_F} = \overalldegree,
\]
using
\[
   \overallcomplexity
\]
operations in $\QQ$.
Moreover, in case of success, $\scrR_F$ describes a roadmap of $(F_Q,\Pcal_F)$.

\paragraph*{Step~\ref{step:return}}
Remark that $\scrW_2$ and $\scrR_F$ both have degree at most $\scrB_{\scrR_F}$.
hence, by \cite[Lemma J.8]{SS2017}, on input $\scrW_2$ and $\scrR_F$, \Union 
either outputs \fail or a \ODP of degree at most $\Otilde(\scrB_{\scrR_F})$ 
using $\Otilde(n\scrB_{\scrR_F}^3)$ operations in $\QQ$. Therefore, the 
complexity of this step is bounded by the one of previous step.
Moreover, in case of success, the output describes $W \cup F_Q$.

\medskip
It follows that under assumption \ref{ass:A}, all assumptions from
Proposition~\ref{prop:connectresult} are satisfied. Hence, since 
$\Zparam(\scrR_F)$ is a roadmap of $(F_Q,\Pcal_F)$ and $\Pcal_F = \Pcal \cup 
(F_Q \cap W)$, by Proposition~\ref{prop:connectresult}, 
Algorithm~\ref{alg:algormunbound} returns a roadmap of $(V,\Pcal)$.
Since $\Pcal$ contains $\Pcal_0$, the output is a roadmap of $(V,\Pcal_0)$ as 
well. 

\medskip
In conclusion, if $\balf \in \ZOpen$ and all calls to the subroutines are 
successful then, on input $\Gamma$ and $\scrP_0$ such that 
assumption \ref{ass:A} is satisfied, Algorithm~\ref{alg:algormunbound} outputs
a \ODP encoding a roadmap of $(V,\Pcal_0)$.
Moreover this parametrization has degree bounded by $\scrB_{\scrR_F}$ and all steps 
have complexity bounded by the one of Step~\ref{step:rmbound}.
Since these bounds match the ones given in the statement of 
Theorem~\ref{thm:corrcomp}, we are done.
\end{proof}

Our main result, namely Theorem~\ref{thm:main}, is a direct consequence of 
Theorem~\ref{thm:corrcomp} since, if $n-c<2$ then $\V(\ff)$ is a roadmap of 
$(\V(\ff),\Zparam(\scrP))$.

\begin{remark}
  Remark that, as long as the restriction of
  $\mapun$ to $\V(\ff)\cap\RR^n$ is proper and bounded below, the above proof
  still holds. This could allow \rev{for} a more \emph{ad-hoc} choice for $\map$.
\end{remark}

 \section{Subroutines}\label{sec:subroutines}

\subsection{Proof of Lemma~\ref{lem:imageZDP}}\label{ssec:basicsubroutines}

\begin{lemma}\label{lem:incvar}
Let $\Gamma$  and $\Gammaphi$ be \SLPs of respective lengths $E$ and $E'$ 
computing sequences of polynomials respectively $\ff$ and 
$\bphi=(\phi_1,\dotsc,\phi_i)$ in $\QQ[\XXi]$. Then there exists an algorithm 
\IncVar which takes as input $\Gamma$, $\Gammaphi$ and returns a \SLP $\Gammat$ 
of length 
\[
  E + E' + i,
\]
that evaluates $\ffphi = (\ff,\phi_1-e_1,\dotsc,\phi_i-e_i)$ in $\QQ[\EE,\XX]$, 
where $\EE=(\EEi)$ are new variables.
\end{lemma}
\begin{proof}
 Up to reordering, we can suppose that the polynomials $\phi_1,\dotsc,\phi_i$ 
correspond to the respective indices $E'-i+1,\dotsc, E'$ in $\Gammaphi$.
 Let $1\leq j \leq N$, then the \SLP
\[
 \Gamma^{\bphi-\EE} = \Big(\Gammaphi,\; (+,\:E'- i + 1,\:-n-i+1), 
\dotsc,(+,\:E',\:-n)\Big)
\]
has length $E'+i$ and computes $(\phi_1-e_1,\dotsc,\phi_i-e_i)$ in 
$\QQ[\EEi,\XXi]$.
Finally let
\[
 \Gamma'=(\Gamma, \Gamma^{\bphi-\EE}),
\]
then $\Gamma'$ is a \SLP of length $E + E'+i$, that computes $\ffphi = 
(\ff,\phi_1-e_1,\dotsc,\phi_i-e_i)$ in $\QQ[\EEi,\XXi]$. 
\end{proof}

Let $1\leq i \leq n$ be integers and
$\bphi=(\phi_1,\dotsc,\phi_i)\subset\CC[\XX]$, and set
\[
  \begin{array}{cccc}
    \Inci\colon & \CC^n & \rightarrow & \CC^{i+n}\\
                & \yy & \mapsto & (\bphi(\yy),\yy)
  \end{array}.
\]
Then $\Inci$ is an isomorphic embedding of algebraic sets, with inverse the 
projection on the last $n$ coordinates. We call $\Inci$ the \emph{incidence 
  isomorphism associated to $\bphi$}.

Let $V\subset \CC^n$ be a $d$-equidimensional algebraic set with $1\leq d 
\leq n$. Then $\Vphi = \Inci(V) \subset \CC^{i+n}$ is called the 
\emph{incidence variety associated to $V$ with respect to $\bphi$}, or in 
short, the incidence variety of $(V,\bphi)$.

Finally, we note $\proj = (e_1,\dotsc,e_i)$ so that for $0\leq j\leq i$, 
$\proj[j]$ is the canonical projection on the first $j$ coordinates in 
$\CC^{i+n}$. 
The following lemma is immediate, and illustrates the main feature that 
motivates the introduction of incidence varieties.
\begin{lemma}\label{lem:diagprojphi}
  For any $0\leq j\leq i$, the following diagram commutes
  \begin{equation*}
\begin{tikzcd}
      V 
      \arrow[rr, "\displaystyle\Inci"] 
      \arrow[rrd, "\displaystyle\bphi_j^{\hspace*{0.5cm}}" below] 
      && \Vphi 
      \arrow[d, "\displaystyle\,\bpi_{j}"]
      \\
      && \CC^j
    \end{tikzcd}\;.
  \end{equation*}
\end{lemma}

\begin{lemma}\label{lem:incparam}
  Let $\scrQ$ be a zero-dimensional parametrization of degree $\kappa$ such that
  $\Zparam(\scrQ)\subset\CC^n$ and let $\Gammaphi$ be a \SLP of length $E'$
  which evaluates polynomials $\bphi=(\phi_1,\dotsc,\phi_i)$. There exists an
  algorithm \IncParam which takes as input $\scrQ$, $\Gammaphi$ and returns a
  \ZDP $\scrQt$ of degree $\kappa$ and encoding $\Inci(\Zparam(\scrQ)) \subset
  \CC^{i+n}$, where $\Inci$ is the incidence isomorphism associated to $\bphi$,
  using
\[
 \Otilde\left(E'\kappa\right)
\]
operations in $\QQ$.
\end{lemma}
\begin{proof}
 Write $\scrQ = ((q, v_1,\dotsc,v_n),\linF)$ following the definition of \ZDPs 
given in the introduction.
 Since 
 \[
  \Zparam(\scrQ)= \{ (v_1(\tt),\dotsc, v_n(\tt)) \mid q(\tt) = 0\}
 \]
 then $\Inci(\Zparam(\scrQ))$ is 
 \[
  \left\{ \Big(
  \phi_1\big(v_1(\tt),\dotsc, v_n(\tt)\big),\dotsc,
  \phi_i\big(v_1(\tt),\dotsc, v_n(\tt)\big),
  v_1(\tt),\dotsc, v_n(\tt) \Big) \mid q(\tt)=0 
  \right\}.
 \]
 Let $\EEi$ be new indeterminates and $\linF'(e_1,\dotsc,e_i,x_1,\dotsc,x_n) = 
\linF(x_1,\dotsc,x_n)$ and for all $1\leq j \leq i$, let $w_j = 
\phi_j(v_1,\dotsc, v_n\big) \mod{q} \in \QQ[t]$.
 Then we claim that $\scrQt=((q,w_1,\dotsc,w_i,v_1,\dotsc,v_n),\linF)$ is a 
\ZDP of $\Inci(\Zparam(\scrQ))$.
 Indeed for all $1\leq j \leq i$, $\deg(w_j)<\deg(q)$ and 
 \[
  \linF'\big(w_1,\dotsc,w_i,v_1,\dotsc,v_n\big) = 
\linF\big(v_1,\dotsc,v_n\big) = t.
 \]
Besides, computing $\scrQt$ is done by evaluating $\Gammaphi$ at 
$v_1,\dotsc,v_n$ doing all operations modulo $q$; this can be done using 
$\softOh{E'\kappa}$ operations in $\QQ$.
\end{proof}

\noindent
We can now prove Lemma~\ref{lem:imageZDP}.

\begin{myproof}{Lemma~\ref{lem:imageZDP}}
 Let $\Inci$ be the incidence isomorphism associated to $\bphi$.
 By Lemma~\ref{lem:diagprojphi}, the image of 
$\Zparam(\scrP)$ by $\map[j]$, can be obtained by projecting the incidence 
variety $\Inci(\Zparam(\scrP))$ on the first $j$ coordinates.

 Hence the algorithm \Image can be performed as follows. 
First, according to Lemma~\ref{lem:incparam}, there exists an algorithm 
\IncParam which, on input $\scrP$ and $\Gammaphi$, computes a \ZDP $\scrPt$ of 
degree $\kappa$, encoding $\Inci(\Zparam(\scrP)) \subset \CC^{j+n}$, and using 
$\softOh{E'\kappa}$ operations in $\QQ$.
Secondly, according to \cite[Lemma J.5.]{SS2017}, there exists an algorithm 
\Proj which, on input $\scrPt$ and $j \in \{1,\dotsc,i\}$, computes a \ZDP 
$\scrQ$ encoding 
\[
 \proj[j](\scrPt) = \proj[j]\left(\Inci(\Zparam(\scrP))\right) = 
\map[j](\Zparam(\scrP)),
\]
using $\softOh{n^2\kappa^2}$ operations in $\QQ$.
\end{myproof}

\subsection{Auxiliary results for generalized polar 
varieties}\label{ssec:incidence}

We reuse the notation introduced in the previous subsection. Let
$\EE=(e_1\dotsc,e_i)$ new indeterminates. Recall that $V\subset \CC^{n}$ is a
$d$-equidimensional algebraic set.
 \begin{lemma}\label{lem:IVphi}
 Let $\hh\subset \CC[\XX]$ be a set of generators of $\I(V)$. Then
 \[
  \hhphi = (\hh,\phi_1-e_1,\dotsc,\phi_i-e_1) \subset \CC[\EE,\XX]
 \]
 is a set of generator of $\I(\Vphi)\subset \CC[\EE,\XX]$, which is 
equidimensional of dimension $d$.
 \end{lemma}
 \begin{proof}
 Remark that by Lemma~\ref{lem:ranksubmatrix}, for any $(\tt,\yy)\in\Vphi$,
  \[
  \rank \jac_{\tt,\yy}(\hhphi) =
  \rank \begin{bmatrix}
         \multicolumn{2}{c}{\OO} & \jac_{\yy}(\hh)\\
         \multicolumn{2}{c}{-I_i}  &\jac_{\yy}(\bphi)
        \end{bmatrix}
  = \rank \jac_{\yy}(\hh) + i,
\]
so that for all $\yy\in\reg(V)$, since $\jac(\hh)$ has rank 
$n-d$ at $\yy$, then $\jac(\hhphi)$ has rank $n-d+i$ at $\Inci(\yy)$. 
Hence, since $\reg(V)$ is Zariski dense in $V$, by \cite[Lemma 15]{SS2011} 
$\pscal{\hhphi}$ is an equidimensional radical ideal of dimension $d$.
  
Besides, let $(\tt,\yy)\in \CC^n$, then $\hhphi(\tt,\yy)=0$ if and only if 
$\hh(\yy)=0$ and $\bphi(\yy)=\tt$ that is $(\tt,\yy) \in \Vphi$ since $\hh$ 
generates $\I(V)$. Hence $\V(\pscal{\hhphi}) = \Vphi$ so that by  Hilbert's 
Nullstellensatz \cite[Theorem 1.6]{Ei1995}, 
\[
 \I(\Vphi) = \sqrt{\pscal{\hhphi}} = \pscal{\hhphi}.
\]
\end{proof}

The following lemma shows an important consequence of
Lemma~\ref{lem:diagprojphi} for polar varieties.

\begin{lemma}\label{lem:InciPolar}
 For $0\leq j \leq i$, the restriction of $\Inci$ induces an isomorphism 
between $\Wphii[j]$ (resp. $\Kphii[j]$) and $\Wproji[j]$ (resp. $\Kproji[j]$). 
\end{lemma}
\begin{proof}
Let $\hh$ be generators of $\I(V)$. By Lemma~\ref{lem:IVphi}, $\hhphi$ are 
generators of $\I(V)$.
 Let $\yy \in V$, $\yyphi = \Inci(\yy) \in \Vphi$ and $0\leq j \leq i$. Then by 
Lemma~\ref{lem:ranksubmatrix},
 \begin{equation}\label{eqn:rankincidence}
  \rank \jac_{\yyphi}([\hhphi,\proj[j]]) =
  \rank \begin{bmatrix}
         \multicolumn{2}{c}{\OO} & \jac_{\yy}(\hh)\\
         \multicolumn{2}{c}{-I_i}  &\jac_{\yy}(\bphi)\\
         I_j & \OO & \OO
        \end{bmatrix}
  = \rank \jac_{\yy}([\hh,\map[j]]) + i,
 \end{equation}
where $I_\ell$ denotes the $\ell\times \ell$ identity matrix.
 Since both $V$ and $\Vphi$ are $d$-equidimensional, then by \cite[Lemma 
A.2]{SS2017}, $\Kphii[j]$ and $\Kproji[j]$ are the sets of points $\yy \in V$ 
and $\yyphi\in\Vphi$ where respectively 
\[
 \jac_{\yy}([\hh,\map[j]])<n-d+j \et \jac_{\yyphi}([\hhphi,\proj[j]]) < n+i-d+j.
\]
Hence by \eqref{eqn:rankincidence}, the two conditions are equivalent and then
it holds that
\[\Inci(\Kphii[j]) = \Kproji[j] \qquad \text{for all $0\leq j \leq i$. }\]
In particular, for $j=0$, $\Inci(\sing(V)) = \sing(\Vphi)$, so that for all
$0\leq j \leq i$,
\[
 \Inci(\Wophii[j]) = \Woproji[j].
\]
Since $\Inci$ is an isomorphism of algebraic sets, it is a homeomorphism for the
Zariski topology, so that it maps the Zariski closure of sets to the Zariski
closure of their image. Hence, we can conclude that $\Inci(\Wphii[j]) =
\Wproji[j]$ for all $0\leq j \leq i$.
\end{proof}

\begin{lemma}[Chart and atlases]\label{lem:Inciatlas}
  Let $1\leq e \leq n$, $Q\subset \CC^e$ be a finite set and $S$ be an algebraic
  set such that $V$ and $S$ lie over $Q$ with respect to $\bphi$. By a slight
  abuse of notation, we denote equally $m \in \CC[\XX]$ when seen in
  $\CC[\EE,\XX]$. Then, the following holds.
\begin{enumerate}
\item Let $\chi=(m,\hh) \subset \CC[\XX]$ be a chart of $(V,Q,S,\bphi)$, then
  $\chiphi=(m,\hhphi) \subset \CC[\EE,\XX]$ is a chart of
  $(\Vphi,Q,\Sphi,\proj)$, where $\Sphi = \Inci(S)$.
\item Let $\bchi = (\chi_j)_{1\leq j\leq s}$ be an atlas of $(V,Q,S,\bphi)$,
  then if $\bchiphi = (\chiphi_j)_{1\leq j \leq s}$ as defined in the previous
  item, $\bchiphi$ is an atlas of $(\Vphi,Q,\Sphi,\proj)$.
\end{enumerate}
\end{lemma}
\begin{proof}
  We start with the first statement. Let $Q,S$ and $\chi=(m,\hh)$ be as in the
  statement. Then, it holds that:
 \begin{enumerate}[label=$\sfC_\arabic*:$]
  \item Let $\yy \in \Ocal(m)\cap V - S$, which is non-empty by property 
$\sfC_1$ of $\chi$. Then by definition $\Inci(\yy) \in \Vphi$, and since 
$\Inci$ is an isomorphism on $\Vphi$, $\Inci(\yy) \notin \Sphi$. Finally since 
$m \in \CC[\XX]$, then $m(\Inci(\yy)) = m(\yy) \neq 0$ so that $\Ocal(m) \cap 
\Vphi - \Sphi$ is not empty.
  
\item Note that since $m\in \CC[\XX]$, $\Inci(\Ocal(m))$ is defined by $m\neq
  0$. By a slight abuse of notation, we still denote this Zariski open set $
  \Ocal(m)$. Hence, it follows from the definition of $\Inci$ that
  $\Inci(\Ocal(m) \cap V - S) = \Ocal(m) \cap \Vphi - \Sphi$. Besides, by
  Lemma~\ref{lem:diagprojphi}, $\proj[e]\circ\Inci$ and $\map[e]$ coincide on
  $V$. Then
  \[
  \mapfbr{Z} = \projfbr{\Inci(Z)} \text{ for any } Z \subset V.
  \]
  Finally, as seen in the proof of Lemma~\ref{lem:IVphi}, $\Inci(\V(\hh)) = 
\V(\hhphi)$.
  Hence by property $\sfC_2$ of $\chi$, 
  \[
   \Ocal(m) \cap \Vphi - \Sphi =
   \Inci(\Ocal(m) \cap \mapfbr{\V(\hh)} - S) =
   \Ocal(m) \cap \projfbr{\V(\hhphi)} - \Sphi,
  \]
  since $\Ocal(m) \cap \mapfbr{\V(\hh)} - S$ is a subset of $V$.
  
  \item Let $c$ be the cardinality of $\hh$, then $\hhphi$ has cardinality 
$c+i$. Hence by property $\sfC_3$ of $\chi$, $e+c+i \leq i+n$ as required.
  
  \item Finally let $\yyphi=(\tt,\yy) \in \Ocal(m) \cap \Vphi - \Sphi$, we know 
from above that $\yy \in \Ocal(m) \cap V - S$, so that by property $\sfC_4$ of 
$\chi$, $\jac_{\yy}[\hh,\map[e]]$ has full rank $c+e$.
  But by  equality \eqref{eqn:rankincidence} in the proof of 
Lemma~\ref{lem:InciPolar}, this means that $\jac_{\yyphi}([\hhphi,\proj[e]])$ 
has full rank $c+i+e$ as required.
 \end{enumerate}
We have shown that charts can be transferred to incidence varieties, let us now
prove that this naturally gives rise to atlases.
Consider an atlas $\bchi = (\chi_j)_{1\leq j\leq s}$ of $(V,Q,S,\bphi)$,
and let $\bchiphi = (\chiphi_j)_{1\leq j \leq s}$, where for all $1\leq j \leq 
s$, $\chiphi_j$ is defined from $\chi_j$ as above. We proved that $\bchiphi$ is 
an atlas of $(\Vphi,Q,\Sphi,\proj)$.

Property $\sfA_1$ is straightforward, and $\sfA_2$ is given by the first
statement of this lemma which we just proved. Finally, since $\Inci(V-S) = \Vphi
- \Sphi$, then for any $\yyphi = (\tt,\yy) \in \Vphi-\Sphi$, by property
$\sfA_3$ of $\bchi$, there exists $1\leq j \leq s$ such that $m_j(\yyphi) =
m_j(\yy)\neq 0$. Then $\bchiphi$ satisfies property $\sfA_3$ of atlases.
\end{proof}

We deduce the following results for two important particular cases.
\begin{lemma}\label{lem:InciPolaratlas}
  Let $S\subset \CC^{n}$ be an algebraic set, $\chi=(m,\hh)$ and $\bchi =
  (\chi_j)_{1\leq j\leq s}$ be respectively a chart and an atlas of $(V,S)$, and
  let $\chiphi$ and $\bchiphi$ the chart and atlas constructed from respectively
  $\chi$ and $\bchiphi$ as in Lemma~\ref{lem:Inciatlas}. The following
  holds.
 \begin{enumerate}
 \item If $\hh$ has cardinality $c$, then for any $c$-minor $m'$ of $\jac(\hh)$
   and any $(c+i-1)$-minor $m''$ of $\jac([\hh,\map[i]])$, containing the rows
   of $\jac(\map[i])$, the following holds. If $\Wchart(\chi, m',m'')$ is a
   chart of $\scrW=(\Wphii,S)$, then $\Wchart(\chiphi,m',m'')$ is a chart of
   $\scrWphi=(\Wproji[i][\Vphi],\Sphi)$.
  \item If $\Watlas(\bchi,V,S,\bphi,i)$ is an atlas of $\scrW$ then, 
$\Watlas(\bchiphi,\Vphi,\Sphi,\proj,i)$ is an atlas of $\scrWphi$.
 \end{enumerate}
\end{lemma}
\begin{proof}
Let $m'$ and $m''$ be respectively a $c$-minor of $\jac(\hh)$ and a 
$(c+i-1)$-minor of $\jac([\hh,\map[i]])$, containing the rows of 
$\jac(\map[i])$.
Assume that 
\[
 \Wchart(\chi,m',m'',\bphi) =
 \Big( \: mm'm'' , \: \big(\hh,\Hphi(\hh,i,m'') \big) \: \Big)
\]
is a chart of $\scrW$. By $\mathsf{C}_1$, $\Ocal(mm'm'')\cap\Wphii-S$ is not 
empty, 
so that $m'$ and $m''$ are not identically zero.
Since
\[
 \jac(\hhphi) = 
 \begin{pmatrix}
  \OO & \jac(\hh)\\
  -I_i & \jac(\map[i])
 \end{pmatrix},
\]
Lemma~\ref{lem:minorssubmat} shows that $m'$ is a $(c+i)$-minor of
$\jac(\hhphi)$ and $m''$ is a $(c+i+i-1)$-minor of $\jac(\hhphi,\proj[i])$
containing $I_i=\jac(\proj[i])$. Hence, according to
Definition~\ref{def:Wchart},
\[
 \Wchart(\chiphi,m',m'') = 
\left(mm'm'',(\hhphi,\Hphi[\proj](\hhphi,i,m''))\right),
\]
where, by definition, $\Hphi[\proj](\hhphi,i,m'')$ is the sequence of 
$(c+i+i)$-minors of $\jac([\hhphi,\proj[i]])$ obtained by successively adding 
the missing row and the missing columns of $\jac([\hhphi,\proj[i]])$ to $m''$.

But, since $m''\neq 0$, Lemma~\ref{lem:chartaltlasWproj} implies that
$\Hphi[\proj](\hhphi,i,m'')$ is, as well, the sequence of $(c+i)$-minors obtained
by successively adding the missing row and the missing columns of
$\jac(\hhphi,i)=\jac([\hh,\map[i])$ to $m''$. We deduce that
\[
 \Hphi[\proj](\hhphi,i,m'') = \Hphi(\hh,i,m''),
\]
so that if $\gg = (\hh,\Hphi(\hh,i,m''))$, then $\ggphi =
(\hhphi,\Hphi[\proj](\hhphi,i,m''))$.\\
Hence $\Wchart(\chiphi,m',m'')$ is the chart constructed from
$\Wchart(\chi,m',m'',\bphi)$ in Lemma~\ref{lem:Inciatlas}, and since, by
Lemma~\ref{lem:InciPolar}, $\Inci(\Wphii) = \Wproji$, the first statement of
Lemma~\ref{lem:Inciatlas} implies that $\Wchart(\chiphi,m',m'')$ is a chart of
$\scrWphi$.

To prove the second assertion, remark that by Lemma~\ref{lem:chartaltlasWproj}
(third assertion),
$\Watlas(\bchiphi,\Vphi,\Sphi,\proj,i)$ is the sequence of all those
$\Wchart(\chiphi_j,m',m'')$, for $j\in\{1,\dotsc,s\}$ and for $m',m''$
respectively a $c+i$-minor of $\jac(\hhphi_j)$ and a $(c+i-1)$-minor of
$\jac(\hhphi_j,i)$ for which $\Ocal(m_jm'm'')\cap\Wproji-S$ is not empty.

As seen above, the polynomials $m'$ and $m''$ are actually $c$-minors of
$\jac(\hh_j)$ and $(c+i-1)$-minors of $\jac([\hhphi_j,\map[i]])$, and in the
first point, we prove that $\Wchart(\chiphi_j,m',m'')$ is the chart constructed
in the first point of Lemma~\ref{lem:Inciatlas} from $\Wchart(\chi_j,m',m'')$.
Hence $\Watlas(\bchiphi,\Vphi,\Sphi,\proj,i)$ is exactly the atlas constructed
from $\Watlas(\bchi,V,S,\bphi,i)$ in the second item of
Lemma~\ref{lem:Inciatlas}. \\
In conclusion, by Lemma~\ref{lem:Inciatlas}, if
$\Watlas(\bchi,V,S,\bphi,i)$ is an atlas of $\scrW$, then
$\Watlas(\bchiphi,\Vphi,\Sphi,\proj,i)$ is an atlas of $\scrWphi$.
\end{proof}

\begin{lemma}\label{lem:InciFiberatlas} 
  Let $1\leq e \leq n$, $Q\subset \CC^e$ be a finite set and $S$ be an algebraic
  set such that $V$ and $S$ lie over $Q$ with respect to $\bphi$. Let further
\[
  \scrF = \left(\,\mapfbr{V},\;\mapfbr{(S\cup\Wphii[e])}\,\right)
\]
and
\[
 \scrFphi = 
\left(\,\projfbr{\Vphi},\;\projfbr{(\Sphi\cup\Wproji[e])}\,
\right).
\]
Let $\chi=(m,\hh)$ and $\bchi = (\chi_j)_{1\leq j\leq s}$ be respectively a 
chart and an atlas of $(V,Q,S,\bphi)$ and let $\chiphi$ and $\bchiphi$ the 
chart and atlas constructed from respectively $\chi$ and $\bchiphi$ as in 
Lemma~\ref{lem:Inciatlas}.

If $\Fatlas(\bchi,V,Q,S,\bphi)$ is an atlas of $\scrF$ then
$\Fatlas(\bchiphi,\Vphi,Q,\Sphi,\proj)$ is an atlas of $\scrFphi$.
\end{lemma}

\begin{proof}
Without loss of generality one can assume that $S \subset V$.
Since by Lemma~\ref{lem:diagprojphi}, $\proj[e]\circ\Inci$ and $\map[e]$ 
coincide on $V$, then 
\[
  \Inci \left( \mapfbr{(S\cup\Wphii[e])} \right) = 
\projfbr{(\Sphi\cup\Wproji[e])}.
\]
Hence, for any $1\leq j\leq s$, $\Ocal(m_j) \cap \projfbr{\Vphi} -
\projfbr{(\Sphi\cup\Wproji[e])} $ coincides with
\[
 \Inci\left(\Ocal(m_j) \cap \mapfbr{V} - 
\mapfbr{(S\cup\Wphii[e])}\right),
\]
so that these sets are not-empty for the same $j$'s in $\{1,\dotsc,s\}$.\\
Hence $\Fatlas(\bchiphi,\Vphi,Q,\Sphi,\proj)$ is the atlas constructed from
$\Fatlas(\bchi,V,Q,S,\proj)$ in Lemma~\ref{lem:Inciatlas}.

In conclusion, by the second assertion of Lemma~\ref{lem:Inciatlas}, if
$\Fatlas(\bchi,V,Q,S,\proj)$ is an atlas of $\scrF$ then
$\Fatlas(\bchi,V,Q,S,\proj)$ is an atlas of $\scrFphi$.
\end{proof}

\subsection{Lagrange systems}\label{ssec:lagrange}
We present here a simplified version of generalized Lagrange systems defined in
\cite[Section 5.2]{SS2017} to encode polar varieties and provide equivalent
results adapted to our case. As we only use a simplified version (involving a
single block of Lagrange multipliers), we simply call them Lagrange
systems.

\subsubsection{Definitions}
The following is nothing but a simplified version of \cite[Definition 
5.3]{SS2017}.
\begin{definition}\label{def:lagrangesystem}
 A \emph{Lagrange system} is a triple $L=(\Gamma,\scrQ,\scrS)$ where
 \begin{itemize}
  \item $\Gamma$ is a \SLP evaluating a sequence of polynomials $\FF=(\ff,\gg)\subset 
\QQ[\XX,\LL]$, where
  \begin{itemize}[label=$-$]
  \item $\XX=(X_1,\dotsc,X_n)$ and $\LL=(L_1,\dotsc,L_{m})$;
  \item $\ff=(f_1,\dotsc,f_p) \subset \QQ[\XX]$ and $\gg=(g_1,\dotsc,g_q) 
\subset \QQ[\XX,\LL]$ with $\deg_{\LL}\gg\leq 1$;
  \end{itemize}
  
  \item $\scrQ$ is a \ZDP with coefficients in $\QQ$, with $Q = \Zparam(\scrQ) 
\subset \CC^e$;
  
  \item $\scrS$ is a \ZDP with coefficients in $\QQ$, with 
$S=\Zparam(\scrS)\subset \CC^n$ lying over $Q$;
  
  \item $(n+m) - (p+q) \geq e$.
 \end{itemize}
 We also define $N$ and $P$ as respectively the number of variables and
 equations, so that
\[
 N=n+m,\quad P=p+q \et d = N-e-P \geq 0.
\]
\end{definition}
One checks that such a Lagrange system is also a generalized Lagrange system in
the sense of \cite[Definition 5.3]{SS2017}. We can then define the same objects
associated to such systems as follows. We denote by $\piX\colon\CC^N\to\CC^n$
the projection on the variables associated to $\XX$ in any set of $\CC^N$
defined by equations in $\CC[\XX,\LL]$.

\begin{definition}\label{def:datalagrange}
 Let $L=(\Gamma,\scrQ,\scrS)$ be a Lagrange system and all associated data 
defined in Definition~\ref{def:lagrangesystem}. We define the following objects:
 \begin{itemize}
  \item the \emph{type of $L$} is the triple $T=(\nn,\pp,e)$ where $\nn=(n,m)$ 
and $\pp=(p,q)$;
  
\item $\scrU(L) = \piX\Big(\projfbr{\V(\FF)}[Q][e] - \piX[-1](S)\Big)\subset 
\CC^n$
  
  \item $\ZUL \subset \CC^n$ the Zariski closure of $\scrU(L)$.
 \end{itemize}
Then we say that $L$ defines $\ZUL$.
\end{definition}
(We see here that Lagrange systems are nothing but generalized 
Lagrange systems of type $(1,\nn,\pp,e)$, in the sense of \cite{SS2017}).
We now define local and global normal forms, that can be seen as equivalent to 
charts and atlases for Lagrange systems, replacing the notion of complete 
intersection by the one of normal form presented below. 

For any non-zero polynomial $M$ of a polynomial ring $\CC[\YY]$ we denote by
$\CC[\YY]_M$ the localization of $\CC[\YY]$ at $M$, that is the of all $g/M^j$
where $\gg \in \CC[\YY]$ and $j \in \mathbb{N}$.

\begin{definition}\label{def:normalform}
 For a non-zero $M\in\QQ[\XX]$ and polynomials $\HH \subset \QQ[\XX,\LL]_M$, we 
say that \emph{$\HH$ is in normal form in $\QQ[\XX,\LL]$} if these polynomials 
have the form
 \[
  \HH=(h_1,\dotsc,h_c,L_1-\rho_1,\dotsc,L_m-\rho_m),
 \]
 where the $h_j$'s are in $\QQ[\XX]$ and the $\rho_j$'s are in $\QQ[\XX]_M$. We
 call $\hh=(h_1,\dotsc,h_c)$ and $\brho=(L_j-\rho_j)_{1\leq j\leq m}$
 respectively the $\XX$- and $\LL$-components of $\HH$.
\end{definition}
\begin{definition}\label{def:localform}
 A \emph{\LNF} of a Lagrange system $L=(\Gamma,\scrQ,\scrS)$ is the data of 
$\psi=(\mLNF,\dLNF,\hh,\HH)$ that satisfies the following conditions:
 \begin{enumerate}[label=$\sfL_{\arabic*}$]
  \item $\mLNF,\dLNF \in \QQ[\XX]-\{0\}$ and $\HH$ is in normal form in 
$\QQ[\XX,\LL]_{\mLNF\dLNF}$ with $\XX$-component $\hh=(h_1,\dotsc,h_c)$;
  
  \item $\HH$ and $\FF$ have the same cardinality $n-c = N-P$;
  
  \item $\pscal{\FF,\I(Q)} = \pscal{\HH,\I(Q)}$ in 
$\QQ[\XX,\LL]_{\mLNF,\dLNF}$;
  
  \item $(\mLNF,\hh)$ is a chart of $(V,Q,S)$;
  
  \item $\dLNF$ does not vanish on $\Ocal(\mLNF)\cap\scrU(L)$.
 \end{enumerate}
Given such a local normal form $\psi$ we will note $\chi=(\mLNF,\hh)$ the 
associated chart.
\end{definition}

As for atlases and charts, we define now global normal forms using local normal
forms. The definition takes into consideration a family $\scrY=(Y_1,\dotsc,Y_r)$
of algebraic sets; this is specifically needed to help us prove correctness of
the main algorithm.
\begin{definition}\label{def:globalform}
 A \emph{\GNF} of a Lagrange system $L=(\Gamma,\scrQ,\scrS)$ is the data of 
$\bpsi=(\psi_j)_{1\leq j \leq s}$ such that:
 \begin{enumerate}[label=$\sfG_{\arabic*}$]
  \item each $\psi_j=(\mLNF_j,\dLNF_j,\hh_j,\HH_j)$ is a \LNF;
  
  \item $\bchi = ((\mLNF_j,\hh_j))_{1\leq j \leq s}$ is an atlas of $(V,Q,S)$.
 \end{enumerate}
 Let further $\scrY=(Y_1,\dotsc,Y_r)$ be algebraic subsets of $\CC^n$. A \GNF 
of $(L;\scrY)$ is the data of a \GNF $\bpsi=(\psi_j)_{1\leq j \leq s}$ of $L$ 
such that for all $1\leq j\leq s$ and $1 \leq k \leq r$:
 \begin{enumerate}[label=$\sfG_{\arabic*}$]
  \setcounter{enumi}{2}
  \item for any irreducible component $Y$ of $Y_k$ contained in $V$ and such 
that $\Ocal(\mLNF_j)\cap Y - S$ is not empty, $\Ocal(\mLNF_j\dLNF_j)\cap Y - 
S$ is not empty.
 \end{enumerate}
We say that $L$ (resp. $(L;\scrY)$) has the \emph{\GNFP} if there exists a \GNF 
$\bpsi$ of $L$ (resp. $(L;\scrY)$) and we will note $\bchi$ the associated 
atlas.
\end{definition}

\subsubsection{Lagrange system for polar varieties}

We give here a slightly different version of results presented in \cite[Section
5.5]{SS2017}. We first recall the construction of \cite[Definition 5.11]{SS2017}
adapted to our more elementary case.
\begin{definition}\label{def:Wlagrange}
  Let $L=(\Gamma,(1),\scrS)$ be a Lagrange system whose type is
  $((n,0),(p,0),0)$, let $\ff\subset\CC[\XX]$ be the polynomials which are
  evaluated by $\Gamma$ and let $i \in \{1,\dotsc, n-p\}$.
 
 Let $\LL=(L_1,\dotsc, L_p)$ be new indeterminates, for $\uu=(u_1,\dotsc,u_p) 
\in \QQ^p$, define
 \[
  \FF_{\uu} = \Big( \ff, \; \lagrange(\ff,i,\LL),\; u_1L_1 + \cdots + u_pL_p-1 
\Big),
 \]
 where $\lagrange(\ff,i,\LL)$ denotes the entries of
 \[
  \big[ L_1 \; \cdots \; L_p \big] \cdot \jac(\ff,i).
 \]
 We define $\Wlag(L,\uu,i)$ as the triplet $(\Gamma_{\uu},\scrQ,\scrS)$, where 
$\Gamma_{\uu}$ is a \SLP that evaluates $\FF_{\uu}$, it is a Lagrange system of 
type $((n,p),(p,n-i+1),0)$.
\end{definition}

We can now prove an analog of \cite[Proposition 5.13]{SS2017}.
\begin{proposition}\label{prop:polarGNF}
Let $V,S \subset \CC^n$ be two algebraic sets with $V$ $d$-equidimen\-sio\-nal and 
$S$ finite. Let $\bchi$ be an atlas of $(V,S)$ and let $i\in\{2,\dotsc, 
(d+3)/2\}$.
Write $W=\Wproji[i][V]$ and assume that the following holds.
Either $W$ is empty or it is equidimensional of dimension $i-1$, with $\sing(W) 
\subset S$, and $\Watlas(\bchi,V,S,\proj,i)$ is an atlas of $(W,S)$.

Let $L=(\Gamma, (1),\scrS)$ be a Lagrange system such that $V=\ULbar$ and $S = 
\Zparam(\scrS)$. Let $\scrY = (Y_1,\dotsc,Y_r)$ be algebraic sets in $\CC^n$ 
and let finally $\bpsi$ be a \GNF for $(L;(W,\scrY))$ such that $\bchi$ is the 
associated atlas of $(V,S)$.
There exists a \NEZO subset $\scrI(L,\bpsi, \scrY)$ of $\CC^p$ such that for 
all $\uu\in\scrI(L,\bpsi,\scrY)\cap\QQ^p$, the following holds:
\begin{itemize}
 \item $\Wlag(L,\uu,i)$ is a Lagrange system that defines $W$;
 \item if $W\neq\emptyset$, then $(\Wlag(L,\uu,i);\scrY)$ has a \GNF whose 
atlas is $\Watlas(\bchi,V,S,\proj,i)$.
\end{itemize}
\end{proposition}
\begin{proof}
  The statement of this proposition is identical to \cite[Proposition
  5.13]{SS2017} except that, in \cite[Proposition 5.13]{SS2017}, our assumptions
  on $W$ are replaced by a generic linear change of variables on $V$.
  \cite[Proposition 5.13]{SS2017} claims the same statements on $V^{\AA}$ where
  $\AA$ is assumed to lie in a \NEZO set $\scrG_1(\bchi, V, \emptyset,S,i)$
  defined in \cite[Proposition 3.4]{SS2017}.

  In the proof of \cite[Proposition 5.13]{SS2017}, the fact that $\AA$ lies in
  $\scrG_1(\bchi, V, \emptyset,S,i)$ allows one to assume that the statements of
  \cite[Proposition 3.4]{SS2017} but also \cite[Lemma B.12]{SS2017} hold. In our
  proposition stated above, according to Lemma~\ref{lem:chartaltlasWproj}, the
  assumptions on $W$ are exactly the conclusion of \cite[Proposition
  3.4]{SS2017}, while \cite[Lemma B.12]{SS2017} is nothing but a consequence of
  these facts. Therefore, under these assumptions, the proof of
  \cite[Proposition 5.13]{SS2017} can be replicated, \emph{mutatis mutandis},
  for $V$ instead of $V^{\AA}$, and constitutes a valid proof for the above
  statement.
\end{proof}

\subsubsection{Lagrange system for fibers}

\begin{definition}\label{def:Flagrange}
 Let $L=(\Gamma,(1),\scrS)$ be a Lagrange system whose type is $((n,0),(p,0),0)$ and 
let $e \in \{1,\dotsc, n-p\}$.
 Let $\scrQ''$ be a \ZDP that encodes a finite set $Q''\subset \CC^e$ and let 
$\scrS''$ be a \ZDP that encodes a finite set $S''\subset \CC^n$ lying over 
$Q''$.
 We define $\Flag(L,\scrQ'',\scrS'')$ as the triplet 
$(\Gamma,\scrQ'',\scrS'')$, it is a Lagrange system of type $((n,0),(p,0),e)$.
\end{definition}

As in the previous paragraph, we state an analogue of \cite[Proposition
5.16]{SS2017} where we replaced the assumption of a generic linear change of
variables by the assumptions that such a change of variables allows us to satisfy.
In addition, we handle here the more general situation where, using the notation
below, $W=\Wproji[e]$, as the case $W=\Wproji[e+1]$ considered in \cite{SS2017}
can be deduced from the former.
\begin{proposition}\label{prop:fiberGNF}
Let $V,S \subset \CC^n$ be two algebraic sets with $V$ $d$-equidi\-men\-sio\-nal and 
$S$ finite. Let $\bchi$ be an atlas of $(V,S)$ and let $e\in\{2,\dotsc, 
(d+3)/2\}$.
Define $W=\Wproji[e]$ and let $\scrQ''$ and $\scrS''$ be \ZDPsQ that 
respectively encode a finite set $Q''\subset \CC^e$ and $S'' = \projfbr{S\cup 
W}[Q''][e]$ and let $V'' = \projfbr{V}[Q''][e]$.
Assume that $S''$ is finite and, either $V''$ is empty or its is 
equidimensional of dimension $d-e$, with $\sing(V'')$ contained in $S''$, and 
$\Fatlas(\bchi,V,S,\scrQ'',\proj)$ is an atlas of $(V'',Q'',S'')$.

Let $L=(\Gamma, (1),\scrS)$ be a Lagrange system such that $V=\ULbar$ and 
$S=\Zparam(\scrS)$. Let $\scrY = (Y_1,\dotsc,Y_r)$ be algebraic sets in $\CC^n$ 
and let finally $\bpsi$ be a \GNF for $(L;(V'',\scrY))$ such that $\bchi$ is 
the associated atlas of $(V,S)$.
Then the following holds:
\begin{itemize}
 \item $\Flag(L,\scrQ'',\scrS'')$ is a Lagrange system that defines $V''$;
 \item if $V''\neq\emptyset$, then $(\Flag(L,\scrQ'',\scrS'');\scrY)$ has a 
\GNF whose atlas is $\Fatlas(\bchi,V,Q'',S,\proj)$.
\end{itemize}
\end{proposition}       
\begin{proof}
  As above, the statement of this proposition is identical to the one in
  \cite[Proposition 5.16]{SS2017}, except that the assumptions on $S''$ and $V''$
  are replaced by a generic change of variables on $V$. Indeed,
  \cite[Proposition 5.16]{SS2017} claims the same statements as we do on 
$V^{\AA}$, where
  $\AA$ is assumed to lie in a \NEZO set $\scrG_3(\bchi,V,\emptyset,S,e)$
  defined in \cite[Proposition 3.7]{SS2017}.

In the proof of \cite[Proposition 5.17]{SS2017}, the fact that $\AA \in 
\scrG_3(\bchi,V,\emptyset,S,e)$ allows us to assume that the statements of 
\cite[Proposition 3.7]{SS2017} but also \cite[Lemma C.1]{SS2017} hold.
In the case of the proposition stated above, the assumptions on $S''$ and $V''$ 
are exactly the statement of \cite[Proposition 3.7]{SS2017}, while \cite[Lemma 
C.1]{SS2017} is nothing but a consequence of these facts.
Again, under these assumptions, the proof of \cite[Proposition 
5.17]{SS2017} can be replicated, \emph{mutatis mutandis}, for $V$ instead of 
$V^{\AA}$, and constitutes a valid proof for the above statement.
\end{proof}

\subsection{Proofs of Lemmas~\ref{lem:crit},~\ref{lem:solvepolar},~\ref{lem:critpolar}
  and~\ref{lem:fiberpolar}}\label{ssec:computepolar}

As done in Subsection~\ref{ssec:subpolarintro}, we fix $1\leq c \leq n-2$ and
we refer to the following objects:
\begin{itemize}
 \item sequences of polynomials $\gg =(g_1,\dotsc,g_c)$ and 
$\bphi=(\phi_1,\phi_2)$ in $\QQ[\XX]$, of maximal degrees $D$, such that $\gg$ 
satisfies assumption $\sfA$ that is: $\gg$ is a reduced regular sequence and 
$\sing(\V(\gg))$ is finite;
 
 \item \SLPs $\Gamma$ and $\Gammaphi$, of respective lengths $E$ and $E'$, 
computing respectively $\gg$ and $\bphi$;

 \item the equidimensional algebraic set $V = \V(\gg)$, of dimension 
$d=n-c$, defined by $\gg$;
 
 \item \ZDPs $\scrS$ and $\scrQ''$, of respective degrees $\sigma$ and 
$\kappa''$, describing finite sets $S\subset \CC^n$ and $Q'' \subset 
\CC$, such that $\sing(V) \subset S$;

 \item an atlas $\bchi$ of $(V,S)$, given by \cite[Lemma A.13]{SS2017}, as $S$ 
is finite and contains $\sing(V)$.
\end{itemize}

Let $\Inci$ be the incidence isomorphism associated to $\bphi$ and let $\ggphi$ 
as defined in Lemma~\ref{lem:IVphi}, so that 
$
 \Vt := \V(\ggphi) = \Inci(V).
$
According to Lemmas~\ref{lem:IVphi} and \ref{lem:InciPolar}, $\Vt\subset 
\CC^{2+n}$ is equidimensional with finitely many singular points.

\begin{lemma}\label{lem:LtGNF}
 Let $\scrY = (Y_1,\dotsc,Y_r)$ be algebraic sets in $\CC^n$.
 There exists an algorithm such that, on input $\Gamma,\scrS$ and 
$\Gammaphi$, runs using at most $\softOh{E'\sigma}$ operations in $\QQ$, and 
outputs 
\begin{itemize}
 \item $\Gammat$, a \SLP of length $E+E'+2$, computing $\ggphi$,
 \item $\scrSt$, a \ZDP of degree $\sigma$, encoding $\St=\Inci(S)$,
\end{itemize}
such that the Lagrange system $\Lt = (\Gammat, (1), \scrSt)$ of type 
$((2+n,0),(2+c,0),0)$ defines $\Vt$, and $(\Lt,\scrY)$ has a \GNF.
\end{lemma}
\begin{proof}
By Lemmas~\ref{lem:incvar} and \ref{lem:incparam}, there exist 
algorithms \IncVar and \IncParam respectively, which, on input 
$\Gamma,\scrS$ and $\Gammaphi$, output $\Gammat$ and $\scrSt$ as described in 
the statement, using at most $\softOh{E'\sigma}$ operations in $\QQ$.
Let $\Lt = (\Gammat, (1), \scrSt)$. By Lemma~\ref{lem:IVphi}, $\ggphi$ is a 
\RRS as $\gg$ is. Then, according to \cite[Proposition 5.10]{SS2017}, $\Lt$ 
defines a Lagrange system that defines $\Vt$ and $\bpsi=((1,1,\ggphi,\ggphi))$ 
is a \GNF of $(\Lt,\scrY)$.
\end{proof}

We deduce an algorithm for computing critical points on $V$.

\begin{myproof}{Lemma~\ref{lem:crit}}
By Lemmas~\ref{lem:IVphi}, \ref{lem:diagprojphi} and 
\ref{lem:InciPolar}, $\Wphii[1]$ can be obtained by 
projecting the incidence polar variety $\Wproji[1][\Vt]$ on the last $n$ 
coordinates.
Computing a parametrization of the latter set can then be done using the 
algorithm \Wun of \cite[Proposition 6.3]{SS2017} on the Lagrange system 
given by \cite[Proposition 5.10]{SS2017}.

According to Lemma~\ref{lem:LtGNF}, we can compute a Lagrange system $\Lt$ of 
type $((2+n,0),(2+c,0),0)$, with the \GNFP, that defines $\Vt$.
Hence, by \cite[Proposition 6.4]{SS2017}, there exists a Monte Carlo 
algorithm \Wun which, on input $\Lt$, either fails or returns a \ZDP 
$\scrWunt$ which describes it using at most
\[
\softOh{(E+E')(n+2)^{4d+8}D^{2n+3}(D-1)^{2d}+n\sigma^2}
\]
operations in $\QQ$.
Moreover, in case of success, $\scrWunt$ describes $\Wproji[1][\Vt]-\St$, with 
the notation of Lemma~\ref{lem:LtGNF}.
Besides, by \cite[Proposition I.1]{SS2017} (or \cite[Proposition 3]{SS2018}) 
the degree of $\Kproji[1][\Vt]$ is upper bounded by
\[
 {\textstyle\binom{n+1}{c+1}} D^{c+2}(D-1)^{d}
 ={\textstyle\binom{n+1}{d}} D^{c+2}(D-1)^{d}.
\] 
Finally, by Lemma~\ref{lem:InciPolar}, $\Wphii[1][V]$ can be obtained by
projecting $\Wproji[1][\Vt]$ on the last $n$ coordinates and taking the union
with $S$. This is done by performing the subroutines \Proj and \Union
\cite[Lemma J.3 and J.5]{SS2017}, which uses at most
\[
 \softOh{n^2{\textstyle\binom{n+1}{c+1}^2} D^{2c+4}(D-1)^{2d}+n\sigma^2}
\] 
operations in $\QQ$.
\end{myproof}

In the following, we consider the polar varieties $W=\Wphii[2][V]$ and 
$\Wt=\Wproji[2][\Vt]$ so that, by Lemma~\ref{lem:InciPolar}, $\Wt = \Inci(W)$.

\begin{lemma}\label{lem:LtWGNF}
  Let $\scrY = (Y_1,\dotsc,Y_r)$ be algebraic sets in $\CC^n$. There exists a
  Monte Carlo algorithm which, on input $\Gamma,\scrS$ and $\Gammaphi$, runs
  using at most $\softOh{E'\sigma + n(E+E')}$
operations in $\QQ$, and outputs a Lagrange system $\LtW$ of type 
\[
\big(\,(2+n,\,2+c),\;(2+c,\,n+1),\;0\big).
\]
Either $W$ is empty or assume that $W$ is 1-equidimensional, with
$\sing(W)\subset S$, and $\Watlas(\bchi,V,S,\bphi,2)$ is an atlas of $(W,S)$.
Then, in case of success, $\LtW$ defines $\Wproji[2][\Vt]$ and $(\LtW,\scrY)$
has a \GNF.
\end{lemma}
\begin{proof}
According to Lemma~\ref{lem:LtGNF}, one can compute, using $\softOh{E'\sigma}$ 
operations in $\QQ$, a Lagrange system $\Lt$ of type $((2+n,0),(2+c,0),0)$, 
defining $\Vt$, and such that $(\Lt, (\Wt,\scrY))$ has a \GNF $\bpsi$.

Let $\uu$ be an arbitrary element of $\QQ^{c+2}$ (such an element can be 
provided by the procedure \Rand we mentioned in 
Subsection~\ref{ssec:algodescription}) and
let $\LtW=\WLag(\Lt,\uu,2)$. According to Definition~\ref{def:Wlagrange}, 
$\LtW$ is a Lagrange system of 
type 
\[
\big(\,(2+n,\,2+c),\;(2+c,\,n+1),\;0\big).
\]
Computing $\LtW$ boils down to apply Baur-Strassen's algorithm \cite{BaSt} to 
obtain a \SLP evaluating the Jacobian matrix associated to $\gg, \bphi$
as in the proof of \cite[Lemma O.1]{SS2017}.

By assumption, either $W$ is empty, and so is $\Wt$, or $W$ is equidimensional 
of dimension 1, with $\sing(W) \subset S$.
Then, by Lemma~\ref{lem:InciPolar}, $\Wt$ is equidimensional of dimension 1, 
with $\sing(\Wt) \subset \Inci(S) = \St$.
Moreover, as $\Watlas(\bchi,V,S,\bphi,2)$ is an atlas of $(W,S)$ then, by 
Lemma~\ref{lem:InciPolaratlas}, $\Watlas(\bchiphi,\Vt,\St,\proj,2)$ is an atlas 
of $(\Wt,\St)$.

Therefore, by Proposition~\ref{prop:polarGNF}, there exists a \NEZO subset 
$\scrI(\Lt,\bpsi,\scrY)$ of $\CC^p$ such that, if $\uu \in 
\scrI(\Lt,\bpsi,\scrY)$ then, either $\Wt\neq \emptyset$ or $(\LtW,\scrY)$ 
admits a \GNF. In both cases, $\LtW$ is a Lagrange system that defines $\Wt$.
\end{proof}

\begin{myproof}{Lemma~\ref{lem:solvepolar}}
According to  Lemmas~\ref{lem:IVphi}, \ref{lem:diagprojphi} and 
\ref{lem:InciPolar}, $\Wphii[2]$ can be obtained by 
projecting the incidence polar variety $\Wproji[2][\Vt]$ on the last $n$ 
coordinates.
Computing a parametrization of the latter set can then be done using the 
algorithm \SolveLag of \cite[Proposition 6.3]{SS2017} on the Lagrange system 
given by Proposition~\ref{prop:polarGNF}.

By Lemma~\ref{lem:LtWGNF}, we can compute a Lagrange system $\LtW$ 
defining $\Wproji[2][\Vt]$, that admits a \GNF.
Then, by \cite[Proposition 6.3]{SS2017}, there exists a Monte Carlo algorithm 
\SolveLag which, on input $\LtW$, either fails or returns a \ODP $\scrWt$ of 
degree at most 
\[
 \delta=\deltavalue,
\]
describing $\ULbar[\LtW]$, which is exactly $\Wt$ by 
Proposition~\ref{prop:polarGNF}. 
Moreover, by \cite[Proposition 6.3]{SS2017}, the execution of \SolveLag uses at 
most
\[
\softOh{(n+c)^3(E+E'+(n+c)^3)D\delta^3+(n+c)\delta\sigma^2}
\]
operations in $\QQ$. 
Finally, by Lemma~\ref{lem:InciPolar}, $W$ can be obtained 
by projecting $\Wt$ on the last $n$ coordinates. Hence, running \Proj, with 
input $\scrWt$ and $n$, we get a \ODP $\scrW$, of degree at most $\delta$, 
encoding $W$. According to \cite[Lemma J.9]{SS2017}, the latter operation 
costs at most $\softOh{n^2\delta^3}$ operations in $\QQ$.
\end{myproof}

\begin{myproof}{Lemma~\ref{lem:critpolar}}
By Lemma~\ref{lem:LtWGNF}, we can compute a Lagrange system $\LtW$ 
defining $\Wproji[2][\Vt]$, such that $(\LtW; \Wproji[1][\Wt])$ has the \GNFP.
Hence, by \cite[Proposition 6.4]{SS2017}, there exists a Monte Carlo 
algorithm \Wun which, on input $\LtW$, either fails or returns a \ZDP 
$\scrKt$ of degree at most $\delta(n+c)D$, where
\[
 \delta=\deltavalue,
\]
describing $\Wproji[1][\ULbar[\LtW]]-\St$, which is exactly 
$\Wproji[1][\Wt]-\St$ by Proposition~\ref{prop:polarGNF}. 
Moreover, by \cite[Proposition 6.3]{SS2017}, the execution of \Wun uses at most
\[
\softOh{(n+c)^{12}(E+E')D^3\delta^2+(n+c)\sigma^2}
\]
operations in $\QQ$. 
Finally, by Lemma~\ref{lem:InciPolar}, $\Wphii[1][W]$ can be obtained 
by projecting $\Wproji[1][\Wt]$ on the last $n$ coordinates and taking the 
union with $S$. This is done using the subroutines \Proj and \Union which, 
according to \cite[Lemma J.3 and J.5]{SS2017}, use at most $\softOh{(n+c)^4 D^2 
\delta^2+n\sigma^2}$ operations in $\QQ$.
\end{myproof}

\begin{myproof}{Lemma~\ref{lem:fiberpolar}}
  By Lemma~\ref{lem:LtWGNF}, we can compute a Lagrange system
  $\LtW$ defining $\Wproji[2][\Vt]$, such that $(\LtW; \Wt \cap
  \bm{\pi}_1^{-1}(\widetilde{Q''}))$ has the \GNFP. Hence, by \cite[Proposition
  6.5]{SS2017}, there exists a Monte Carlo algorithm \Fiber which, on input
  $\LtW$, either fails or returns a \ZDP $\scrFt$ of degree at most
  $\kappa''\delta$ where
\[
 \delta=\deltavalue,
\]
describing 
$[\ULbar[\LtW] \cap \bm{\pi}_1^{-1}(\widetilde{Q''})] - \St$, which is exactly 
$[\Wt \cap \bm{\pi}_1^{-1}(\widetilde{Q''})] - \St$ by 
Proposition~\ref{prop:polarGNF}. 
Moreover, by \cite[Proposition 6.3]{SS2017}, the execution of \FiberPolar uses 
at most
\[ 
\softOh{(n+c)^4\big[E+E'+(n+c)^2\big]D(\kappa'')^2\delta^2 + 
(n+c)\sigma^2}
\]
operations in $\QQ$,according to \cite[Definition 6.1]{SS2017}.
Finally, by Lemma~\ref{lem:InciPolar}, $W \cap \mapunrec(Q'')$ 
can be obtained by projecting $\Wt \cap \bm{\pi}_1^{-1}(\widetilde{Q''})$ on 
the last $n$ coordinates and taking the union with $S$.. This is done, using 
the subroutines \Proj and \Union which, according to \cite[Lemma J.3 and 
J.5]{SS2017}, use at most $\softOh{(n+c)^2(\kappa'')^2\delta^2 + n\sigma^2}$ 
operations.
\end{myproof}

\subsection{Proof of 
Proposition~\ref{prop:comp-broadmap}}\label{ssec:boundedrmp}

This paragraph is devoted to prove Proposition~\ref{prop:comp-broadmap}. We
recall its statement below.
\begin{proposition*}[\ref{prop:comp-broadmap}]
Let $\Gamma$ and $\Gammaphi$ be \SLPs, of respective length $E$ and 
$E'$, computing polynomials $\gg = (g_1, \ldots, g_p)$ and
$\map=(\phi_1,\dotsc,\phi_n)$ in $\QQ[x_1, \ldots, x_n]$, of degrees bounded 
by $D$. Assume that $\gg$ satisfies \ref{ass:A}.
Let $\scrQ$ and $\scrS_Q$ be \ZDPs of respective degrees $\degQ$ 
and $\degS$ that encode finite sets $Q \subset \CC^{e}$ (for some 
$0<e\leq n$) and $S_Q \subset \CC^n$, respectively. 
Let $V = \V(\gg)$ and $F_Q = \mapfbr{V}[Q][e]$, and assume 
that
  \begin{itemize}
   \item $F_Q$ is equidimensional of dimension $d-e$, where $d=n-p$;
  \item $\Fatlas(\bchi,V,Q,\bphi)$ is an atlas of 
$(F_Q,S_Q)$, and $\sing(F_Q) \subset S_Q$;
  \item the real algebraic set $F_Q\cap 
\RR^n$ is bounded.
  \end{itemize}
  Consider additionally a zero-dimensional parametrization $\scrP$ of
  degree $\degP$ encoding a finite subset  $\Pcal$ of $F_Q$, which contains 
$S_Q$. Assume that $\degS \leq ((n+e)D)^{n+e}$.
  
  There exists a probabilistic algorithm $\RMBound$ which takes as input
  $((\Gamma, \Gammaphi, \scrQ,\scrS), \scrP)$ and which, in 
case of success, outputs a roadmap of 
  $(F_Q, \Pcal)$,
  of degree
  \[\softOh{
    {(\degP + \degQ) 16^{3d_F}
    (n_F\logde{n_F})^{2(2d_F + 12\logde{d_F}) (\logde{d_F} + 5)}
    D^{(2n_F+1)(\logde{d_F} + 3)}}},
\]
 where $n_F = n+e$ and $d_F = d-e$, and using
\[
  \softOh{
    {\mu'^3 16^{9d_F} E''
(n_F\logde{n_F})^{6(2d_F + 12\logde{d_F}) (\logde{d_F} + 6)}
D^{(6n_F+3)(\logde{d_F} + 4)}}
}
\]
arithmetic operations in $\QQ$ where $\mu' = (\degP + \degQ)$ and $E''=E+E'+e$.
\end{proposition*}

We start by proving a variant of this result \rev{that applies} when $\map$ encodes
projections. Then, using incidence varieties and the associated subroutines, we
 generalize it to arbitrary polynomial maps.

\subsubsection{The particular case of projections}

We study here algorithm \RMRecLag from \cite[Section~7.1]{SS2017}. It
takes as input a Lagrange system $L_\rho =
(\Gamma_\rho,\scrQ_\rho,\scrS_\rho)$ having the \GNFP, and a
zero-dimensional parametrization $\scrP_\rho$, where
$\Zparam(\scrQ_\rho)$ lies in $\CC^{e_\rho}$, for some $e_\rho > 0$;
the output is a roadmap for the algebraic set defined by $L_\rho$, and
$\Zeroes{\scrP_\rho}$. The following proposition ensures correction
and \rev{establishes runtime}. The discussion is entirely similar to
that of \cite[Proposition O.7]{SS2017}, but the analysis done there
assumed that $\scrQ_\rho$ was empty and had $e_\rho = 0$ (the notation
we use, with objects subscripted by $\rho$, is directly taken from there, in order to
facilitate the comparison). In what
follows, let $\XXreci$, where $\nrec\geq 0$, be new indeterminates.

\begin{proposition}\label{prop:roadmaprec}
Let $\ff = (f_1, \ldots, f_{p_\rho})\subset \QQ[\XXreci]$ be
  given by a straight-line program $\Gamma_\rho$ of length $E_\rho$ with
  $\deg(f_i) \leq D$ for $1\leq i \leq p_\rho$, let $\scrQ_\rho$ and 
$\scrS_\rho$ be
  \ZDPs which have respective degrees $\degQ_\rho$ and $\degS_\rho$ and encode 
finitely many points in respectively $\CC^{e_\rho}$ (for some $e_\rho>0$) and 
in $\CC^\nrec$. Assume that
  the Lagrange system $L_\rho = (\Gamma_\rho,\scrQ_\rho,\scrS_\rho)$ has the 
\GNFP. Let $d_\rho = \nrec - p_\rho - e_\rho$, hence the dimension of 
$\Var{\Gamma_\rho}{\scrQ_\rho}{e_\rho}$.

Consider a \ZDP $\scrP_\rho$ of degree $\degP_\rho$ such that 
$\Zeroes{\scrP_\rho}$ is a finite subset of 
$\Var{\Gamma_\rho}{\scrQ_\rho}{e_{\rho}}$ which contains $\Zeroes{\scrS_\rho}$. 
 Assume that $\degS_\rho\leq (\nrec D)^\nrec$.
  
  There exists a Monte Carlo algorithm $\RMRecLag$ which takes as input
  $((\Gamma_\rho, \scrQ_\rho,\scrS_\rho), \scrP_\rho)$ and which, in case of
  success, outputs a roadmap for $(\Var{\Gamma_\rho}{\scrQ_\rho}{e_{\rho}}, 
\scrP_\rho)$ of degree
\[O\tilde{~}\left( (\degP_\rho + \degQ_\rho) 16^{3d_\rho}
    (\nrec\logde{\nrec})^{2(2d + 12\logde{d_\rho}) (\logde{d_\rho} + 5)}
    D^{(2\nrec+1)(\logde{d_\rho} + 3)}\right)\]
  using
  \[O\tilde{~}\left( (\degP_\rho + \degQ_\rho)^3 16^{9d_\rho} E_\rho 
(\nrec\logde{\nrec})^{(12d + 24\logde{d_\rho}) (\logde{d_\rho} + 6)}
D^{(6\nrec+3)(\logde{d_\rho} + 4)}\right)\]
arithmetic operations in $\QQ$.
\end{proposition}

\begin{proof}
Since, by assumption, $L_\rho$ has the \GNFP, one can call the algorithm 
$\textsf{RoadmapRecLagrange}$ from~\cite[Section 7.1]{SS2017} on input 
$L_\rho=(\Gamma_\rho, \scrQ_\rho, \scrS_\rho)$ and $\scrP_\rho$.
This algorithm computes data structures, which are called generalized Lagrange
systems, that encode:
\begin{itemize}
\item a polar variety in $\Var{\Gamma_\rho}{\scrQ_\rho}{e_\rho}$ of dimension
  $\tilde{d}-1\simeq d_\rho / 2$ for $\tilde{d} = \lfloor \frac{d_\rho+3}{2}
  \rfloor$;
\item appropriate fibers in $\Var{\Gamma_\rho}{\scrQ_\rho}{e_\rho}$ of 
dimension $d_\rho  - (\tilde{d} - 1)\simeq d_\rho / 2$.
\end{itemize}
A generalized Lagrange system (see \cite[Definition 5.3]{SS2017}) is encoded by a
triplet $L = (\Gamma, \scrQ, \scrS)$ such that $\Gamma$ is a straight-line
program that evaluates some polynomials, say $\bm{F} = (\bm{f}, \bm{f}_1,\ldots,
\bm{f}_s)$ where
\begin{itemize}
\item $\bm{f}$ lies in $\QQ[\XXrec]$, with $\XXrec = (\XXreci)$;
\item $\bm{f}_i$ lies in $\QQ[\XXrec, \bm{L}_1, \ldots, \bm{L}_i]$ and has
  length $p_i$, where the $\bm{L}_j$'s are sequences of extra variables of 
length   $\nrec_j$ (these are called blocks of Lagrange multipliers);
\item for any $f_{i,j}$ in $\bm{f}_i$, the degree of $f_{i,j}$ in $\bm{L}_j$ is
  at most $1$ for $1\leq i \leq p_i$ and $1\leq j \leq i$.
\end{itemize}
Also, $\scrQ$ (resp. $\scrS$) is a zero-dimensional parametrization encoding
points in $\CC^e$ (resp. $\CC^\nrec$).

\smallskip
The algebraic set of $\CC^\nrec$ defined by $L = (\Gamma, \scrQ, \scrS)$ is the
Zariski closure of the projection on the $\XXrec$-space of
$\Var{\bm{F}}{\scrQ}{\XXrec,e}\setminus \pi_{\XXrec}^{-1}(\Zeroes{\scrS})$. 

\medskip
\paragraph*{Short description of $\textsf{RoadmapRecLagrange}$}
From a generalized Lagrange system $L$ satisfying the global normal form
property and encoding some algebraic set $X$, one can build a generalized
Lagrange system encoding a polar variety $W$ over $X$ using \cite[Definition 5.11 and
Proposition 5.13]{SS2017}, which satisfies the global normal form property, up to some
generic enough linear change of coordinates and some restriction on the
dimension of $W$. Additionally, given finitely many base points $Q'\subset
\CC^{e'}$ encoded by a zero-dimensional parametrization $\scrQ'$, \cite[Definition
5.14 and Proposition 5.16]{SS2017} show how to deduce from $L$ and $\scrQ'$ a
generalized Lagrange system for $\projfbr{X}[Q'][e']$ satisfying the global
normal form property, again assuming the coordinate system is generic enough.

Maintaining the global normal form property allows us to call recursively
$\textsf{RoadmapRecLagrange}$. All in all, these computations are organised in a
binary tree $\mathcal{T}$, whose root is denoted by $\rho$. Each child node
$\tau$ encodes computations performed by a recursive call with input some
generalized Lagrange system $L_\tau = (\Gamma_\tau, \scrQ_\tau, \scrS_\tau)$ and
some zero-dimensional parametrization $\scrP_\tau$ encoding some control points.
Both $L_\tau$ and $\scrP_\tau$ have been computed by the parent node.
Correctness is proved in \cite[Section N.3]{SS2017}. Further, we denote by
$\degQ_\tau$, $\degS_\tau$ and $\degP_\tau$ the respective degrees of
$\scrQ_\tau$, $\scrS_\tau$ and $\scrP_\tau$.

The dimension of  $\Vlag{L_\tau}$ is denoted by $d_\tau$. Calling 
$\textsf{RoadmapRecLagrange}$ with input $L_\tau$ sets $\tilde{d}_\tau =
\lfloor \frac{d_\tau + 3}{2} \rfloor$ and computes
\begin{itemize}
\item[(a)] a generalized Lagrange system $L'_\tau$ which encodes the polar 
variety $W
  = \Polar(e_\tau, d'_\tau, \Vlag{L_\tau}^\AA)$, where $\AA$ is randomly chosen;
\item[(b)] a zero-dimensional parametrization $\scrB_\tau$ which encodes
  the union of $\Zeroes{\scrP}^\AA$ with $\Polar(e_\tau, 1, W)$; we denote its
  degree by $\degB_\tau$; note that by construction (see \cite[]{SS2017},
  $\Zeroes{\scrB_\tau}$ contains $\Zeroes{\scrS_\tau}$); 
\item[(c)] a zero-dimensional parametrization $\scrQ''_\tau$ which encodes the
  projection of $\scrB_\tau$ on the $e_{\tau}''$ first coordinates (with
  $e_{\tau}'' = e_\tau + \tilde{d}_\tau - 1$ ); we denote its degree by
  $\degQ''_\tau$;

\item[(d)] a \ZDP $\scrP'_\tau$ encoding
  $\Zeroes{\scrP_\tau}^\AA\cup Y_\tau$ with $Y_\tau =
  \Var{\Vlag{L'_\tau}}{\scrQ''_\tau}{e_{\tau}''}$ and a zero-dimensional 
  parametrization $\scrP''_\tau$ which encodes those points of 
  $\Zeroes{\scrP'_\tau}$ which project on $\Zeroes{\scrQ''_\tau}$ ; further we 
  denote their degrees by $\degP'_\tau$ and $\degP''_\tau$, the degree of 
  $Y_\tau$ will be denoted by $\degY_\tau$;

\item[(e)] zero-dimensional parametrizations $\scrS'_\tau$ and $\scrS''_\tau$ 
  of respective degrees $\degS'_\tau$ and $\degS''_\tau$ which do encode
  $\Zeroes{\scrS_\tau}^\AA\cup Y_\tau$ and those points of 
$\Zeroes{\scrS''_\tau}$
  which project on $\Zeroes{\scrQ''_\tau}$; note that by construction,
  $\Zeroes{\scrS'_\tau}$ and $\Zeroes{\scrS''_\tau}$ are contained in
  $\Zeroes{\scrP'_\tau}$ and $\Zeroes{\scrP''_\tau}$ respectively;

\item[(f)] and a generalized Lagrange system $L''_\tau$ which encodes
$\projfbr{\Vlag{L_\tau}}[\Zeroes{\scrQ''_\tau}][e_{\tau}'']$. 
\end{itemize}
The recursive calls of $\textsf{RoadmapRecLagrange}$ are then performed on
$(L'_\tau, \scrP'_\tau)$ and $(L''_\tau, \scrP''_\tau)$.

For a given generalized Lagrange system $L_\tau$ corresponding to some node
$\tau$, the number of blocks of Lagrange multipliers is denoted by $k_\tau$. The
total number of variables (resp. polynomials) lying in $\QQ[\XXrec, \bm{L}_1,
\ldots, \bm{L}_i]$ for $i \leq k_\tau$ is denoted by $\Nrec_{i, \tau}$ (resp. 
$P_{i,  \tau}$). By construction, for $i = 0$, we have $P_{0, \tau} = 
p_{\rho}$. For $i = k_\tau$, we denote $\Nrec_{k_\tau, \tau}$ (resp. 
$P_{k_\tau, \tau}$) by $\Nrec_{\tau}$ (resp. $P_{\tau}$).

As in \cite[Section 6.1]{SS2017}, we attach to each such generalized Lagrange
system the quantity
\[
  \delta_\tau = (P_\tau + 1)^{k_\tau} D^p (D-1)^{\nrec - e_\tau 
    -p_\rho}
  \prod_{i=0}^{k_\tau - 1}\Nrec_{i+1, \tau}^{\Nrec_{i, \tau} - e_\tau - P_{i, 
      \tau}}. 
\]
We establish below that the degree of $\Vlag{L_\tau}$ is bounded by
$\degQ_\tau\delta_\tau$. 

\paragraph*{Complexity analysis}
The complexity of $\textsf{RoadmapRecLagrange}$ is analysed in \cite[Section
O]{SS2017}, assuming that $e_\rho = 0$ (see \cite[Proposition O.7]{SS2017}). This is 
done by proceeding in two steps:
\begin{itemize}
\item {\itshape Step (i)} proves some elementary bounds on the number of 
variables
  and polynomials (the $\nrec_i$'s and the $p_i$'s) involved in the 
data-structures
  encoding these polar varieties and fibers in the recursive calls (see
  \cite[Section O.1]{SS2017});
\item {\itshape Step (ii)} proves uniform degree bounds for the parametrizations
  $\scrP'_\tau, \scrP''_\tau$, $\scrB_\tau$, $\scrQ'_\tau, \scrQ''_\tau$, as
  well as $\scrS'_\tau, \scrS''_\tau$ where $\tau$ ranges over all nodes of the
  binary tree $\mathcal{T}$. Uniform degree bounds are also given for all
  $\Vlag{L_\tau}$.
  
  These degree bounds are used in combination with the complexity estimates of
  \cite[Section 6.2]{SS2017} for solving generalized Lagrange systems and
  \cite[Sections J.1 and J.2]{SS2017} which do depend polynomially on these bounds
  and the ones established in {\itshape (i)}.
\end{itemize}
Since the total number of nodes is $O(\nrec)$, it suffices to take $\nrec$ 
times the sum of all costs established by {\itshape (ii)}. Hereafter, we 
slightly 
extend this analysis when $e_\rho > 0$, following the same reasoning, which we 
recall step by step by highlighting the main (and tiny) differences.

\smallskip\noindent {\itshape \bfseries Step {\itshape (i)}.} Both
\cite[Lemma O.1]{SS2017} and \cite[Lemma O.2]{SS2017} control the lengths of the
straight-line programs, the numbers of blocks of Lagrange multipliers and their
lengths, as well as the numbers of polynomials and total number of variables
remain valid, assuming $e_\rho = 0$. Their proofs are based on how these
quantity evolve when building generalized Lagrange systems encoding polar
varieties and fibers (see \cite[Lemmas 5.12 and 5.15]{SS2017}). This is not
changed in our context where the initial call to $\textsf{RoadmapRecLagrange}$
is done with some base points $\Zeroes{\scrQ_\rho}$ with $e_\rho>0$, because for
each node $\tau$, we take $\tilde{d}_\tau = \lfloor \frac{d_\tau + 3}{2}
\rfloor$ as in \cite{SS2017}. This implies that the conclusions of \cite[Lemma
O.1]{SS2017} and \cite[Lemma O.2]{SS2017} still hold when taking $d_\rho =
\nrec-p_\rho-e_\rho$.

All in all, we deduce that:
\begin{itemize}
\item the maximum number of blocks of Lagrange multipliers and the depth of
  $\mathcal{T}$ are bounded by $\lceil \logde{d_\rho} \rceil$
\item All straight-line programs have length bounded by
  $4\nrec^{4+2\logde{d_\rho}}(E_\rho+\nrec^4)$
\item the total number of variables for the generalized Lagrange system $L_\tau$
  is bounded by $(\nrec^2)^{\frac{d_\rho}{h_\tau} + 1}$ where $h_\tau$ is the 
height of the node $\tau$. 
\end{itemize}

\smallskip\noindent {\itshape \bfseries Step {\itshape (ii)}.} 
The two main quantities to consider are
\[\bdelta = 16^{d_\rho + 2} \nrec^{2d_\rho + 12\logde{d_\rho}}D^\nrec\]
and
\[ \bzeta = (\degP_\rho + \degQ_\rho) 
16^{2(d_\rho+3)}(\nrec\logde{\nrec})^{2(2d_\rho +
    12\logde{d_\rho})}D^{(2\nrec+1)(\logde{d_\rho + 2})}.\]

The first step is to prove that for any node $\tau$, the degree of
$\Vlag{L_\tau}$ is dominated by $\degQ_\tau\bdelta$. Using the global normal
form property, \cite[Propositions 5.13 and 6.2]{SS2017} prove that the degree of
$\Vlag{L_\tau}$ is upper bounded by $\degQ_\tau\delta_\tau$.
Recall that, by definition,
\[
  \delta_\tau = (P_\tau + 1)^{k_\tau} D^p (D-1)^{\nrec - e_\tau -p_\rho}
  \prod_{i=0}^{k_\tau - 1}\Nrec_{i+1, \tau}^{\Nrec_{i, \tau} - e_\tau - P_{i,
      \tau}}.
\]
\cite[Lemma O.4]{SS2017} shows that the above left-hand side quantity is
dominated by $\bdelta$, using the results of Step {\itshape (i)} which we 
proved to
still hold. We then deduce that the degree of $\Vlag{L_\tau}$ is upper bounded 
by
$\degQ_\tau\bdelta$.

\cite[Lemma O.5]{SS2017} establishes recurrence formulas for the quantities
$\degB_\tau$, $\degY_\tau$, $\degP_\tau + \degQ_\tau$ and $\degS_\tau$ when
$\tau$ ranges in the set of nodes of the binary tree $\mathcal{T}$. It states
that, letting $\tau'$ and $\tau''$ be the two children of $\tau$, $\degB_\tau$,
$\degY_\tau$, $\degP_{\tau'} + \degQ_{\tau'}$, $\degP_{\tau''} +
\degQ_{\tau''}$, $\degS_{\tau'}$ and $\degS_{\tau''}$ are bounded above by
$2\bdelta^2\zeta_\tau (\degP_\tau + \degQ_\tau)$ where $\zeta_\tau = \left(
  \nrec^2\log_2(\nrec)D \right)^{\frac{d_\rho}{2^{h_\tau}} + 1}$ (here $h_\tau$
is the
height of $\tau$) in the context of \cite{SS2017} with $e_\rho = 0$ and assuming
that $\Zeroes{\scrS_\tau}$ is contained in $\Zeroes{\scrP_\tau}$ for any node
$\tau$ of $\mathcal{T}$ (this is used to prove the statements on $\degS_\tau,
\degS_{\tau'}$ and $\degS_{\tau''}$). In the context of \cite{SS2017}, we have
$\Zeroes{\scrS_\rho} = \emptyset$. In our context, we still take
$\tilde{d}_\tau = \lfloor \frac{d_\tau + 3}{2} \rfloor$ as in \cite{SS2017},
hence the structure of our binary tree $\mathcal{T}$ is the same as the one in
\cite{SS2017}. Also we assume that $\Zeroes{\scrS_\rho}$ is contained in
$\Zeroes{\scrP_\tau}$ and that its degree is bounded by $(\nrec D)^\nrec$. This 
is enough to transpose the recursion performed in the proof of \cite[Lemma O.5 
and Proposition O.3]{SS2017} and deduce that $\degP_\tau, \degQ_\tau$ and $\degS_\tau$ 
are bounded by $\bzeta$ when $\tau$ ranges over the set of nodes of 
$\mathcal{T}$.

The runtime estimates in \cite[Section O.3]{SS2017} to compute the parametrizations
and generalized Lagrange systems in steps (a) to (f) above are then the same
(they depend on $\bdelta$, $\bzeta$ and the above bounds on deduced at Step 
{\itshape (i)}). The statements of \cite[Lemmas O.8, O.9, O.10 and 
O.11]{SS2017} can
then be applied here {\itshape mutatis mutandis} which, as in \cite[Section
O.3]{SS2017}, allow us to deduce the same statement as \cite[Proposition O.7]{SS2017},
i.e. that the total runtime lies in
\[O\tilde{~}\left((\degP_\rho + \degQ_\rho)^3 16^{9d_\rho} E_\rho
    (\nrec\logde{\nrec})^{6(2d + 12\logde{d_\rho}) (\logde{d_\rho} + 6)}
    D^{3(2\nrec+1)(\logde{d_\rho} + 4)}\right)\]
and outputs a roadmap of degree in
\[O\tilde{~}\left((\degP_\rho + \degQ_\rho) 16^{3d_\rho}
    (\nrec\logde{\nrec})^{2(2d + 12\logde{d_\rho}) (\logde{d_\rho} + 5)}
    D^{(2\nrec+1)(\logde{d_\rho} + 3)}\right).\]
\end{proof}

\subsubsection{Proof of Proposition~\ref{prop:comp-broadmap}}
To prove Proposition~\ref{prop:comp-broadmap}, we now show how to return to the
case of projections from the general one, before calling the procedure 
\RMRecLag,
whose complexity is analysed in Proposition~\ref{prop:roadmaprec}.

Consider the notations introduced in the statement of the proposition.
In the following let $\Inciphi$ be the incidence isomorphism associated 
to $\map[e]$ and let $\ggphi[e]$ as defined in Lemma~\ref{lem:IVphi}, so that 
$\Vt := \V(\ggphi[e]) = \Inciphi(V)$.
According to Lemma~\ref{lem:IVphi} and \ref{lem:InciPolar}, $\Vt\subset 
\CC^{e+n}$ is equidimensional with finitely many singular points.
Additionally, let $\FQt = \Inciphi(F_Q)$ and $\SQt = \Inciphi(S_Q)$, so that 
$\FQt = \projfbr{\Vt}$, according to Lemma~\ref{lem:diagprojphi}.

\begin{lemma}\label{lem:LtFGNF}
 There exists an algorithm such that, on input 
$\Gamma$, $\Gammaphi$, $\scrQ$ and $\scrS$ as above, runs using at most 
$\softOh{E'\sigma}$ operations in $\QQ$, and outputs a Lagrange system $\LtF$ 
of type 
\[
\big(\,(e+n,0),\;(e+c,0),\;e\big).
\]
Under the assumptions of Proposition~\ref{prop:comp-broadmap}, $\LtF$ has a
\GNF, and defines $\FQt$.
\end{lemma}
\begin{proof}
According to Lemma~\ref{lem:LtGNF}, we can compute a Lagrange system $\Lt$ of 
type $((e+n,0),(e+c,0),0)$, with the \GNFP, that defines $\Vt$.
 Let $\LtF=\Flag(\Lt,\scrQ,\scrS)$, as defined in 
Definition~\ref{def:Flagrange}, it is a Lagrange system of type 
$\left((\,(e+n,0),\;(e+c,0),\;e\right)$.

By assumptions of Proposition~\ref{prop:comp-broadmap}, either $F_Q$ is empty, 
and so is $\FQt$, or $F_Q$ is equidimensional of dimension $d-e$, with 
$\sing(F_Q) \subset S_Q$.
Then, by Lemma~\ref{lem:InciPolar}, $\FQt$ is equidimensional of dimension 
$d-e$, with $\sing(\FQt) \subset \Inci(S_Q) = \SQt$.
Moreover, as $\Fatlas(\bchi,V,S_Q,\bphi)$ is an atlas of $(F_Q,S_Q)$ then, by 
Lemma~\ref{lem:InciFiberatlas}, $\Fatlas(\bchiphi,\Vt,\SQt,\proj)$ is an 
atlas of $(\FQt,\SQt)$.

Hence, by Proposition~\ref{prop:fiberGNF}, either $\FQt = \emptyset$ or 
$\LtF$ admits a \GNF. 
\end{proof}

Suppose now that the Lagrange system $\LtF$ given by Lemma~\ref{lem:LtFGNF} 
has been computed. According to Lemma~\ref{lem:incparam}, one can compute a 
\ZDP $\scrPt$, encoding $\Pcalt = \Inciphi(\Pcal)$, within the same complexity 
bound. One checks, by assumption, that $\SQt \subset \Pcalt \subset \FQt$  and 
that $\SQt$ has degree bounded by $((n+e)D)^{n+e}$

Therefore, according to Proposition~\ref{prop:roadmaprec}, with $m=n+e$, there exists 
a Monte Carlo algorithm \RMRecLag which, on input $\LtF$ and $\scrPt$, outputs, 
in case of success, a roadmap $\scrRFQt$ of $(\FQt,\Pcalt)$ of degree
\[\softOh{
    {(\degP + \degQ) 16^{3d_F}
    (n_F\logde{n_F})^{2(2d_F + 12\logde{d_F}) (\logde{d_F} + 5)}
    D^{(2n_F+1)(\logde{d_F} + 3)}}},
\]
 where $n_F = n+e$ and $d_F = d-e$, and using
\[
  \softOh{
    {\mu'^3 16^{9d_F} E''
(n_F\logde{n_F})^{6(2d_F + 12\logde{d_F}) (\logde{d_F} + 6)}
D^{3(2n_F+1)(\logde{d_F} + 4)}}
}
\]
arithmetic operations in $\QQ$ with $\mu' = (\degP + \degQ)$ and $E'' = E+E'+e$.

Finally, let $\mathscr{B}_{\mathrm{RM}}$ be the degree bound, given above, on 
the roadmap $\scrRFQt$ of $(\FQt,\Qt)$ output by \RMRecLag. Then, by 
\cite[Lemma J.9]{SS2017}, one can compute the projection $\scrR_{F_Q}$ of  
$\scrRFQt$ on the last $n$ variables. The complexity of this step is bounded 
by $\softOh{n_F^2\mathscr{B}_{\mathrm{RM}}^3}$, that is bounded by
\[
  \softOh{{(\degP + \degQ)^3
  16^{9d_F} (n_F\logde{n_F})^{6(2d_F + 12\logde{d_F})(\logde{d_F} + 6)}
  D^{3(2n_F+1)(\logde{d_F} + 3)}}}
\]
operations in $\QQ$. Finally,  since $\Inciphi$ is an isomorphism of algebraic 
sets, it induces a one-to-one homeomorphic correspondence between the \SACCs of 
$\FQt\cap\RR^{n_F}$ and $F_Q\cap\RR^{n}$ by \cite[Lemma 2.1]{PSS2024}. 
Therefore, $\scrR_{F_Q}$ is a roadmap of $(F_Q,\Pcal)$.

\section{Proof of Proposition~\ref{prop:fibergenalg}: finiteness of 
fibers}\label{sec:fibergenalg}
We recall the statement of the proposition we address to prove. 
\begin{proposition*}[\ref{prop:fibergenalg}]
Let $V\subset \CC^n$ be a $d$-equidimensional algebraic set with finitely many
  singular points and $\polun$ be in $\CC[\XX]$. Let $2\leq \ifi \leq d+1$. For 
$\balf=(\balf_1,\dots,\balf_{\ifi})$  in $\CC^{\ifi n}$, we define
   $\bphi = \left(\phi_1(\XX,\balf_1),  
\dotsc,\phi_{\ifi}(\XX,\balf_{\ifi})\right)$, where for $2\leq j \leq \ifi$
 \begin{align*}
  \phi_1(\XX,\balf_1) = \polun(\XX) + 
  \sum_{k=1}^n \alf_{1,k}x_k
  \et
  \phi_j(\XX,\balf_j) = \sum_{k=1}^n \alf_{j,k}x_k.
 \end{align*}
  Then, there exists a non-empty Zariski open subset $\ZOnoether(V, \polun,\ifi)
 \subset \CC^{\ifi n}$ such that for every $\balf \in 
\ZOnoether(V,\polun,\ifi)$ and $i \in \{1,\dotsc,\ifi\}$, the following holds:
 \begin{enumerate}
  \item either $\Wphii$ is empty or $(i-1)$-equidimensional;
  \item the restriction of $\map[i-1]$ to $\Wphii$ is a Zariski-closed map;
  \item for any $\zz \in \CC^{i-1}$, the fiber $\Kphii \cap \map[i-1][-1](\zz)$
    is finite.
\end{enumerate}
\end{proposition*}

The rest of this section is devoted to the proof of this result. We first
establish a general lower bound on the dimension of the non-empty generalized 
polar varieties. This is a direct generalization of \cite[Lemma B.5 \& 
B.13]{SS2017}.
\begin{lemma}\label{lem:lowerboundglobal}
 Let $\CCgen$ be an algebraically closed field, and let $V\subset \CCgen^n$ 
be a $d$-equidimensional algebraic set. 
Then, for any $\bvphi = (\vphi_1,\dotsc,\vphi_{d+1}) 
\subset \CCgen[\XX]$, and any $1 \leq i \leq d+1$, all irreducible 
components of $\Wvphii$ have dimension at least $i-1$. 
\end{lemma}
\begin{proof}
  Since $V$ is $d$-equidimensional, the case $i=d+1$ is immediate; assume now 
that $i\leq d$. According to \cite[Lemma A.13]{SS2017}, there exists an atlas 
$\bchi=(\chi_j)_{1\leq j \leq s}$ of  $(V,\,\sing(V))$. For $1 \leq j \leq s$, 
let $\chi_j=(m_j,\hh_j)$. By \cite[Lemma A.12]{SS2017}, $\hh_j$ has cardinality 
$c = n-d$. According to Lemma~\ref{lem:jacrankchart}, fix $j \in 
\{1,\dotsc,s\}$, the following holds in $\Ocal(m_j) - \sing(V)$,
\begin{equation}\label{eqn:localpolarvariety}
  \Wvphii\;=\; \{\yy \in \Voreg(\hh_j) \mid \rank(\jac_\yy(\hh_j,\bvphi_i) < 
c+i\} = \Wovphii[i][\Voreg(\hh_j)].
 \end{equation}
 Let $\yy \in \Wovphii = \Wvphii - \sing(V)$. Since $\yy \in V$, there exists $j
 \in \{1,\dotsc,s\}$ such that $\yy \in \Ocal(m_j)$. Hence, by
 \eqref{eqn:localpolarvariety}, in $\Ocal(m_j) - \sing(V)$, the irreducible
 component of $\Wvphii$ containing $\yy$ is the same as the irreducible 
component
 of the Zariski closure of $\Wovphii[i][\Voreg(\hh_j)]$ containing $\yy$. Since
 these irreducible components are equal over a non-empty Zariski open set, they
 have same dimension by \cite[Theorem 1.19]{Sh1994}. Hence, proving that this
 common dimension is at least $i-1$ allows us to conclude.
 
Let $\Iminor \subset \CCgen[\XX]$ be the ideal generated by the $(c+i)$-minors 
of 
$\jac[\hh_j,\bvphi_i]$.
 Then,
 \[
 \Wovphii[i][\Voreg(\hh_j)] = \Voreg(\hh_j) \cap \V(\Iminor)
 \]
 which is contained in the algebraic set $\Vreg(\hh_j) \cap \V(\Iminor)$.
 We assume that $\Vreg(\hh_j) \cap \V(\Iminor)$ is not empty otherwise the 
statement of the proposition trivially holds.

Note that any irreducible component $Z$ of $\Vreg(\hh_j) \cap \V(\Iminor)$, has
an ideal of definition $\IPminor$ in $\CCgen[\Vreg(\hh_j)]$ that is an isolated
prime component of the determinantal ideal $\Iminor\cdot\CCgen[\Vreg(\hh_j)]$. 
Then
by \cite[Theorem 3]{EN1962}, $\IPminor$ has height at most $n-c-(i-1)$ so that
the codimension of $Z$ in $\Vreg(\hh_j)$ is at most $n-c-(i-1)$. Since
$\Vreg(\hh_j)$ has dimension $n-c$, the dimension of $Z$ is then at most $i-1$.
 
One concludes by observing that, any irreducible component of the Zariski
closure of $\Wovphii[i][\Voreg(\hh_j)]$ is the union of irreducible components 
of
$\Vreg(\hh_j) \cap \V(\Iminor)$.
\end{proof}

\subsection{An adapted Noether normalization lemma}

Consider an algebraically closed field $\Kclos$, let $\ff=(f_1,\dots,f_m)\colon
\Kclos^n\to\Kclos^m$ be a polynomial map and $V\subset\Kclos^n$ and
let $Y\subset \Kclos^m$ be algebraic sets such that $\ff(V)\subset
Y$. Finally, consider the restriction $\tilde \ff:V \to Y$ of $\ff$,
and recall that the pullback $\tilde{\ff}^*$ of $\tilde{\ff}$ is
defined by
 \[
  \begin{array}{cccc}
\tilde{\ff}^*\colon&\Kclos[Y]=\Kclos[y_1,\dots,y_m]/\I(Y)& \longrightarrow 
&\Kclos[V]=\Kclos[x_1,\dots,x_n]/\I(V)\\[0.5em]
  &g&\longmapsto& g\circ\ff
  \end{array}.
 \]
\begin{definition}[{\cite[Section 5.3]{Sh1994}}]
 We say that the restriction $\tilde{\ff}$ of $\ff$ is a {\em finite
   map} if
\begin{enumerate}
 \item $\ff(V)$ is dense in $Y$, which is equivalent to
   $\tilde{\ff}^*$ being injective;
 \item the extension $\Kclos[Y]\hookrightarrow\Kclos[V]$ induced by
 $\tilde{\ff}^*$ is integral.
\end{enumerate} 
\end{definition}

The following lemma shows that to verify such conditions, we may not
have to work over an algebraically closed field: if $V$ and $Y$ are
defined over a subfield $\KK$ of $\Kclos$, finiteness of $\tilde\ff$
is equivalent to the pullback $\KK[\YY]/\I(Y) \to \KK[\XX]/\I(V)$
being injective and integral.

\begin{lemma}\label{lemma:extension}
  Let $\KK \subset \LL$ be two fields, let $I,J$ be ideals in
  respectively $\KK[\YY]=\KK[y_1,\dots,y_m]$ and
  $\KK[\XX]=\KK[x_1,\dots,x_n]$ and let $I',J'$ be their extensions in
  respectively $\LL[\YY]$ and $\LL[\XX]$. Let finally
  $\ff=(f_1,\dots,f_m)$ be in $\KK[\XX]$, such that for $g$ in $I$, $g
  \circ \ff$ is in $J$.

  Consider the ring homomorphisms
  \[
    \zeta_{\KK}:\KK[\YY]/I
    \to \KK[\XX]/J \text{ and }
    \zeta_{\LL}:\LL[\YY]/I' \to
    \LL[\XX]/J',
  \]
  that both map $y_j$ to $f_j$, for all $j$.  Then,
  $\zeta_{\KK}$ is injective, resp.\ integral, if and only if
  $\zeta_{\LL}$ is.
\end{lemma}
\begin{proof}
  Injectivity of $\zeta_\KK$ is equivalent to equality between ideals $I=(J
  \KK[\YY,\XX] + \langle y_1 -f_1,\dots, y_m-f_m\rangle) \cap \KK[\YY]$;
  similarly, injectivity of $\zeta_{\LL}$ is equivalent to $I'=(J'\LL[\YY,\XX] +
  \langle y_1 -f_1,\dots, y_m-f_m\rangle) \cap \LL[\YY]$. These properties can
  be determined by Gr\"obner basis calculations; since the generators of $I,J$
  are the same as those of $I',J'$, they are thus equivalent.

  Next, integrality of $\zeta_{\KK}$ directly implies that
  of $\zeta_{\LL}$. Conversely, integrality of $\zeta_{\LL}$ is
  equivalent to the existence of polynomials $G_1,\dots,G_n$ in
  $\LL[y_1,\dots,y_m,s]$, all monic in $s$, such that
  $G_j(f_1,\dots,f_m,x_j)$ is in $J'$ for all $j$.
If we assume that such polynomials exist, 
we can then linearize these membership equalities,
reducing such properties to the existence of a solution to certain
linear systems with entries in $\KK$. Since we know that solutions
exist with entries in $\LL$, some must also exist with entries in
$\KK$. This then yields integrality of $\zeta_{\KK}$.
\end{proof}

The Noether normalization lemma says that for $V$ $r$-dimensional and
$Y=\Kclos^r$, the restriction of a generic linear mapping $\Kclos^n
\to \Kclos^m$ to $V$ is finite. We give here a proof of this lemma
adapted to our setting, where the shape of the projections we perform
is made explicit. We start with a statement for ideals rather than
algebraic sets.

\begin{proposition}[Noether normalization]\label{prop:noether}
  Let $\KK$ be a field, let $J$ be an ideal in $\KK[\XX]$ and let $r$
  be the dimension of its zero-set over an algebraic closure of $\KK$. 
  Let further $\aa$ be $r(n-r)$ new indeterminates. Then the
  $\KK(\aa)$-algebra homomorphism
 \[
  \begin{array}{cccc}
  \zeta_{\aa}\colon&\KK(\aa)[\ZZi]& \longrightarrow &\KK(\aa)[\XX]/ 
J\KK(\aa)[\XX] \\[0.5em]
  &z_j&\longmapsto& \displaystyle x_j + \sum_{k=1}^{n-r} a_{j,k}\,x_{r+k} 
\mod{J}
  \end{array}
  \]
  is injective and makes $\KK(\aa)[\XX]/J \KK(\aa)[\XX]$ integral over
  $\KK(\aa)[\ZZi]$.
\end{proposition}
\begin{proof}
We proceed by induction on the number $n$ of variables. The case $n=0$
is straightforward. Assume now that $n>0$ and that the statement holds
for $k<n$ variables. Remark that if $J = \{0\}$, we have $r=n$,
$\zeta_{\aa}$ is the isomorphism $\KK[z_1\dots,z_n]\to\KK[\XX]$
mapping $z_i$ to $x_i$ for all $i$ (which is then integral); in this
case, we are done.

Assume now that $J \neq \{0\}$, and let $f$ be non-zero in $J$. Let
$\delta$ be the total degree of $f$ and let
$\bm{\ell}=(\ell_1,\dotsc,\ell_{n-1})$ be new indeterminates.  Writing
$f=\sum_{i_1,\dotsc,i_n}c_{i_1,\dotsc,i_n}x_1^{i_1}\cdots x_n^{i_n}$,
the leading coefficient of
$f(x_1+\ell_1x_n,\dotsc,x_{n-1}+\ell_{n-1}x_n,x_n)$ in $x_n$ is
 \[
   \sum_{i_1+\cdots+i_{n}=\delta} c_{i_1,\dotsc,i_n} \ell_1^{i_1}\cdots 
\ell_{n-1}^{i_{n-1}} = f_{\delta}(\ell_1,\dotsc,\ell_{n-1},1)
 \]
 where $f_{\delta}$ is the homogeneous degree-$\delta$ component of $f$.
 Therefore, taking $F=f_{\delta}(\ell_1,\dotsc,\ell_{n-1},1)$, $F$ is not 
the zero polynomial and the polynomial
\[
   \frac{1}{F}f(x_1+\ell_1 x_n,\dotsc, x_{n-1}+\ell_{n-1} x_n, x_n) 
\in 
\KK(\bell)[\XX]
 \]
 is monic in $x_n$. Let further $J'$ be the extension of $J$ to
 $\KK(\bell)[\XX]$, let $\YY=(y_1,\dotsc,y_{n-1})$ be new
 indeterminates and consider the $\KK(\bell)$-algebra homomorphism
\[
\begin{array}{cccc}
  \tau\colon&\KK(\bell)[\YY]& \longrightarrow &\KK(\bell)[\XX]\\[0.5em]
  &y_j&\longmapsto& x_j - \ell_j x_n
  \end{array};
\]
the contraction $J'^{c}=\tau^{-1}(J')$ is an ideal in
$\KK(\bell)[\YY]$.  For $1\leq j \leq n-1$, let $\ybar[j]=y_j
\mod{J'^c}$ and for $1\leq k \leq n$ let $\xbar[j]=x_j \mod{J'}$. Then
let
\[
\begin{array}{cccc}
  \taubar\colon&\KK(\bell)[\YY]/J'^c& \longrightarrow
  &\KK(\bell)[\XX]/J'\\[0.5em] &\ybar[j] &\longmapsto& \xbar[j] -
  \ell_j \xbar[n]
  \end{array}
\]
and
\[
  g(s) = 
\frac{1}{F}f\big(\ybar[1]+\ell_1 
s , \dotsc, \ybar[n-1]+\ell_{n-1} s, s\big) \in 
\Big(\KK(\bell)[\YY]/J'^c\Big)[s];
\]
this is a monic polynomial in $s$.
If we extend $[\tau]$ to a $\KK(\bell)$-algebra homomorphism
$\KK(\bell)[\YY]/J'^c\,[s] \to \KK(\bell)[\XX]/J'\,[s]$,
$g$ satisfies
\[
 [\tau](g)\big(\xbar[n]\big) = 
\frac{1}{F} 
f(\xbar[1],\dotsc,\xbar[n]) = 0,
\]
since $f \in J$ by assumption. Since $[\tau]$ is by construction
injective, it makes $\KK(\bell)[\XX]/J'$ an integral extension of
$\KK(\bell)[\YY]/J'^c$ (the integral dependence relation for 
$[x_j]$, for $j < n$, is obtained by replacing $s$ by $(s-[y_j])/\ell_j$ in
$g$ and clearing denominators).

In particular, these two rings have the same Krull
dimension~\cite[Corollary~2.13]{Kunz85}. This latter dimension is the
same as that of $\KK[\XX]/J$ (because it can be read off a Gr\"obner
basis of $J$, and such Gr\"obner bases are also Gr\"obner bases of
$J'$), that is, $r$. In other words, the zero-set of $J'^c$ over an
algebraic closure of $\KK(\bell)$ has dimension $r$.

Then we can apply the induction hypothesis to $J'^c \subset
\KK(\bell)[\YY]$. If we let $\bb=(b_{i,j})_{1 \le i \le r, 1 \le j \le
  n-1-r} $ be $r(n-1-r)$ new indeterminates, and introduce $\ZZ=(\ZZi)$,
the $\KK(\bell,\bb)$-algebra homomorphism
 \[
  \begin{array}{cccc}
  \eta_{\bb} \colon&\KK(\bell,\bb)[\ZZ]& \longrightarrow 
&\KK(\bell,\bb)[\YY]/J'^c \KK(\bell,\bb)[\YY]\\[0.5em]
  &z_j&\longmapsto& \displaystyle \ybar[j] + \sum_{k=1}^{n-1-r} 
b_{j,k}\,\ybar[r+k]
  \end{array}
  \]
 is thus injective and realizes an integral extension of the
 polynomial ring $\KK(\bell,\bb)[\ZZ]$. On the other hand, by
 Lemma~\ref{lemma:extension}, the extended map \[
 \taubar^e:\KK(\bell,\bb)[\YY]/J'^c\KK(\bell,\bb)[\YY]
 \longrightarrow\KK(\bell,\bb)[\XX]/J'\KK(\bell,\bb)[\XX]\] remains
 injective and integral.  By transitivity, it follows that the
 $\KK(\bell,\bb)$-algebra homomorphism
\[
  \begin{array}{cccc}
 \taubar^e \circ\eta_{\bb}  \colon&\KK(\bell,\bb)[\ZZ]& \longrightarrow 
&\KK(\bell,\bb)[\XX]/J' \KK(\bell,\bb)[\XX]\\[0.5em]
  &z_j&\longmapsto& \displaystyle \xbar[j] + \sum_{k=1}^{n-r} 
m_{j,k}\,\xbar[r+k]
  \end{array},
 \]
 where for all $1\leq j \leq r$,
 \[
 \mm_j = \left(b_{j,1},\dotsc, b_{j,n-1-r},\; 
-\ell_j - \sum_{k=1}^{n-1-r} b_{j,k} \, \ell_{r+k}\right),
 \]
 is injective and integral as well. In particular, the restriction of
 $\taubar^e \circ\eta_{\bb}$ to a mapping $\KK(\mm)[\ZZ] \to
 \KK(\mm)[\XX]/J \KK(\mm)[\XX]$ is still injective and integral, by
 Lemma~\ref{lemma:extension} (here, we write 
$\KK(\mm)=\KK(m_{1,1},\dots,m_{r(n-r)})$).
 
 Letting $\aa$ be $r(n-r)$ new indeterminates, we observe that $\iota:
 a_{i,j} \mapsto m_{i,j}$ defines a $\KK$-isomorphism $\KK(\aa) \to
 \KK(\mm) \subset \KK(\bell,\bb)$, since the entries of $\mm$ are
 $\KK$-algebraically independent. The conclusion follows.
 \end{proof}

\begin{corollary}\label{coro:noether}
  Let $V\subset \CC^n$ be an $r$-dimensional algebraic set. Let
  $\aa=(a_{i,j})_{1 \le i \le r, 1 \le j \le n-r}$ be $r(n-r)$ new
  indeterminates and let $\VKclos \subset \Kclos^n$ be the extension
  of $V$ to the algebraic closure $\Kclos$ of $\KK=\CC(\aa)$.  Then
  the restriction $\tilde \ff: \VKclos \to \Kclos^r$ of the polynomial
  map $\ff=(f_1,\dots,f_r)$ given by
  $$f_j =x_j + \sum_{k=1}^{n-r} a_{j,k}\,x_{r+k},\quad 1 \le j \le r $$
  is finite.
\end{corollary}
\begin{proof}
  Let $J$ be the defining ideal of $V$ in $\CC[\XX]$. Letting
  $\ZZ=\ZZi[r]$ be $r$ new indeterminates, the previous proposition
  shows that $\tilde \ff^*: \CC(\aa)[\ZZ] \to \CC(\aa)[\XX]/J
  \CC(\aa)[\XX]$ is injective and integral. By
  Lemma~\ref{lemma:extension}, we further deduce that it is also the
  case for the extension of $\ff^*: \Kclos[\ZZ] \to \Kclos[\XX]/J
  \Kclos[\XX]$.

  Because $\CC$ is algebraically closed, $J$ remains radical in
  $\Kclos[\XX]$, so that $J \Kclos[\XX]$ is the defining ideal of
  $\VKclos$, and we are done.
\end{proof}

\subsection{Finiteness on polar varieties}
In this section, we prove the core of the
Proposition~\ref{prop:fibergenalg}, by proving finiteness properties
on the restriction of the considered morphisms to their associated
polar varieties.

\begin{proposition}\label{prop:finitepolar}
  Let $V \subset \CC^n$ be a $d$-equidimensional algebraic set with finitely 
  many singular points and let $\polun \in \CC[\XX]$. 
  For $\balf=(\balf_1,\dots,\balf_{d+1})$  in $\CC^{(d+1)n}$, and for $2\leq j 
  \leq d+1$, let
  \begin{align*}
    \phi_1(\XX,\balf_1) = \polun(\XX) + 
    \sum_{k=1}^n \alf_{1,k}x_k
    \et
    \phi_j(\XX,\balf_j) = \sum_{k=1}^n \alf_{j,k}x_k.
  \end{align*}
  Then for any $1\leq i \leq d+1$, there exists a non-empty Zariski
  open set $\ZOpen_i\subset \CC^{(d+1)n}$ such that if $\balf \in
  \ZOpen_i$ and $\bphi = \left(\phi_1(\XX,\balf_1), \dotsc,
  \phi_{d+1}(\XX,\balf_{d+1}) \right)$, then the restriction
  of $\bphi_{i-1}$ to $\Wphii$ is a finite map.
\end{proposition}
\begin{proof}
  Let $\aa = (\aa_i)_{1\leq j \leq d+1}$, with $\aa_j =
  (a_{j,1},\dotsc,a_{j,n})$ for all $j$, be $(d+1)n$ new
  indeterminates, and let $\CC(\aa)$ be the field of rational
  fractions in the entries of $\aa$. We let $\Kclos$ be the algebraic
  closure of $\CC(\aa)$, and we denote by $\VKclos\subset\Kclos^n$ the
  extension of $V$ to $\Kclos$.  Let further
  \[
  \vphi_1(\XX,a_1) = \polun(\XX) + 
  \sum_{k=1}^n a_{1,k}x_k
  \et
  \vphi_j(\XX,a_j) = \sum_{k=1}^n a_{j,k}x_k, \quad 2 \le j \le d+1
  \]
  and define $\bvphi=(\vphi_1,\dotsc,\vphi_{d+1})$
  in $\CC(\aa)[\XX]$; as before, for $1 \le i \le d+1$, we write
  $\bvphi_i=(\vphi_1,\dotsc,\vphi_{i})$.
  We will prove the following property, which we call $ \Pcal(i)$, by
  decreasing mathematical induction, for $i=d+1,\dotsc,1$:
  \begin{description}
   \item [$\Pcal(i):$] the restriction of $\bvphi_{i-1}$ to
     $\Wphii[i][\VKclos][\bvphi]$ is a finite map.
  \end{description}
  Let us first see how to deduce the proposition from this claim;
  hence, we start by fixing $i$ in $1,\dots,d+1$ and assume that
  $\Pcal(i)$ holds.

Since $\bvphi_i$ and $\VKclos$ are defined by polynomials with
coefficients in $\CC(\aa)$, it is also the case for
$\Wphii[i][\VKclos][\bvphi]$ by~\cite[Lemma A.2]{SS2017}. Then
$\Pcal(i)$ shows (via the discussion preceding
Lemma~\ref{lemma:extension}) that the pullback $\tilde\bvphi_{i-1}^*:
\CC(\aa)[z_1,\dots,z_{i-1}] \to
\CC(\aa)[\XX]/\I(\Wphii[i][\VKclos][\bvphi])$ is injective and
integral.
  \begin{itemize}
  \item Injectivity means that the ideal generated by
    $\I(\Wphii[i][\VKclos][\bvphi])$ and the polynomials 
    $z_1-\vphi_1,\dots,z_{i-1}-\vphi_{i-1}$ in
    $\CC(\aa)[z_1,\dots,z_{i-1},\XX]$ has a trivial intersection with
    $\CC(\aa)[z_1,\dots,z_{i-1}]$. Then, this remains true for the
    restriction of $\bphi_{i-1}$ to $\Wphii$ for $\balf$ in a
    non-empty Zariski-open set in $\CC^{(d+1)n}$. For instance, it is
    enough to ensure that the numerators and denominators of the
    coefficients of all polynomials appearing in a lexicographic
    Gr\"obner basis computation for the ideal above, in
    $\CC(\aa)[z_1,\dots,z_{i-1},\XX]$, do not vanish at $\balf$.
  \item Integrality means that there exist $n$ monic polynomials
    $P_1,\dots,P_n$ in $\CC(\aa)[z_1,\dots,z_{i-1}][s]$ such that all
    polynomials
    $$P_j\left (\phi_1,\dots,\phi_{i-1},x_j\right), \quad 1 \le j \le
    n$$ belong to $\I(\Wphii[i][\VKclos][\bvphi]$ in $\CC(\aa)[\XX]$. Taking
    $G \in \CC[\aa]$ as the least common multiple of the denominators
    of all coefficients that appear in these membership relations, we
    see that for $\balf$ in $\CC^{(d+1)n}$, if $G(\balf) \ne 0$,
    $\bphi_{i-1}$ makes $\CC[\XX]/ \I(\Wphii[i][V][\bphi])$ integral
    over $\CC[z_1,\dots,z_{i-1}]$.
  \end{itemize}
 
\paragraph{Base case: $\bm{i=d+1}$} We prove
$\Pcal(d+1)$. As $\Tg_\yy \VKclos$ has dimension $d$, for every $\yy
\in \reg(\VKclos)$, the polar variety $\Wphii[d+1][\VKclos][\bvphi]$
is nothing but $\VKclos$ (since the latter only admits finitely many
singular points); hence, we have to prove that the restriction of
$\bvphi_d$ to $\VKclos$ is finite.

Let $y_1,\dots,y_d$ be new variables and consider the algebraic set
$V' \subset \CC^{d+n}$ defined by $y_1-\polun, y_2,\dots,y_d$ and all
polynomials $f$, for $f$ in $\I(V)$; as above, we denote its extension
to $\Kclos^{d+n}$ by $\VKclos'$. Apply Corollary~\ref{coro:noether} to
$V'$ (which is still of dimension $d$): we deduce that 
the restriction of $\bvphi_d$ to $\VKclos'$ is finite. 
Since $\VKclos'$ and $\VKclos$ are isomorphic (since $\VKclos'$ is a graph 
above $\VKclos$), we are done with this case.

\paragraph{Induction step: $\bm{1\leq i\leq d}$}
Assume now that $\Pcal(i+1)$ holds. Thus, the restriction of
${\bvphi_{i}}$ to a mapping $\Wphii[i+1][\VKclos][\bvphi] \to
\Kclos^i$ is finite. By \cite[Theorem 1.12]{Sh1994} this restriction
is a Zariski-closed map so that, since $\Wphii[i][\VKclos][\bvphi]
\subset \Wphii[i+1][\VKclos][\bvphi]$,
$\bvphi_i(\Wphii[i][\VKclos][\bvphi])\subset \Kclos^i$ is an
algebraic set and the restriction of $\bvphi_i$ to a mapping
$\Wphii[i][\VKclos][\bvphi] \to \vphiWvphii$ is finite as well.

Let $\YY=(y_1,\dots,y_i)$ be new indeterminates. Because these sets
are defined over $\CC(\aa)$, we deduce that the pullback
$\CC(\aa)[\YY]/\I(\vphiWvphii) \to
\CC(\aa)[\XX]/\I(\Wphii[i][\VKclos][\bvphi])$ that maps $y_j$ to
$\phi_j$ (for all $j \le i$) is injective and integral
(Lemma~\ref{lemma:extension}).

On another hand, by the theorem on the dimension of the fibers
\cite[Theorem 1.25]{Sh1994}, for any irreducible component $C$ of
$\Wphii[i][\VKclos][\bvphi]$ and for a generic $\yy \in \bvphi_i(C)$,
\[
 \dim{C} - \dim{\bvphi_i(C)} = \dim{\bvphi_i^{-1}(\yy)\cap C} = 0
\]
since, as a finite map, the restriction of $\bvphi_i$ to 
$\Wphii[i][\VKclos][\bvphi]$ 
has finite fibers.
By an algebraic version of Sard's theorem \cite[Proposition 
B.2]{SS2017}
\[
 \dim{\,\vphiWvphii} \leq i-1,
\]
so that $\dim{\Wphii[i][\VKclos][\bvphi]} \leq i-1$ as well. Together with
Lemma~\ref{lem:lowerboundglobal}, this proves that both
$\Wphii[i][\VKclos][\bvphi]$ and its image $\vphiWvphii$ are either empty or
equidimensional of dimension $i-1$. If they are empty, there is
nothing to do, so suppose it is not the case.

Let $\ZZ=(z_1,\dots,z_{i-1})$ and $\bell=(\ell_1,\dotsc,\ell_{i-1})$
be new indeterminates.  Since $\Wphii[i][\VKclos][\bvphi]$, and thus
its image $\vphiWvphii$, are defined over $\CC(\aa)$, we can apply
Noether normalization to $\vphiWvphii$
(Proposition~\ref{prop:noether}) with coefficients in $\CC(\aa)$, and
deduce that the $\CC(\aa,\bell)$-algebra homomorphism
\[
  \begin{array}{cccc}
  \zeta\colon&\CC(\aa,\bell)[\ZZ]& \longrightarrow 
&\CC(\aa,\bell)[\YY]/\I(\vphiWvphii)\\[0.5em]
  &z_j&\longmapsto& \displaystyle y_j + \ell_j y_i 
\mod{\I(\vphiWvphii)}
  \end{array}
 \]
is injective and integral. Besides, we deduce from
Lemma~\ref{lemma:extension} that after scalar extension, the ring
homomorphism
\[
  \CC(\aa,\bell)[\YY]/\I(\vphiWvphii) \to
  \CC(\aa,\bell)[\XX]/\I(\Wphii[i][\VKclos][\bvphi])
\]
that maps $y_j$ to
$\phi_j$ (for all $j \le i$) is still injective and integral.  If we
set \begin{equation}\label{eqn:defpsi}
  \psi_j = \vphi_j + \ell_j \vphi_i \text{ \quad for $1\leq j\leq i-1$}
  \quad\text{and}\quad
  \psi_j = \vphi_j \text{ \quad for $i\leq j \leq d+1$},
\end{equation}
and finally
$\bpsi=(\psi_1,\dotsc,\psi_{d+1})\subset\CC(\aa,\bell)[\XX]$,
then, by transitivity,
\[
  \begin{array}{cccc}
  \bpsi_{i-1} \colon&\CC(\aa,\bell)[\ZZ]& \longrightarrow 
&\CC(\aa,\bell)[\XX]/\I(\Wphii[i][\VKclos][\bvphi])\\[0.5em]
  &z_j&\longmapsto& \displaystyle \psi_j(\XX) 
\mod{\I(\Wphii[i][\VKclos][\bvphi])}
  \end{array}
\]
is injective and integral as well.

Since the first $i$ entries of $\bpsi$ are elementary row operations of the
first $i$ entries of $\bvphi$, we deduce that
$\Wphii[i][\VKclos][\bvphi]=\Wphii[i][\VKclos][\bpsi]$.
Besides, injecting the definition of the $\vphi_j$'s in \eqref{eqn:defpsi}, one 
gets that $\bpsi(\XX) = 
\bvphi(\XX,\mm)$, where 
\[   \mm = \left(\aa_1+\ell_1 \aa_i,\dotsc,\aa_{i-1}+\ell_{i-1}
\aa_i,\aa_i,\dotsc, \aa_{d+1}\right)
\]
is a vector of $(d+1)n$ $\CC$-algebraically independent elements of 
$\CC(\aa,\bell)$.
Through the isomorphism $\CC(\aa,\bell) \to \CC(\mm)$, we see that
\[
  \begin{array}{cccc}
  \bvphi_{i-1} \colon&\CC(\aa,\bell)[\ZZ]& \longrightarrow 
&\CC(\aa,\bell)[\XX]/\I(\Wphii[i][\VKclos][\bvphi])\\[0.5em]
  &z_j&\longmapsto& \displaystyle \vphi_j(\XX) 
\mod{\I(\Wphii[i][\VKclos][\bvphi])}
  \end{array}
\]
is injective and integral. From Lemma~\ref{lemma:extension}, we see
that this precisely gives that the restriction of $\bvphi_{i-1}$ to
$\Wvphii[i][\VKclos][\bvphi]$ is finite.  This ends the proof of the
induction step, and, by mathematical induction, of the proposition.
\end{proof}

\subsection{Proof of the main proposition}
We conclude by proving Proposition~\ref{prop:fibergenalg}, which is a direct 
consequence of the previous results. Let $V$, $\polun$ and $2\leq \ifi \leq 
d+1$ as given in the statement of the proposition. 

Let $\Omega$ be the \NEZO subset of $\CC^{(d+1)n}$ obtained as the 
intersection, for all $1\leq i \leq d+1$, of the $\ZOpen_i$'s given by 
application of Proposition~\ref{prop:finitepolar}.
Let $\aa=(\aa_i)_{1\leq i \leq d+1}$, where $\aa_i=(a_{i,1},\dotsc,a_{i,n})$ be 
$(d+1)n$ new indeterminates. By definition, there exists 
$\ff=(f_{1},\dotsc,f_{p}) \subset \CC[\aa]$, such that $\Omega = \CC^{(d+1)n} - 
\V(\ff)$.
Then, let $\ZOnoether(V,\polun,\ifi)$ be the projection on the first $\ifi n$ 
coordinates of $\Omega$, it is the union, for all $\balf'' \in 
\CC^{(d+1-\ifi)n}$, of the \NEZO sets
\[
 \CC^{\ifi n} - \V\left(\ff(\aa',\balf'')\right),
\]
where $\aa'=(\aa_1,\dotsc,\aa_{\ifi})$, hence a \NEZO subset of 
$\CC^{\ifi n}$.

Let $\balf' \in \ZOnoether(V, \polun,\ifi)$ and  $\bphi = 
\left(\phi_1(\XX,\balf_1'),\dotsc,\phi_{\ifi}(\XX,\balf_{\ifi}')\right)$.
Let $i \in \{1,\dotsc,\ifi\}$ then, there exists $\balf''\in\CC^{(d+1-\ifi)n}$ 
such that $(\balf',\balf'') \in \ZOpen_i$. Therefore by 
Proposition~\ref{prop:finitepolar}, the restriction of $\bphi_{i-1}$ to
$\Wphii[i][V][\bphi]$ is finite.

In particular, by \cite[Section 5.3]{Sh1994}, the restriction of $\map[i-1]$ to 
$\Wphii$ is a Zariski-closed map that has finite fibers.
Moreover, since $\sing(V)$ is finite, we deduce that $\Kphii \cap 
\map[i-1][-1](\zz)$ is finite for any $\zz \in \CC^{i-1}$.
Finally, as a consequence, and by \cite[Theorem 
1.12 and 1.25]{Sh1994}, $\Wphii$ is equidimensional of dimension 
$i-1$. It is worth noting that the latter can also be seen as a consequence of 
\cite[Theorem 1.25]{Sh1994} and Lemma~\ref{lem:lowerboundglobal}.
  
 \section{Proof of Proposition~\ref{prop:polargen}: atlases for polar
  varieties}\label{sec:polargen}
This section is devoted to prove Proposition~\ref{prop:polargen}, that we recall
below.
\begin{proposition*}[\ref{prop:polargen}]
 Let $V,S \subset \CC^n$ be two algebraic sets with $V$ $d$-equidi\-men\-sio\-nal and
  $S$ finite and $\bchi$ be an atlas of $(V,S)$. For $2\leq \ifi\leq d+1$, let 
$\bpolun = (\polun_1,\dotsc,\polun_{\ifi})$ and $\bpolde = 
(\polde_1,\dotsc,\polde_{\ifi})$, and for $1\leq j \leq \ifi$, let
  $\balf_j =(\alf_{j,1}, \ldots, \alf_{j,n})\in \CC^{n}$ and
  \[
    \phi_j(\XX,\balf_j) = \polun_j(\XX) + \sum_{k=1}^n \alf_{j,k}x_k +
    \polde_j(\balf_j) \in \CC[\XX].
  \]
where $\polun_j \in \CC[\XX]$ and $\polde_j\colon\CC^n\to\CC$ is a polynomial 
map, with coefficients in $\CC$. 

  There exists a non-empty Zariski open subset 
$\ZOpolar(\bchi,V,S,\bpolun,\bpolde,\ifi) \subset \CC^{\ifi n}$ such that for 
every $\balf \in \ZOpolar(\bchi,V,S,\bpolun,\bpolde,\ifi)$, writing $\bphi = 
\left(\phi_1(\XX,\balf),\dotsc,\phi_{\ifi}(\XX,\balf)\right)$, the following 
holds. For $i$ in $\{1,\dots,\ifi\}$, either $\Wphii$ is empty or
 \begin{enumerate}
 \item $\Wphii$ is an equidimensional algebraic set of dimension $i-1$;
 \item if $2 \leq i \leq (d+3)/2$, then $\Watlas(\bchi,V,S,\bphi,i)$ is an atlas
   of $(\Wphii,S)$\\ and $\sing(\Wphii) \subset S$.
 \end{enumerate}
\end{proposition*}

\subsection{Regularity properties}
\emph{In this subsection, we fix the three integers $(d, \ifi,i)$ such that 
$2 \leq \ifi \leq d+1 \leq n+1$ and $1 \leq i \leq \ifi$.}

For $1\leq j \leq i$, let $a_j = (a_{j,1},\dotsc,a_{j,n})$ be new 
indeterminates, and let $\AA = (a_j)_{1\leq j\leq i}$.
For $1\leq j \leq i$, we will also denote by $\AA_{\leq j}$, the subfamily 
$(a_1,\dotsc,a_j)$.
Finally, we consider sequences $\hh = (h_1,\dotsc,h_c) \subset \CC[\XX]$, where 
$c=n-d$, and  $\vphi = (\vphi_1,\dotsc,\vphi_i)$ such that
\[
  \vphi_j(\XX,a_j) = \polun_j(\XX) + 
  \sum_{k=1}^n a_{j,k}x_k  +
  \polde_j(a_j) \in \CC[\XX,\AA_{\leq j}],
\]
for $1\leq j \leq i$.
We start by investigating the regular situation. The first step towards the
proof of Proposition~\ref{prop:polargen} is to establish the following
statement.
\begin{proposition}\label{prop:localgenericityrankgen}
There exists a non-empty Zariski open set
  $\ZOpolarloc \subset \CC^{i n}$, such that for all $\balf \in \ZOpolarloc$, 
  and $\bphi = \big(\vphi_1(\XX,\balf_1),\dotsc,\vphi_i(\XX,\balf_i)\big)\subset
  \CC[\XX]$, the following holds:
 \begin{enumerate}
  \item for all $\yy \in \Voreg(\hh)$, there exists a $c$-minor $m'$ of 
$\jac(\hh)$ such that $m'(\yy) \neq 0$;
  \item the irreducible components of $\Wphii[i][\Vreg(\hh)]$ have dimension 
less than $i-1$; 
\end{enumerate}
Assume now that $i \leq (d+3)/2$, and let $m'$ be any $c$-minor of 
$\jac(\hh)$ and let $m''$ be any $(c+i-1)$-minors of $\jac([\hh,\map[i]])$ 
containing the rows of $\jac(\map[i])$. Then, the following holds:
\begin{enumerate}\setcounter{enumi}{2}
  \item for all $\yy \in \Voreg(\hh)$ there exists $m''$ as above, such that 
$m''(\yy) \neq 0$;
  \item $\Wophii[i][\Voreg(\hh)]$ is defined on $\Ocal(m'm'')$ by the 
vanishing of the polynomials in $(\hh,\Hphi(\hh,i,m''))$;
 \item $\jac(\hh,\Hphi(\hh,i,m''))$  has full rank $n-(i-1)$ on 
  $\Ocal(m'm'') \cap \Wophii[i][\Voreg(\hh)]$.
 \end{enumerate}
\end{proposition}

\subsubsection{Rank estimates}
We start by proving some genericity results on the ranks of some Jacobian
matrix. Two direct consequences (namely Corollaries~\ref{cor:rankhphigen} and 
\ref{cor:rankphigen}) of Proposition~\ref{prop:generalrankdefectgen}
below will establish the third statement of 
Proposition~\ref{prop:localgenericityrankgen}. 

Let $1\leq p \leq n-1$ and $\MRk(\XX,\AA_{\leq 1})$ be a $p\times n$ matrix 
with coefficients in $\CC[\XX,\AA_{\leq 1}]$. For $1 \leq j \leq i$, let
\[
 \Ji[j](\XX,\AA_{\leq j}) = \begin{bmatrix}
                            \\[-0.5em]& \MRk(\XX,\AA_{\leq 1}) & \\\\[-0.5em]
                            \partial_{x_1} \vphi_1(\XX,a_{1,1}) 
                            & \cdots &
                            \partial_{x_n} \vphi_1(\XX,a_{1,n})\\
                            \vdots & & \vdots\\
                            \partial_{x_1} \vphi_j(\XX,a_{j,1}) & \cdots & 
                            \partial_{x_n} \vphi_j(\XX,a_{j,n})
                         \end{bmatrix},                         
\]
where for all $1\leq k \leq i$ and $1 \leq \ell \leq n$, $\partial_{x_\ell} 
\vphi_k =  \frac{\partial\polun_k(\XX)}{\partial x_\ell}+ a_{k,\ell} \in 
\CC[\XX,a_{k,\ell}]$. 
 Proposition~\ref{prop:generalrankdefectgen} below generalizes 
\cite[Proposition B.6]{SS2017}. 
Our proof follows the same pattern as the one of \cite[Proposition B.6]{SS2017}.
\begin{proposition}\label{prop:generalrankdefectgen}
  Assume that there exists a non-empty Zariski open subset $\ZOpolarrank_0
  \subset \CC^{n}$ such that for all $(\yy,\balf)\in \Voreg(\hh) \times
  \ZOpolarrank_0$, the matrix $\MRk(\yy,\balf)$ has full rank $p$. Then, for 
  every 
  \[
    1\;\leq\; j\;\leq\; \min\left\{\,i,\, c-p+(d+3)/2\,\right\},
  \]
  there exists a non-empty Zariski open subset
  $\ZOpolarrank_i \subset \CC^{in}$ such that for all $(\yy,\balf)\in
  \Voreg(\hh) \times \ZOpolarrank_i$,
 \[
   \rank \MRk(\yy,\balf) = p \et \rank \Ji(\yy,\balf) \geq p+j-1.
 \]
\end{proposition}

Before proving the above proposition, we first give two direct consequences of 
it, whose conjunction proves the third item of 
Proposition~\ref{prop:localgenericityrankgen}. Taking $\MRk = \jac(\hh)$, the
next lemma is a direct consequence of the definition of $\Voreg(\hh)$.
\begin{corollary}\label{cor:rankhphigen}
 If $1\leq i \leq (d+3)/2$ then, there exists a non-empty Zariski open subset  
$\ZOpolarrank_i' \subset \CC^{in}$ such that for all $(\yy,\balf)\in 
\Voreg(\hh) 
\times \ZOpolarrank_i'$, the matrix $\jac_{(\yy,\balf)}([\hh,\vphi])$ has rank 
at least $c+i-1$.
\end{corollary}

Besides we deduce the following more subtle consequence.
\begin{corollary}\label{cor:rankphigen}
  If $1\leq i \leq (n+c+1)/2$ then, there exists a
  non-empty Zariski open subset $\ZOpolarrank_i'' \subset \CC^{in}$ such that 
for
  all $(\yy,\balf)\in \Voreg(\hh) \times \ZOpolarrank_i''$, the matrix
  $\jac_{(\yy,\balf)}(\vphi)$ has full rank $i$.
\end{corollary}
\begin{proof}
  Take $\MRk = \jac(\vphi_1)$. 
The matrix $\jac(\vphi_1)$ has not full rank if, and only if, all 
the derivatives of $\vphi_1$ vanish at this point.
  Following the proof strategy of Lemma~\ref{lem:rm}, let
 \[
  \Zo = Z \cap \Voreg(\hh) \subset \CC^{n+n} \quad \text{where} \quad
  Z = \V\left(\hh,\frac{\partial \vphi_1}{\partial x_1},\dotsc,\frac{\partial 
\vphi_1}{\partial x_n}\right).
 \]
 whose following Jacobian matrix, has full rank $c+n$ at any $(\yy,\balf) \in 
\Zo$
 \[
  \jac_{(\XX,a_{1,1},\dotsc a_{1,n})}\left(\hh,\frac{\partial \vphi_1}{\partial 
x_1},\dotsc,\frac{\partial \vphi_1}{\partial x_n}\right) =
  \left[
 \begin{array}{c|ccc}
    \jac(\hh)&&\OO&\\\hline
    * & 1 & \cdots & 0\\[-0.4em]
    \vdots & \vdots & \ddots & \vdots\\
    * & 0 & \cdots & 1
  \end{array}\right].
 \]
 Hence, by the Jacobian criterion \cite[Theorem 16.19]{Ei1995}, $\Zo$ is either 
empty or a $d$-equidimensional \LCS. Since $d <n$ by assumption, then the 
projection of $\Zo$ on the variables $\AA_{\leq 1}$ is a proper subset of 
$\CC^n$ and taking $\ZOpolarrank_0$ as its complement allows us to conclude.
 
 Indeed, for any $1\leq i \leq (n+2)/2$, by 
Proposition~\ref{prop:generalrankdefectgen}, there exists a non-empty Zariski 
open subset $\ZOpolarrank_i$ of $\CC^{in}$ such that for all $(\yy,\balf)\in 
\Voreg(\hh)\times \ZOpolarrank_i$,
 \[
   \rank \jac_{(\yy,\balf)}(\vphi_1,\dotsc,\vphi_i) = 
   \rank \jac_{(\yy,\balf)}(\vphi_1,\vphi_1,\dotsc,\vphi_i) =
   1+i-1 = i.
 \]
\end{proof}
The rest of this paragraph is devoted to the proof of 
Proposition~\ref{prop:generalrankdefectgen}.
Following the construction of the proof of \cite[Proposition B.6]{SS2017}, we 
proceed by induction on $j$.
For all $1\leq j \leq \min\{i,\lfloor c-p+(d+3)/2\}$, we denote by $\sfR_j$
the statement of Proposition~\ref{prop:generalrankdefectgen}.

\paragraph*{Base case: ${j=1}$}
By assumption, there exists a non-empty Zariski open subset $\ZOpolarrank_0
\subset \CC^{n}$ such that for all $(\yy,\balf)\in \Voreg(\hh) \times
\ZOpolarrank_0$, the matrix $\MRk(\yy,\balf_1)$ has full rank $p$.
Therefore, the matrix $\Ji[1]$, containing $\MRk$, has rank at least $p$.
This proves that $\sfR_1$ holds.

\paragraph*{Induction step: ${2\leq j\leq 
\min\{\,i,\,c-p+(d+3)/2\,\}}$}\quad
Assume that $\sfR_{j-1}$ holds, and let us prove that so does $\sfR_j$.
Let $\Eminor$ be the set of ordered pairs $\minor = (\minor_r,\minor_c)$ 
where 
\begin{itemize} 
 \item $\{1,\dotsc,p\} \subset \minor_r \subset \{1,\dotsc,p+j-1\}$
 \item $\minor_c \subset \{1,\dotsc,n\}$
 \item $|\minor_r| = |\minor_c| = p+j-2$
\end{itemize}

Then, for each such $\minor$, let $\Jm$ be the square submatrix of $\Ji[j]$
obtained by selecting the rows and the columns in respectively $\minor_r$ and
$\minor_c$. 
Such a submatrix can also be obtained by removing from $\Ji$, $n-p-j+2$ columns 
and and two rows, which includes the last row.  
Besides, let $\gm \in \CC[\XX,\AA_{\leq j-1}]$ be the determinant of $\Jm$, 
that is the $(p+j-2)$-minor of $\Ji[j]$ associated to $\minor$.
Finally, let $\subi[j]$ be the subset of $\minor \in \Eminor$ such that there
exists $(\yy,\balf)\in \Voreg(\hh) \times \CC^{jn}$ such that
$\gm(\yy,\balf)\neq 0$. 
\begin{lemma}\label{lem:subiempty}
  The set $\subi[j]$, thus defined, is not empty.
\end{lemma}
\begin{proof}
 By induction assumption $\sfR_{j-1}$, there exists a non-empty Zariski open
  subset $\ZOpolarrank_{j-1} \subset \CC^{(j-1)n}$ such that for all
  $(\yy,\balf')\in \Voreg(\hh) \times \ZOpolarrank_{j-1}$, the matrix
  $\Ji[j-1](\yy,\balf')$ has rank at least $p+j-2$ and $\MRk(\yy,\balf')$ has
  full rank $p$. We deduce that there exists a non-zero $(p+j-2)$-minor of
  $\Ji[j-1](\yy,\balf')$ containing the rows of $\MRk(\yy,\balf')$. Then, by
  definition of $\Eminor$,
 \begin{equation}\label{eq:mineurnonnulgen}
  \forall (\yy,\balf')\in \Voreg(\hh) \times \ZOpolarrank_{j-1}, \;\exists 
\minor \in \Eminor, \quad\gm(\yy,\balf') \neq 0, 
 \end{equation}
 where $\gm \in \CC[\XX,\AA_{\leq j-1}]$. This proves, in particular, that 
$\subi[j]$ is not empty, as neither $\Vreg(\hh)$ nor $\ZOpolarrank_{j-1}$ is 
empty.
\end{proof}

We now prove the following lemma, which is the key step in the proof of 
$\sfR_j$.
\begin{lemma}\label{lem:rm}
 For all $\minor \in \subi[j]$, there exists a non-empty Zariski open subset 
$\Em \subset \CC^{jn}$ such that, for all $(\yy,\balf)\in \Voreg(\hh) \times 
\Em$, if $\gm(\yy,\balf)\neq0$, then $\Ji[j](\yy,\balf)$ has rank at least 
$p+j-1$.
\end{lemma}
\begin{proof}
Take $\minor$ in $\subi[j]$. We proceed to show that the subset of the $\balf 
\in \CC^{jn}$ such that, for all $\yy\in \Voreg(\hh)$, 
$\gm(\yy,\balf)\neq0$ and $\Ji[j](\yy,\balf)$ has rank at most $p+j-2$ is 
a \emph{proper algebraic subset} of $\CC^{jn}$. Then, taking the complement 
will give us $\Em$.
 
 Up to reordering, assume that the rows and columns of $\Ji[j]$ that are not 
in $\Jm$ are the ones of respective indices $p+j-1,p+j$ (the last two rows) and 
$p-j+3,\dotsc,n$ (the last $n-p+j-2$ columns).
 In other words, $(p+k,\ell) \notin \minor_r \times \minor_c$ for all 
$k\in\{j-1,j\}$  and $\ell\in\{p-j+3,\dotsc,n\}$.
For such $k,\ell$, we denote by $\gm[k,\ell]$ the minor of $\Ji[j]$ obtained by 
adding to $\Jm$ the row and column indexed by respectively $p+k$ and $\ell$. 
 Let $\AA''$ be the subset of elements of $\AA_{\leq }$ formed by the 
$2(n-p-j+2)$ indeterminates 
\[
 a_{j-1,p-j+3},\dotsc,a_{j-1,n} \et a_{j,p-i+3},\dotsc, a_{j,n},
\]
and let $\AA' = \AA_{\leq j} - \AA''$. 
 Remark then that for any such $k\in\{j-1,j\}$  and 
$\ell\in\{p-j+3,\dotsc,n\}$, by cofactor expansion there exists a polynomial 
$g_{k,\ell} \in \CC[\XX,\AA']$ such that
\begin{equation}\label{eq:guvminorgen}
 \gm[u,v] = \gm \cdot \frac{\partial \vphi_k}{\partial x_\ell}(\XX,a_{k,\ell})
 + \:
 g_{k,\ell}(\XX,\AA')
\end{equation}
 
Let $\bgm[]$ be the sequence of the $2(n-p-j+2)$ minors $\gm[k,\ell]$. We 
proceed to prove that, the set of specialization values $\balf\in \CC^{jn}$ of 
the genericity parameters (the entries of $\AA_{\leq j}$), such that all these
minors in $\bgm[](\XX, \balf)$ are identically zero but not $\gm(\XX, \balf)$,
is a proper algebraic subset of $\CC^{jn}$. Hence, let $t$ a new indeterminate
and consider the \LCS
\[
\Zo = Z \cap \Voreg(\hh) \subset \CC^{n+jn+1} \quad\text{where}\quad 
Z = \V\left(\hh,\bgm[],1-t\gm\right). 
\]
One observes that if $(\yy,\balf,t) \in \Zo$ then $\yy \in \Voreg(\hh)$, 
$\gm(\yy,\balf)\neq 0$ and all the $\gm[k,\ell]$'s vanish.

We claim first that $\Zo$ is not empty. Indeed, since $\minor \in \subi[j]$, 
there exists $(\yy,\balf) \in \Voreg(\hh) \times \CC^{jn}$ such that 
$\gm(\yy,\balf)\neq 0$. Since $\gm \in \CC[\XX,\AA']$, it is independent of the 
entries of $\AA''$. 
Besides, for any $k\in\{j-1,j\}$ and $\ell\in\{p-j+3,\dotsc,n\}$,
\begin{equation}\label{eq:dpsi}
 \frac{\partial \vphi_k}{\partial x_\ell}(\yy,a_{k,\ell}) =
 \frac{\partial \polun_k}{\partial x_\ell}(\XX) + a_{k,\ell} \in \CC[\XX][\AA'']
\end{equation}
is a non-constant polynomial in the entries of $\AA''$.
Then, according to \eqref{eq:guvminorgen}, for every such $k,\ell$, one can 
choose 
$\alf_{k,\ell} \in \CC$ such that $\gm[k,\ell](\yy,\balf',\alf_{k,\ell}) = 0$.
Let $\tilde{\balf}$ be the element of $\CC^{jn}$ obtained by this choice, then
\[
 \big(\;\yy,\;\tilde{\balf},\;1/\gm(\yy,\tilde{\balf})\;\big) \in \Zo.
\]
We deduce that $\Zo$ is non-empty. We now estimate the dimension of $\Zo$ .
According to \eqref{eq:guvminorgen} and \eqref{eq:dpsi} the following Jacobian
matrix has full rank $c+2(n-p-j+2)+1$ at every point of $\Zo$:
\[\jac_{(\XX,\AA',\AA'',t)}(\hh,\bgm[],1-t\gm) = 
 \left[
 \begin{array}{c|ccc|ccc|c}
    \jac(\hh)&&\OO&&&\OO&&0\\ \hline 
    * &*&& * &\gm& 0& \vdots& 0\\
     &&&& 0 & \ddots & 0 & \vdots\\
    * &&*&&& 0 &\gm& 0\\\hline
    * &*&&*&*&*& * &\gm
  \end{array}\right].
\]
Therefore, by the Jacobian criterion \cite[Theorem 16.19]{Ei1995}, $\Zo$ is an 
equidi\-men\-sio\-nal \LCS of dimension $jn-(n-p)+2(j-2)$.
Let $Z' \subset \CC^{jn}$ be the Zariski closure of the projection of $\Zo$ on 
the coordinates associated to the variables $\AA$, then 
\[
 \dim Z' \leq \dim \Zo = jn+d-2(n-p)+2(j-2) < jn, \quad
 \text{since $j \leq c-p+(d+3)/2$.}
\]
Hence $Z'$ is a proper algebraic set of $\CC^{jn}$, so that its complement
$\ZOpolarrank_\minor$ a non-empty Zariski open subset of 
$\CC^{jn}$. Further, for any $(\yy,\balf) \in \Voreg(\hh) \times 
\ZOpolarrank_\minor$ 
such that $\gm$ does not vanish at $(\yy,\balf)$, the point 
$(\yy,\balf,1/\gm(\yy,\balf)\big)$ is not in $\Zo$,
otherwise $\balf$ would be in $Z'$. Hence, there exists $(k,\ell)$ as above 
such that $\gm[k,\ell](\yy,\balf)\neq 0$, so that $\Ji[j](\yy,\balf)$ has a 
non-zero $(c+j-1)$-minor, and then, has rank at least $c+j-1$. 
This proves the lemma.
\end{proof}

We can now conclude on the induction step as follows.
Since, by Lemma~\ref{lem:subiempty}, $\subi[j]$ is not empty, let
\[
  \ZOpolarrank_j = (\ZOpolarrank_{j-1} \times \CC^n) \:\cap
  \bigcap_{\minor\in\subi[j]}\Em,
\]
where the $\Em$ are the \NEZO sets given by Lemma~\ref{lem:rm}. 
Remark first that $\ZOpolarrank_j$ is a non-empty Zariski
open subset of $\CC^{jn}$ since it is a finite intersection of non-empty Zariski
open sets. 
Let $(\yy,\balf',\balf_j) \in \Voreg(\hh) \times
\ZOpolarrank_j$, as seen in \eqref{eq:mineurnonnulgen}, there exists 
$\minor_0 \in \subi[j]$ such that $\gm[\minor_0](\yy,\balf')\neq 0$.
By construction, $\balf=(\balf',\balf_j)$ belongs to $\Em[\minor_0]$ so that, 
by Lemma~\ref{lem:rm}, $\Ji[j](\yy,\balf)$ has rank at least $p+j-1$.
Besides, since $\balf'\in\ZOpolarrank_{j-1}$, $\MRk(\yy,\balf')$ has full rank 
$p$. 

In conclusion, we proved that $\sfR_j$, which the induction step, and, by 
mathematical induction, this proves Proposition~\ref{prop:generalrankdefectgen}.

\subsubsection{Dimension estimates}
In this paragraph, we aim to prove the second point of 
Proposition~\ref{prop:localgenericityrankgen}, using transversality results.
Let
\[
 \begin{array}{lcccccclccccc}
  \Phi\colon \hspace*{-1cm}& \quad\CC^n&\times&\CC^{in} &\times& \CC^c& \times& \CC^i &
\longrightarrow & \;\CC^c \times \CC^n\\
             & (\;\yy&,&\balf&,&\blambda&,&\bvtheta \quad)& \longmapsto &
\Big( 
\;\hh(\yy) , {}^t[\blambda,\bvtheta] \cdot 
\jac_{(\yy,\balf)}(\hh,\vphi)\Big)
 \end{array}
\]
and for any $\balf \in \CC^{in}$, let $\Phi_{\balf} = (\yy,\blambda,\bvtheta) 
\mapsto \Phi(\yy,\blambda,\bvtheta,\balf)$.
The interest of such a map is illustrated by the following lemma.
Let $\scrA \subset \CC^{n+in+c+i}$ be the Zariski open subset of the elements 
$(\yy,\blambda,\bvtheta)$ where $\blambda\neq \OO$ and $\jac_{\yy}(\hh)$ has 
full rank. 
\begin{lemma}\label{lem:woprojgen}
  Let $\balf \in \CC^{in}$ and
 \[
  \Wo_{\balf} = \Big\{ \yy \in \CC^n \mid \yy \in \Voreg(\hh) \et 
\rank \jac_{(\yy,\balf)}(\hh,\vphi) \leq c+i-1\Big\}.
  \] 
  Then $\Wo_{\balf} = \pi_{\XX}\big(\scrA \cap \Phi^{-1}_{\balf}(\OO)\big).$
\end{lemma}
\begin{proof}
 Let $\balf \in \CC^{in}$ and $\yy \in \Voreg(\hh)$. Then $\yy \in \Wo_{\balf}$ 
if and only if $\jac_{(\yy,\balf)}(\hh,\vphi)$ has not full rank, which is 
equivalent to having a non-zero vector in its co-kernel by duality.
 Besides, since $\yy \in \Voreg(\hh)$, the matrix $\jac_{\yy}(\hh)$ has full 
rank. 
 Hence $\yy$ belongs to $\Wo_{\balf}$ if and only if there exists a non-zero 
vector $(\blambda,\bvtheta) \in \CC^{c+i}$ such that 
$\Phi(\yy,\balf,\blambda,\bvtheta) = 0$ and $\jac_{\yy}(\hh)$ has full rank.
 Finally $\bvtheta$ cannot be zero otherwise $\jac_{\yy}(\hh)$ would have a 
non-trivial left-kernel (containing $\blambda$), and then would not be full 
rank.
\end{proof}

\begin{lemma}\label{lem:rankjacPhigen}
Let $\scrA \subset \CC^{n+c+i}$ be the Zariski open subset of the elements 
$(\yy,\blambda,\bvtheta)$ where $\bvtheta\neq \OO$ and $\jac_{\yy}(\hh)$ has 
full rank.
 There exists a non-empty Zariski open subset $\ZOpolardim_i \subset \CC^{in}$ 
such 
that for all $\balf \in \ZOpolardim_i$, 
$\jac_{(\yy,\blambda,\bvtheta)}\Phi_{\balf}$ has 
full rank $c+n$ at any $(\yy,\blambda, \bvtheta) \in \scrA \cap 
\Phi^{-1}_{\balf}(\OO)$.
\end{lemma}
\begin{proof}
We have
\[
  \jac_{(\XX,\lambda,\vtheta,a_1,\dotsc,a_i)}\Phi=
  \left[\begin{array}{c|c|ccccccc}
    &&&&\\[-0.9em]
    \jac(\hh) & \hspace*{1.5em}\OO\hspace*{1.5em} &\OO & \cdots & \OO\\ \hline
     &&&&\\[-0.9em]
     \bstar & \bstar & \vtheta_1 I_n & \cdots & \vtheta_{i} I_n
  \end{array}\right]
\]
where $I_n$ is the identity matrix of size $n$. Let $\balf \in \CC^{in}$, and
$(\yy, \blambda, \bvtheta) \in \scrA$ such that the above Jacobian matrix has
full rank $c+n$ at $(\yy,\balf,\blambda, \bvtheta)$. Hence $\OO$ is a regular
value of $\Phi$ on $\scrA \times \CC^{in}$. Therefore, by the Thom's weak
transversality theorem \cite[Proposition B.3]{SS2017}, there exists a non-empty
Zariski open subset $\ZOpolardim_i \subset \CC^{in}$ such that for all $\balf
\in \ZOpolardim_i$, $\OO$ is a regular value of $\Phi_{\balf}$ on $\scrA$. In
other words, for all $\balf \in \ZOpolardim_i$, the matrix $\jac \Phi_{\balf}$
has full rank $c+n$ over $\scrA\cap\Phi_{\balf}^{-1}(\OO)$.
\end{proof}

\begin{lemma}\label{lem:dimwogen}
  Let $\ZOpolardim_i \subset \CC^{in}$ be the non-empty Zariski subset defined
  in Lemma~\ref{lem:rankjacPhigen}. Then, for all $\balf \in \ZOpolardim_i$,
  $\Wo_{\balf}$ has dimension at most $i-1$.
\end{lemma}
\begin{proof}
  Let $\balf \in \ZOpolardim_i$ and suppose that $\Wo_{\balf}$ is not empty.
  Then, according to Lemma~\ref{lem:woprojgen}, $\scrA \cap
  \Phi^{-1}_{\balf}(\OO)$ is non-empty as well. By Lemma~\ref{lem:rankjacPhigen}
  and \cite[Lemma A.1]{SS2017}, $\scrA \cap \Phi^{-1}_{\balf}(\OO)$ is a
  non-singular equidimensional \LCS and
 \[
  \dim\big(\scrA \cap \Phi^{-1}_{\balf}(\OO)\big) = n+c+i-(c+n)=i.
 \]
 Let $C$ be the Zariski closure of $\scrA \cap \Phi^{-1}_{\balf}(\OO)$ and
 $(C_j)_{1\leq j\leq \ell}$ be its irreducible components. For all $1\leq j \leq
 \ell$, let $T_j$ be the Zariski closure of $\pi_{\XX}(C_j)$. Since $\Wo_{\balf}
 \subset \bigcup_{1\leq j \leq \ell} T_j$, it is enough to prove that $\dim T_j
 \leq i-1$ for all $1\leq j \leq \ell$.

 Fix $1\leq j \leq \ell$. The restriction $\pi_{\XX} \colon C_j \to T_j$ is a
 dominant regular map between two irreducible algebraic sets. Then one can apply
 the theorem on the dimension of fibers from \cite[Theorem 1.25]{Sh1994} and
 claim that there exists a non-empty Zariski open subset $\Omega_1$ of $T_j$
 such that
 \begin{equation}\label{eqn:dimfbrCj}
  \forall \zz \in \Omega_1, \; 
  \dim\big(\pi_{\XX}^{-1}(\zz) \cap C_j \big) = \mypos{}{Cj}\dim C_j - \dim 
T_j= 
i - \dim T_j.
\end{equation}
Then it is enough to prove that $\dim(\pi_{\XX}^{-1}(\zz) \cap C_j) \geq 1$. 
Let 
$J' = \{1\leq k \leq \ell \mid T_k = T_j\}$. Then it holds that
\[
  \Omega_2 = T_j - \bigcup_{k \notin J'} T_k
\]
is a non-empty Zariski open subset of $T_j$. Besides, for all $\zz \in 
\Omega_2$, $\pi_{\XX}^{-1}(\zz) \cap C_j = \pi_{\XX}^{-1}(\zz) \cap C$ which is 
the Zariski closure of $\pi_{\XX}^{-1}(\zz) \cap \scrA \cap 
\Phi^{-1}_{\balf}(\OO)$ if and only if $\zz \in \Wo_{\balf}$ (otherwise it is 
empty).

However, by definition, $C_j' = \scrA \cap \Phi^{-1}_{\balf}(\OO) \cap C_j$ is 
a 
non-empty Zariski open subset of $C_j$, and then $\pi_{\XX}(C_j')$ is a 
non-empty Zariski subset of $T_j$. Since it contains $\pi_{\XX}(C_j')$,  the 
set $\Omega_3 = \Wo_{\balf} \cap T_j$ is a non-empty Zariski open subset of 
$T_j$ as well. 

Now, let $\Omega = \Omega_1 \cap \Omega_2 \cap \Omega_3$,  it is a non-empty
(Zariski open) subset of $T_j$, and let $\zz \in \Omega$.
Since $\zz$ is in $\Omega_3$, it is in $\Wo_{\balf}$ by definition.
Besides, $\zz \in \Omega_2$, so that
\[
  \dim(\pi_{\XX}^{-1}(\zz) \cap C_j)
  =
  \dim\Big(\pi_{\XX}^{-1}(\zz) \cap
  \scrA \cap \Phi_{\balf}^{-1}(\OO)  \Big).
\]
Since $\zz \in \Omega_1$, together with \eqref{eqn:dimfbrCj}, one gets
that
\begin{equation}\label{eq:inegdim0gen}
 \forall \zz \in \Omega, \quad \zz \in \Wo_{\balf} \et
 \dim T_j =
  i - \dim\Big(\pi_{\XX}^{-1}(\zz) \cap
  \scrA \cap \Phi_{\balf}^{-1}(\OO)  \Big),
\end{equation} 
Let $\zz \in \Omega$, remark that 
\[
 \pi_{\XX}^{-1}(\zz) \cap \scrA \cap \Phi_{\balf}^{-1}(\OO) 
 \;=\; \{\zz\} \times \Big( E_{\zz} \cap \Ocal(\vtheta_{e+1},\dotsc,\vtheta_i) 
\Big)
\]
where $E_{\zz}$ is a linear subspace of $\CC^{c+i}$. Indeed, $E_{\zz}$ is 
defined 
by homogeneous linear equations in the entries of $(\blambda,\bvtheta)$.
Since $\zz \in \Wo_{\balf} \subset \pi_{\XX}(\scrA \cap 
\Phi_{\balf}^{-1}(\OO))$, 
there exists a non-zero $(\blambda,\bvtheta) \in \CC^{c+i}$ such that 
$(\zz,\blambda,\bvtheta) \in \scrA \cap \Phi_{\balf}^{-1}(\OO)$. Then $E_{\zz}$ 
contains a non-zero vector, so that $\dim E_{\zz} \geq 1$.
 Finally, injecting this inequality in \eqref{eq:inegdim0gen} leads to $\dim 
T_j 
\leq i-1$ as required.
\end{proof}

\subsubsection{Proof of Proposition~\ref{prop:localgenericityrankgen}.}~\\
We can now tackle the proof of the main proposition of this subsection.
Recall that we have fixed three integers $(d,\ifi,i)$ such that $2\leq \ifi 
\leq d+1 \leq n+1$ and $1 \leq i \leq \ifi$. Moreover, we consider polynomials 
$\hh = (h_1,\dotsc,h_c)$ in $\CC[\XX]$, where $c = n-d$. 
Finally, let $\vphi = (\vphi_1,\dotsc,\vphi_i)$, such that
\[
  \vphi_j(\XX,a_j) = \polun_j(\XX) + 
  \sum_{k=1}^n a_{j,k}x_k  +
  \polde_j(a_j) \in \CC[\XX,\AA_{\leq j}],
\]
for all $1\leq j \leq i$. Let $\ZOpolarloc$ be the non-empty Zariski open 
subset of $\CC^{in}$ defined by
\[
  \ZOpolarloc = 
\left\{ \begin{array}{cl} \ZOpolardim_i \cap \ZOpolarrank_i' \cap
      \ZOpolarrank_i''
            & \text{if } i \leq (d+3)/2;\\
            \ZOpolardim_i & \text{else,}
          \end{array}\right.
\]
where $\ZOpolardim_i$, $\ZOpolarrank_i'$ and $\ZOpolarrank_i''$ are the \NEZO 
sets given respectively by Lemma~\ref{lem:rankjacPhigen}, 
Corollaries~\ref{cor:rankhphigen} and \ref{cor:rankphigen}. Note that the 
assumptions of Corollary~\ref{cor:rankphigen} since $d \leq n-1$.

\smallskip
Now let $\balf \in \ZOpolarloc$ and $\bphi = 
(\vphi_1(\XX,\balf),\dotsc,\vphi_i(\XX,\balf))$. 
The first item of the proposition is a direct consequence definition of
$\Voreg(\hh)$. Besides, according to \cite[Lemma A.2]{SS2017}, the set 
$\Wo_{\balf}$ defined in Lemma~\ref{lem:woprojgen} is nothing but 
$\Wophii[i][\Vreg(\hh)]$, whose Zariski
closure is $\Wophii[i][\Vreg(\hh)]$, by definition. Hence, since $\balf \in
\ZOpolardim_i$, the second item is exactly the statement of
Lemma~\ref{lem:dimwogen}.

\smallskip
Suppose now that $i \leq (d+3)/2$, so that $\balf \in \ZOpolarrank_i' \cap
\ZOpolarrank_i''$. Hence, by Corollaries~\ref{cor:rankhphigen}
and~\ref{cor:rankphigen}, for all $\yy \in \Voreg(\hh)$,
\[
  \rank \jac_{\yy}(\map[i]) = i \et \rank \jac_{\yy}(\hh,\map[i]) \geq c+i-1.
\]
Hence, there exists a $(c+i-1)$-minor $m''$ of $\jac_{\yy}(\hh,\map[i])$,
containing the rows of $\jac(\map[i])$, that does not vanish at $\yy$. This
proves the third item.

\smallskip
In the remaining we proceed to prove the last two items. Let $m'$ be a 
$c$-minor of $\jac(\hh)$ and $m''$ be a $(c+i-1)$-minor of 
$\jac([\hh,\map[i]])$ containing the rows of $\jac(\map[i])$. Assume, without 
loss of generality, that $m''$ is not the zero polynomial. 
The next lemma establishes the second to last item of 
Proposition~\ref{prop:localgenericityrankgen}.
\begin{lemma}\label{lem:intercomplW}
  Let $m'$ and $m''$ as above. The set $\Wophii[i][\Voreg(\hh)]$ is defined on 
$\Ocal(m'm'')$ by the vanishing set of the polynomials 
$(\hh,\Hphi(\hh,i,m''))$. Equivalently,
 \[
    \Ocal(m'm'') \; \cap \;\Wophii[i][\Voreg(\hh)] = \Ocal(m'm'') \; \cap \; 
\V\big(\hh,\Hphi(\hh,i,m'')\big).
 \]

\end{lemma}
\begin{proof}
  Inside the Zariski open set $\Ocal(m')$, the matrix $\jac(\hh)$ has full rank,
  which implies by \cite[Lemma A.2]{SS2017} that
\[
 \Ocal(m') \cap \Wophii[i][\Voreg(\hh)] = \Ocal(m') \cap
 \Big\{ \yy \in \V(\hh) \mid \rank(\jac([\hh,\map[i]]) < c + i \Big\}. 
\]
Besides, by the exchange lemma of \cite[Lemma 1]{BGHM2001}, if $m$ is a
$(c+i)$-minor of $\jac([\hh,\map[i]])$, then one can write
\[
 m'' m = \sum_{j=1}^{} \eps_j m_j m_j'' \quad
 \text{where } \eps_j = \pm 1 \et N \in \{1,\dotsc, d-i+1\}
\]
and where $m_j''$ (resp. $m_j$) is obtained by successively adding to $m''$
(resp. removing to $m$) the missing row and a missing column of
$\jac([\hh,\map[i]])$ that are in $m$. Remark that, for such a $m$, all the
$m_j''$'s are in $\Hphi(\hh,i,m'')$, by definition.

Hence, for all $\yy \in \V(\hh)$, if $m''(\yy) \neq 0$, then all the
$(c+i)$-minors of $\jac([\hh,\map[i]])$ vanish at $\yy$ if and only if all the
polynomials of $\Hphi(\hh,i,m'')$ vanish at $\yy$. In other words:
\[
 \Ocal(m'm'') \cap \Wophii[i][\Voreg(\hh)] = 
 \Ocal(m'm'') \cap \V\big(\hh,\Hphi(\hh,i,m'')\big).
\]
\end{proof}

In order to prove the the last item of 
Proposition~\ref{prop:localgenericityrankgen}, we need introduce Lagrange 
systems for general polynomial applications. This generalizes, in some sense, 
the construction of \cite[Subsection 5.1]{SS2017}, also presented 
in Subsection~\ref{ssec:lagrange}.

Let $\Lindet[c]$ and $\Tindet[i]$ be new indeterminates, since $m''\neq 0$, 
consider the ring of rational fractions $\CC[\XX,\Lindet[c],\Tindet[i]]_{m''}$  
of the form $f/(m'')^r$, for $f \in \CC[\XX,\Lindet[c],\Tindet[i]]$ and $r \in
\mathbb{N}$. This the localization ring at the multiplicative set $\{(m'')^r
\mid r \in \mathbb{N}\}$.

Let $\Ical_W$ the ideal of $\CC[\XX,\Lindet[c],\Tindet[i]]_{m''}$ generated by 
the entries of
\[
 \hh,\quad
 [\Lindet[c],\Tindet[i]] \cdot \begin{bmatrix}
                                      \jac(\hh) \\
                                      \jac(\map[i])
                                     \end{bmatrix}.
\]
The following lemma is an immediate generalization of \cite[Proposition 
5.2.]{SS2017}.
\begin{lemma}\label{lem:lagrange}
  Let $1 \leq \iota \leq c$ such that the index of the row of
  $\jac([\hh,\map[i]])$ not in $m''$ has index $\iota$. Then there exist
  $(\lambda_j)_{1\leq j\neq \iota \leq c}$ and $(\tau_j)_{1\leq j \leq i}$ in
  $\CC[\XX]_{m''}$ such that $\Ical_W$ is generated by the entries of
\begin{equation}\label{eqn:polynomelag}
 \hh, \quad 
 L_{\iota}\Hphi(\hh,i,m''), \quad 
 (L_j - \lambda_j L_{\iota})_{1\leq j\neq \iota \leq c}, \quad 
 (T_j - \tau_j L_{\iota})_{1\leq j\leq i}.
\end{equation}
\end{lemma}

\begin{proof}
  For the sake of simplicity, suppose that $m''$ is the lower-left minor of
  $\jac([\hh,\map[i]])$, so that $\iota = 1$. Then $\Hphi(\hh,i,m'')$ is the
  sequence of minors obtained by adding the first row and columns in the ones of
  index $c+i,\dotsc,n$. We denote these minors by $M_1,\dotsc,M_{n-c-i+1}$.
  Then, we write
 \[
   \jac(\hh,\map[i]) = \begin{pmatrix}
                         \bm{u}_{1,c+i-1} & \bm{w}_{1,n-c-i+1}\\
                         \bm{m}_{c+i-1,c+i-1} & \bm{v}_{c+i-1,n-c-i+1}
                   \end{pmatrix}
 \]
 such that $m'' = \det(\bm{m})$ and the indices are the dimensions of the
 submatrices. As $m''$ is not zero, it is a unit of
 $\CC[\XX,\Lindet[c],\Tindet[i]]_{m''}$, so that $\mm$ has an inverse with
 coefficients in the same ring, given by $m''^{-1}$ and the cofactor matrix of
 $\mm$. Hence $\Ical_W$ is generated by the entries of $\hh$ and
 \begin{align*}
  &[\Lindet[c],\Tindet[i]] \cdot \jac([\hh,\map[i]]) \cdot 
  \begin{bmatrix}
   \mm^{-1} & \OO\\
   \OO & 1
  \end{bmatrix} \cdot
  \begin{bmatrix}
   I_{c+i-1} & -\vv\\
   \OO & 1
  \end{bmatrix} \\[0.5em]
  = \quad& 
  [\Lindet[c],\Tindet[i]] \cdot
  \begin{bmatrix}
   \uu\mm^{-1} & \ww-\uu\mm^{-1}\vv\\
   I_{c+i-1} & \OO
  \end{bmatrix},
\end{align*}
where $I_{c+i-1}$ is the identity matrix of size $c+i-1$. The first $c-1$ 
entries are the $L_j - [\uu\mm^{-1}]_j L_1$ for $1< j \leq c$ and the $i$ 
followings are the $T_j - [\uu\mm^{-1}]_j L_1$ for $1\leq j \leq i$. Hence 
taking $(\blambda,\btau)= \uu\mm^{-1}$ gives the last terms in 
\eqref{eqn:polynomelag}.

Finally, since $\mm$ is invertible, we can compute the minors
$M_1,\dotsc,M_{n-c-i+1}$ of $\jac(\hh,\map[i])$, using the block structure we
described above (see e.g. \cite[Proposition 2.8.3]{Be2009} and \cite[Theorem
1]{Si2000}) to obtain that for all $1\leq j \leq n-c-i+1$,
\[
 M_j = (-1)^{c+i-1} m''[\ww-\uu\mm^{-1}\vv]_j.
\]
Hence, the last $n-c-i+1$ entries are $L_1 
M_1/m'',\dotsc,L_1 M_{n-c-i+1}/m''$ up to sign, we are done.
\end{proof}

The next lemma ends the proof of the last item 
Proposition~\ref{prop:localgenericityrankgen}, and then conclude the proof of 
the whole proposition.
\begin{lemma}
 The Jacobian matrix of the polynomials in
$(\hh,\Hphi(\hh,i,m''))$ has full rank $n-(i-1)$ at every point of the 
set $\Ocal(m'm'') \; \cap \; \Wophii[i][\Voreg(\hh)]$.
\end{lemma}
\begin{proof}
Recall that $\bphi = (\vphi_1(\XX,\balf),\dotsc,\vphi_i(\XX,\alpha))$, 
where $\balf \in \ZOpolarloc$. Then, remark that
\[
 \left(\hh(\XX),\;
 [\Lindet[c],\Tindet[i]] \cdot
 \begin{bmatrix}
  \jac_{\XX}(\hh) \\
  \jac_{\XX}(\map[i])
  \end{bmatrix}\right) 
 = \Phi_{\balf}\Big(\XX,\Lindet[c],\Tindet[i] \Big),
\]
where $\Phi_{\balf}$ is the polynomial map considered in 
Lemma~\ref{lem:rankjacPhigen}.
Let $(\lambda_j)_{1\leq j\neq \iota \leq c}$ and $(\tau_j)_{1\leq j \leq i}$ in
$\CC[\XX]_{m''}$ given by Lemma~\ref{lem:lagrange}.

Now fix $\yy \in \Ocal(m'm'') \cap \Wphii[i][\Voreg(\hh)]$, and let 
$\blambda=(\blambda_j)_{1\leq j\leq c}$ and $\bvtheta=(\bvtheta_j)_{1\leq j 
\leq i}$ where
\begin{align*}
 \blambda_{\iota} = 1 \et
 \blambda_j = \lambda_j(\yy) &\text{ for all } 1 \leq j \neq \iota \leq c,\\
 \bvtheta_j = \tau_j(\yy) &\text{ for all } 1\leq j \leq i.
\end{align*}
These are well defined since $m''(\yy)\neq 0$. 
Since $\hh$ and $\Hphi(\hh,i,m'')$ vanish at $\yy$, by 
Lemma~\ref{lem:intercomplW}, all the polynomials in \eqref{eqn:polynomelag} 
vanish at $(\yy,\blambda,\bvtheta)$. 
Moreover, according to Lemma~\ref{lem:lagrange} and the above remark, the 
polynomials in \eqref{eqn:polynomelag} and the entries of 
$\Phi_{\balf}(\XX,\Lindet[c],\Tindet[i])$ generates the same ideal $\Ical_W$ in 
$\CC[\XX]_{m''}$.
Hence, since $m''(\yy)\neq 0$, the entries of $\Phi_{\balf}$ vanish at 
$(\yy,\blambda,\bvtheta)$ as well, that is 
$\Phi_{\balf}(\yy,\blambda,\bvtheta)=\OO$

Besides, since $\yy\in\Ocal(m')$,  $\jac_{\yy}(\hh)$ has full rank.
Then, $\bvtheta$ cannot be zero, since $\blambda \neq \OO$
$\jac(\hh)$ has a trivial left-kernel.
Hence, according to the notation of Lemma~\ref{lem:rankjacPhigen}, 
$(\yy,\blambda,\bvtheta) \in \scrA \cap \Phi_{\balf}^{-1}(\OO)$.
Therefore, by Lemma~\ref{lem:rankjacPhigen}, $\jac \Phi_{\balf}$ has full rank 
$n+c$ at $(\yy,\blambda,\bvtheta)$, as $\balf \in \ZOpolardim_i\subset 
\ZOpolarloc$. 

Finally, remark that the sequence of polynomials in \eqref{eqn:polynomelag} has 
length $n+c$. Hence, since the latters generate the same ideal than the entries 
of the entries of $\Phi_{\balf}(\XX,\Lindet[c],\Tindet[i])$, their Jacobian 
matrix has full rank $n+c$ at this point as well. 
Computing this Jacobian matrix the latter rank statement amounts to the 
Jacobian matrix of 
 \[
  \Big(\hh,\;\Hphi(\hh,i,m'')\Big)
 \]
 having full rank $n-(i-1)$ at $\yy$.
\end{proof}

\subsection{Proof of Proposition~\ref{prop:polargen}}

Let $V,S \subset \CC^n$ be be two algebraic sets with $V$ $d$-equidimensional 
and $S$ finite, and let $\bchi = (\chi_j)_{1\leq j \leq s}$ be an atlas of 
$(V,S)$ with $\chi_j = (m_j,\hh_j)$ for $1\leq j \leq s$. According to 
\cite[Lemma A.12]{SS2017}, all the $\hh_j$'s have same cardinality 
$c=n-d$.

Besides, let $2\leq\ifi \leq d+1$ and the sequences 
$\bpolun=(\polun_1,\dotsc,\polde_{\ifi})$ and 
$\bpolde=(\polde_1,\dotsc,\polde_{\ifi})$ in $\CC[\XX]$. For $1\leq j \leq 
\ifi$, let $\balf_j =(\alf_{j,1}, \ldots, \alf_{j,n})\in \CC^{n}$ and
\[
    \phi_j(\XX,\balf_j) = \polun_j(\XX) + \sum_{k=1}^n \alf_{j,k}x_k +
    \polde_j(\balf_j) \in \CC[\XX].
\]
Then, for $1 \leq i \leq \ifi$, we can apply 
Proposition~\ref{prop:localgenericityrankgen} to the sequences $\hh_j$, 
$\bpolun$ and $\bpolde$, there exist a \NEZO subset $\ZOpen(\hh_j,i)$ of 
$\CC^{in}$ such that for all $\balf \in \ZOpen(\hh_j,i)$, the sequence $\bphi = 
\left(\phi_1(\XX,\balf),\dotsc,\phi_{i}(\XX,\balf)\right)$ satisfies the 
statements of Proposition~\ref{prop:localgenericityrankgen}.
Then we define the  following \NEZO subset of $\CC^{\ifi n}$,
\[
  \ZOpolar(\bchi,V,S,\bpolun,\bpolde,\ifi)
  = \bigcap_{1\leq i \leq \ifi} \: \bigcap_{1 \leq j \leq s} \ZOpen(\hh_j,i) 
  \times \CC^{(\ifi-i)n}.
\]

Fix now $\balf \in \ZOpolar(\bchi,V,S,\bpolun,\bpolde,\ifi)$ and $\bphi = 
\left(\phi_1(\XX,\balf),\dotsc,\phi_{\ifi}(\XX,\balf)\right)$. 
From now on, fix also $1\leq i \leq \ifi$ and suppose that $\Wphii$ is not 
empty. 
In the following, and for conciseness, we might identify $\ZOpen(\hh_j,i)$ to 
$\ZOpen(\hh_j,i) \times \CC^{(\ifi-i)n}$. in a straightforward way.
We start with the first item statement of Proposition~\ref{prop:polargen}. 
Again, it is proved
through the properties of atlases, but when $i$ is restricted to some values.
\begin{lemma}\label{lem:Wiequidimgen}
 The algebraic set $\Wphii$ is equidimensional of dimension $i-1$.
\end{lemma}
\begin{proof}
  By Lemma~\ref{lem:polaronchart}, for all $1\leq j \leq s$, as $\chi_j$ is a 
chart of $(V,S)$ then,
 \[
  \Ocal(m_j) \cap \Wphii - S \;=\;
  \Ocal(m_j) \cap \Wophii[i][\Vreg(\hh_j)] - S.
 \]
 Let $\yy \in \Wphii - S$. Since $\yy \in V$, by property $\sfA_3$ of the atlas
 $\bchi$, there exists $j \in \{1,\dotsc,s\}$ such that $\yy \in \Ocal(m_j)$.
 Hence, by the above equality, in $\Ocal(m_j) - S$, the irreducible component of
 $\Wphii$ containing $\yy$ coincides with the one of $\Wphii[i][\Vreg(\hh_j)]$
 containing $\yy$. Since these irreducible components are equal over a non-empty
 Zariski open set, they have the same dimension (see e.g. \cite[Proposition
 10.(1)]{Fu2008}). By the second item of
 Proposition~\ref{prop:localgenericityrankgen}, since $\balf \in
 \ZOpen(\hh_j,i)$, this dimension is less than $i-1$.

We just showed that the Zariski closure of $\Wphii - S$ has dimension less than 
$i-1$. If $i=1$, since $S$ is finite this means that $\Wphii$ is finite as well 
and we are done.
If $i \geq 2$, then by Lemma~\ref{lem:lowerboundglobal}, the irreducible 
components of $\Wphii$ have dimension at least $i-1 \geq 1$ so that the Zariski 
closure of $\Wphii - S$ is $\Wphii$. Hence the irreducible components of 
$\Wphii$ have dimension exactly $i-1$
\end{proof}

We now prove a strict generalization of \cite[Lemma B.12.]{SS2017} 
which gives the key arguments for the proof of the second item statement of
Proposition~\ref{prop:polargen}.
\begin{lemma}\label{lem:chartrecov}
  Let $\chi=(m,\hh)$ be a chart of $(V,S)$. Then for any $c$-minor $m'$ of
  $\jac(\hh)$ and any $(c+i-1)$-minor $m''$ of $\jac([\hh,\map[i]])$, containing
  the rows of $\jac(\map[i])$, the following holds.
 \begin{enumerate}
  \item The sets $\Ocal(mm'm'') \cap \Wphii - S$ and $\Ocal(mm'm'') \cap 
\V(\hh,\Hphi(\hh,i,m''))-S$ coincides;
 \item if they are not empty, then $\Wchart(\chi,m',m'')$ is a chart of 
$(\Wphii,S)$.
 \end{enumerate}
Moreover, if $i \leq (d+3)/2$ then the following holds.
 \begin{enumerate}
 \setcounter{enumi}{2}
  \item The sets $\Ocal(mm'm'')-S$, for all $m',m''$ as above, cover $\Ocal(m) 
\cap V - S$;
  \item the sets $\Ocal(mm'm'')-S$, for all $m',m''$ as above, 
$\Ocal(m)\cap\Wphii-S$.
 \end{enumerate}
\end{lemma}
\begin{proof}
By Lemma~\ref{lem:polaronchart}, since $\chi$ is a chart of $(V,S)$,
 \[
  \Ocal(m) \cap \Wphii - S =
  \Ocal(m) \cap \Wophii[i][\Vreg(\hh)] - S.
 \]
Besides, by the second to last item of 
Proposition~\ref{prop:localgenericityrankgen}, 
$\Wophii[i][\Vreg(\hh)]$ is defined in $\Ocal(m'm'')$ by the vanishing of the 
polynomials $(\hh,\Hphi(\hh,i,m''))$, so that
 \begin{equation}\label{eq:polarintersectioncomplete}
  \Ocal(mm'm'') \cap \Wphii - S = 
  \Ocal(mm'm'') \cap \V(\hh,\Hphi(\hh_j,i,m'')) - S.
 \end{equation}
The first item is proved.
Suppose now that the former sets are not-empty, we proceed to prove that 
$\Wchart(\chi,m',m'')$ is a chart of $(\Wphii,S)$. 
 Property $\sfC_1$ holds by assumption, while property $\sfC_2$ of 
$\Wchart(\chi,m',m'')$ is exactly equation \eqref{eq:polarintersectioncomplete}.
Besides, since $(\hh,\Hphi(\hh_j,i,m''))$ has length $n-i-1 \leq n$, then 
$\sfC_3$ holds as well. 
Finally, by the last item of Proposition~\ref{prop:localgenericityrankgen}, 
$\jac(\hh,\Hphi(\hh,i,m''))$ has  full rank on 
\[
\Ocal(m'm'') \cap \Wophii[\Vreg(\hh)].
\]
Then, by \eqref{eq:polarintersectioncomplete},
$
  \jac(\hh,\Hphi(\hh,i,m''))
$
 has full rank on $\Ocal(mm'm'') \cap \Wphii - S$. This proves that
 $\Wchart(\chi,m',m'')$ satisfies the last property $\sfC_4$ of charts and the
 second statement of the lemma is proved.

 Suppose now that $i \leq (d+3)/2$ and let $\yy \in \Ocal(m) \cap V - S$. Then,
 by property $\sfC_4$ of $\chi$, $\jac(\hh)$ has full rank in $\yy$, so that
 $\yy \in \Voreg(\hh)$. Therefore, by the first and third item of
 Proposition~\ref{prop:localgenericityrankgen}, there exists a $c$-minor $m'$ of
 $\jac(\hh)$ and a $(c+i-1)$-minor $m''$ of $\jac([\hh,\map[i]])$, containing
 the rows of $\jac(\map[i])$, such that $(m'm'')(\yy) \neq 0$. Hence $\yy \in
 \Ocal(mm'm'') - S$ and the third item of the lemma is proved.
 
 Finally, if $\yy \in \Ocal(m) \cap \Wphii - S$, then one still has $\yy \in
 \Ocal(mm'm'') - S$, as $\Wphii \subset V$. This proves the last item.
\end{proof}

We can now prove the second statement of Proposition~\ref{prop:polargen}. With
the above lemmas, it is mainly a matter of verification. Suppose that $2 \leq i
\leq (d+3)/2$. We prove that $\Watlas(\bchi,V,S,\bphi,i)$ is an atlas of
$(\Wphii,S)$. In the following, for $1\leq j \leq s$, we refer to $m_j'$ and
$m_j''$ as respectively a $c$-minor of $\jac(\hh_j)$ and a $(c+i-1)$-minor of
$\jac([\hh_j,\map[i]])$, containing the rows of $\jac(\map[i])$.

\begin{enumerate}[label=$\sfA_\arabic*:$]
\item Since, by Lemma~\ref{lem:Wiequidimgen}, $\Wphii$ has dimension at least 
$1$, it  is not contained in $S$. In particular, there exists $1\leq j\leq s$ 
such that $\Ocal(m_j) \cap \Wphii - S$ is not empty. Hence, by the third item of
  Lemma~\ref{lem:chartrecov}, there exist minors $m_j'$ and $m_j''$ such that
  $\Ocal(m_jm_j'm_j'')\cap\Wphii-S$ is not empty.
 
 \item For $m_j$,$m_j'$ and $m_j''$ as in the previous item, since 
$\Ocal(m_jm_j'm_j'')\cap\Wphii-S$ is not empty, then the second item of 
Lemma~\ref{lem:chartrecov} shows that the sequence $\Wchart(\chi_j,m_j',m_j'')$ is a chart 
of 
$(\Wphii,S)$.

 \item Let $\yy \in \Wphii-S$, by property $\sfA_3$ of $\bchi$ there exists 
$1\leq j \leq s$ such that $\yy \in \Ocal(m_j)$.
 Then, by the third item of Lemma~\ref{lem:chartrecov}, there exist $m_j'$ and 
$m_j''$ as in the previous points such that $\yy \in \Ocal(m_jm_j'm_j'')$. In 
particular $\Ocal(m_jm_j'm_j'')\cap\Wphii-S$ is not empty.
 \end{enumerate}

Hence $\Watlas(\bchi,V,S,\bphi,i)$ is an atlas of $(\Wphii,S)$, and since we 
proved that $\Wphii$ is equidimensional, then by \cite[Lemma A.12]{SS2017} 
$\sing(\Wphii) \subset S$.

 \section{Proof of Proposition~\ref{prop:dimfiber}: atlases for 
fibers}\label{sec:fibergen}
This section is devoted to the proof of Proposition~\ref{prop:dimfiber}. We
recall its statement below.
\begin{proposition*}[\ref{prop:dimfiber}]
 Let $V,S \subset \CC^n$ be two algebraic sets with $V$ $d$-equidi\-men\-sio\-nal and 
$S$ finite. Let $\bchi$ be an atlas of $(V,S)$.
 Let $2\leq \ifi \leq d+1$ and $\bphi=(\phi_1,\dotsc,\phi_{\ifi}) \subset 
\CC[\XX]$. For $2\leq j \leq d$, let $\balf_j = (\alf_{j,1}, \ldots, 
\alf_{j,n})\in \CC^{n}$ and 
 \begin{align*}
  \phi_1(\XX,\balf_1) = \polun(\XX) + 
  \sum_{k=1}^n \alf_{1,k}x_k
  \et
  \phi_j(\XX,\balf_j) = \sum_{k=1}^n \alf_{j,k}x_k 
 \end{align*}
 where $\polun \in \CC[\XX]$.

 There exists a non-empty Zariski open subset $\ZOfiber(\bchi,V,S, \polun,\ifi)
 \subset \CC^{\ifi n}$ such that for every $\balf = (\balf_1, \ldots, 
\balf_{\ifi}) \in
  \ZOfiber(\bchi,V,S,\polun,\ifi)$ and writing
    \[
      \bphi = \left(\phi_1(\XX, 
        \balf_1),
        \dotsc,\phi_{\ifi}(\XX, \balf_{\ifi})\right),
    \]
the following holds. Let $0 
\leq e
 \leq d$, $Q \in \CC^e$ a finite subset and $F_Q$ and $S_Q$ be as in
 Definition~\ref{def:Fatlas}. Then either $F_Q$ is empty or
 \begin{enumerate}
 \item $S_Q$ is finite;
  \item $V_Q$ is an equidimensional algebraic set of dimension $d-e$;
  \item $\Fatlas(\bchi,V,Q,S,\bphi)$ is an atlas of $(F_Q,S_Q)$ and $\sing(F_Q) 
\subset S_Q$.
 \end{enumerate}
\end{proposition*}

Let $V,S$ and $\bchi = (\chi_j)_{1\leq j \leq s}$ be as above, with
$\chi_j = (m_j,\hh_j)$ for $1\leq j \leq s$. 
Consider and integer $2\leq \ifi \leq d+1$, we show in the following that it 
suffices to take $\ZOfiber(\bchi,V,Q,S,\polun,\ifi)$ as the non-empty Zariski 
open subset $\ZOnoether(V,\polun,\ifi)$ of $\CC^{\ifi n}$ obtained by the 
application of Proposition~\ref{prop:fibergenalg} to $V,\polun$ and $\ifi$.

Let $\balf \in \ZOfiber(\bchi,V,S,\polun,\ifi)$ and $\phi = 
(\bphi_1(\XX,\balf),\dotsc,\bphi_{\ifi}(\XX,\balf))$ where for $2\leq j\leq 
\ifi$,
\begin{align*}
  \phi_1(\XX,\balf_1) = \polun(\XX) + 
  \sum_{k=1}^n \alf_{1,k}x_k
  \et
  \phi_j(\XX,\balf_j) = \sum_{k=1}^n \alf_{j,k}x_k 
 \end{align*}
For $1\leq e \leq \ifi-1$, let $Q \subset \CC^e$ be a finite set and $F_Q, S_Q$ 
as in 
Definition~\ref{def:Fatlas}. Suppose also that $F_Q$ is not empty.
We start with the following lemma, proving local statements on the fibers. It is
a direct generalization of \cite[Lemma C.1]{SS2017}.
\begin{lemma}\label{lem:chartonfiber}
 Let $1\leq j \leq s$ and $m=m_j$, $\hh=\hh_j$ and $\chi=(m,\hh)$.
 Then either $\Ocal(m)\cap F_Q$ is empty or $\chi$ is a chart of 
$(F_Q,Q,S_Q,\bphi)$, and $S_Q$ is finite.
\end{lemma}
\begin{proof}
 Remark first that since $\balf \in \ZOpen(V,\polun)$, then by 
Proposition~\ref{prop:fibergenalg}, the set
 \[
  S_Q = \big(S \cup \Wphii[e]\big)\cap \map[e][-1](Q)
 \]
 is finite, since $S$ and $Q$ are. Assume now that $\Ocal(m)\cap F_Q$ is not 
empty, and let us prove that $\chi$ is a chart of $(F_Q,Q,S_Q,\bphi)$.
 
 \begin{enumerate}[label=$\sfC_\arabic*:$]
  \item This holds by assumption.
  
  \item By property $\sfC_2$ of $\chi$, the sets $F_Q$ and 
$\mapfbr{\V(\hh)}$ 
coincide in $\Ocal(m) - S$. But since $S \subset S_Q$ in $\phibr$ then these 
sets coincide in $\Ocal(m) - S_Q$ as well.
  
  \item Since $V$ is $d$-equidimensional, then by \cite[Lemma A.12]{SS2017}, $c 
= n-d$. Hence, since $e \leq \ifi -1 \leq d$, the inequality $e+c \leq n$ holds.
  
  \item Finally, let $\yy \in \Ocal(m) \cap F_Q - S_Q$.
  Since $\yy \notin S_Q$ then $\yy \notin \Wphii[e] \cap \phibr$, but since 
$\yy 
\in \phibr$ then actually $\yy \notin \Wphii[e]$.
  Hence since $\yy \in \Ocal(m)$, then by Lemma~\ref{lem:jacrankchart}, 
$\jac_{\yy}(\hh,\map[e])$ has full rank $c+e$.
 \end{enumerate}
All the properties of charts being satisfied, we are done.
\end{proof}

We now proceed to prove Proposition~\ref{prop:dimfiber}. The first
statement is given by Lemma~\ref{lem:chartonfiber}. If $e=d$, then the 
second statement is satisfied by the last item 
Proposition~\ref{prop:fibergenalg}, since
$\Kphii[d+1]=V$.
Assume now that $e<d$. By Krull's principal ideal Theorem \cite[Theorem
B]{Ei1995} or equivalently the theorem on the dimension of fibers \cite[Theorem
1.25]{Sh1994}, all irreducible components of $F_Q$ have dimension at least
$d-e>0$.

We now prove the last statement that is that $\Fatlas(\bchi,V,Q,S,\bphi)$ is an 
atlas of
$(F_Q,Q,S_Q,\bphi)$:
\begin{enumerate}[label=$\sfA_\arabic*:$]
 \item Since $F_Q$ has positive dimension and $S_Q$ is finite, then $F_Q-S_Q$ 
is 
not empty. Since $F_Q \subset V$, then by property $\sfA_3$ of $\bchi$, there 
exists $1\leq j \leq s$ such that $\Ocal(m_j)\cap F_Q - S_Q$ is not empty.
 
 \item Let $1\leq j\leq s$ such that $\Ocal(m_j)\cap F_Q - S_Q$ is not empty, 
then by Lemma~\ref{lem:chartonfiber}, $\chi_j$ is a chart of 
$(F_Q,Q,S_Q,\bphi)$. Since the elements of $\Fatlas(\bchi,V,Q,S,\bphi)$ are 
exactly such $\chi_j$, we are done. 
 
 \item Finally let $\yy \in F_Q - S_Q$, since $\yy \in \phibr$ then $\yy \notin 
S$. Since $F_Q \subset V$, then by property $\sfA_3$ of $\bchi$, there exists 
$1\leq j \leq s$ such that $\yy \in \Ocal(m_j)$.
 In particular, $\Ocal(m_j) \cap F_Q - S_Q$ is non-empty, so that $\chi_j \in 
\Fatlas(\bchi,V,Q,S,\bphi)$.
\end{enumerate}
Hence $\Fatlas(\bchi,V,Q,S,\bphi)$ is an atlas of $(F_Q,Q,S_Q,\bphi)$. 
In particular, since $V$ is $d$-equidimensional, all the $\hh_j$'s have same 
cardinality $c = n-d$ by \cite[Lemma A.12]{SS2017}.
Hence by \cite[Lemma A.11]{SS2017}, $F_Q-S_Q$ is a non-singular 
$(d-e)$-equidimensional \LCS.
Since $F_Q$ has positive dimension and $S_Q$ is finite, we deduce that 
$F_Q$ is the Zariski closure of $F_Q-S_Q$ and then, is a 
$(d-e)$-equidimensional 
algebraic set, smooth outside $S_Q$.
This concludes the proof of Proposition~\ref{prop:dimfiber}.

\section*{Acknowledgments}
The first author was supported by the FWO grant G0F5921N. The second author is
supported by ANR grants ANR-18-CE33-0011 \textsc{Sesame}, and ANR-19-CE40-0018
\textsc{De Rerum Natura}, the joint ANR-FWF ANR-19-CE48-0015 \textsc{ECARP}
project, the joint French–Austrian project \textsc{Eagles} (ANR-22-CE91-0007 \&
FWF I6130-N) and the Grant FA8665-20-1-7029 of the EOARD-AFOSR. The third author
was supported by an NSERC Discovery Grant.

\bibliographystyle{abbrv}
\bibliography{PrSaSc2025-preprint.bib}
\end{document}